%% file: main_ICDT_2025.tex
\documentclass[a4paper,UKenglish,cleveref, autoref]{lipics-v2021}

\input{macros}

\usepackage{soul}
\usepackage{url}
\usepackage{amsmath}
\usepackage{amsthm}
\usepackage{booktabs}
\usepackage{algorithm}
\usepackage[noend]{algorithmic}
\usepackage{balance}
\usepackage{multicol}

\usepackage{subcaption}

\nolinenumbers

\title{Enumeration of Minimal Hitting Sets Parameterized by Treewidth}	

\author{Batya Kenig}{Technion, Israel Institute of Technology}{batyak@technion.ac.il}{https://orcid.org/0000-0001-6349-128X}{}
\author{Dan Shlomo Mizrahi}{Technion, Israel Institute of Technology}{mizrahid@campus.technion.ac.il}{}{}

\authorrunning{Batya Kenig and Dan Shlomo Mizrahi} 

\Copyright{Batya Kenig and Dan Shlomo Mizrahi} 



\category{} 

\relatedversion{} 

\EventEditors{Sudeepa Roy and Ahmet Kara}
\EventNoEds{2}
\EventLongTitle{28th International Conference on Database Theory (ICDT 2025)}
\EventShortTitle{ICDT 2025}
\EventAcronym{ICDT}
\EventYear{2025}
\EventDate{March 25--28, 2025}
\EventLocation{Barcelona, Spain}
\EventLogo{}
\SeriesVolume{328}
\ArticleNo{5}

\funding{This work was supported by the US-Israel Binational Science
	Foundation (NSF-BSF) Grant No. 2020751.}

\hideLIPIcs

\begin{document}
	
		\maketitle
		
\begin{abstract}
Enumerating the minimal hitting sets of a hypergraph is a problem which arises in many data management applications that include constraint mining, discovering unique column combinations, and enumerating database repairs.
Previously, Eiter et al.~\cite{DBLP:journals/siamcomp/EiterGM03} showed that 
the minimal hitting sets of an $n$-vertex hypergraph, with treewidth $w$, can be enumerated with delay $O^*(n^{w})$ (ignoring polynomial factors), with space requirements that scale with the output size. We improve this to fixed-parameter-linear delay, following an FPT preprocessing phase. The memory consumption of our algorithm is exponential with respect to the treewidth of the hypergraph.
\eat{
For hypergraphs with bounded treewidth, Eiter et al.~\shortcite{DBLP:journals/siamcomp/EiterGM03} gave an algorithm
	that runs in polynomial-delay, \eat{and whose delay is $\tilde{O}(n^{w+1})$ for an $n$-vertex hypergraph whose treewidth is $w$.} and whose space requirement scales with the output size. 
	We improve this to fixed-parameter-linear delay, and polynomial space, following an FPT preprocessing phase.
}
	
\end{abstract}

\input{Introduction}

\input{relationtoMetaTheoremOfCourcelle}
\noindent {\bf Organization.} Following preliminaries in Section~\ref{sec:Preliminaries}, we show in Section~\ref{sec:TransEnumReduction} how \textsc{trans-enum} and \textsc{cover-enum} over a hypergraph with treewidth $w$ can be reduced to \textsc{dom-enum} over a graph with treewidth $w'\leq w+1$. In Section~\ref{sec:overview} we present an overview of the algorithm.
In Section~\ref{sec:PreprocessingForEnumeration}, we describe the FPT preprocessing-phase, and in Section~\ref{sec:enuemration} prove correctness, and establish the FPL-delay guarantee of the enumeration algorithm. We conclude in Section~\ref{sec:conclusion}. Due to space constraints, we defer some of the technical details and proofs to the Appendix. In Section~\ref{sec:lowerBounds} of the Appendix  we discuss lower bounds, and other related work.

\input{prelims}
\input{ReductionFromTransEnumToDomEnum}

\input{AlgorithmOverview}
\input{preprocessing}
\input{nonRankedEnum}

\input{conclusion}

\newpage


\clearpage
\balance
\bibliography{main}

\newpage
\appendix

\input{ProofsFromSection3}
\input{AppendixOverviewDetails}

\input{preprocessing8states}

\input{AppendixNewEnumerationProofs}
\input{JustificationForNotUsingTreeAutomaton}
\input{LowerBounds}

\end{document}

%% file: macros.tex
\usepackage{amsfonts}
\usepackage{booktabs}
\usepackage{siunitx}
\usepackage{multirow}
\usepackage{xcolor}
\usepackage{graphicx}
\usepackage{listings}
\usepackage{commath}
\usepackage{epstopdf}
\usepackage{url}
\usepackage{mathtools}
\usepackage{alltt}
\usepackage{mathrsfs}
\usepackage{wrapfig}
\usepackage{tabularx}
\usepackage{balance}
\usepackage{csquotes}

\usepackage{subcaption}
\usepackage[normalem]{ulem}
\usepackage{etoolbox}
\usepackage{stmaryrd}
\usepackage[noend]{algorithmic}

\usepackage{mathtools}
\usepackage{etoolbox}
\usepackage{dsfont}
\usepackage{color}
\usepackage{colortbl}
\usepackage{multirow}
\usepackage[scaled]{helvet}

\usepackage{fullminipage}
\usepackage{thmtools}

\usepackage{graphicx,fancyvrb}
\usepackage[colorinlistoftodos]{todonotes}
\usepackage[noend]{algorithmic}
\usepackage[shortlabels]{enumitem}
\setlist{nolistsep}
\setlist[enumerate]{nosep}
\usepackage{booktabs}
\usepackage{capt-of}
\usepackage{xcolor}
\usepackage{multirow}

\usepackage{tikz}
\usepackage{circuitikz}
\usetikzlibrary{arrows,shapes,positioning,backgrounds,fit,quotes,graphs}
\usepackage{tkz-euclide}
\usepackage{forest}

\usepackage{float}
\usepackage{subfloat}

\usepackage{array}
\usepackage{arydshln}
\usepackage{todonotes}

\definecolor{mygreen}{rgb}{0,0.6,0}
\definecolor{mygray}{rgb}{0.5,0.5,0.5}
\definecolor{mymauve}{rgb}{0.58,0,0.82}

\usepackage{listings}
\lstnewenvironment{VerbatimText}[1][]{
    
    \lstset{fancyvrb=true,basicstyle=\footnotesize,captionpos=b,xleftmargin=2em,#1}
}{}

\newcommand{\ignore}[1]{}

\newcommand{\N}{\mathbb N} 
\newcommand{\D}{\mathbf D} 

\newcommand{\revone}[1]{{\color{blue} {#1}}}

\newcommand{\revth}[1]{{\color{purple} {#1}}}

\newcommand{\A}{\texttt{A}}

\newcommand{\sch}{{\mathcal{S}}}

\newcommand{\real}{{\mathbb{R}}}

\newcommand{\true}{{\tt true}}

\newcommand*{\rom}[1]{\expandafter\@slowromancap\romannumeral #1@}
\newcommand{\RNum}[1]{\uppercase\expandafter{\romannumeral #1\relax}}

\newcommand{\sel}[1]{{\sigma}}

\newcommand{\cut}[1]{}
\newcommand{\eat}[1]{}

\newcommand{\setof}[2]{\{{#1}\mid{#2}\}}        

\def\set#1{\mathord{\{#1\}}}

\def\eqdef{\mathrel{\stackrel{\textsf{\tiny def}}{=}}}

\def\A{\mathcal{A}}
\def\I{\mathcal{I}}

\def\H{\mathcal{H}}
\def\D{\mathcal{D}}
\def\e#1{\emph{#1}}
\newenvironment{citedtheorem}[1]
{\begin{theorem}{\it\e{(#1)}}\,\,}
	{\end{theorem}}
\newenvironment{citedlemma}[1]
{\begin{lemma}{\it\e{(#1)}}\,\,}
	{\end{lemma}}
\newenvironment{citeddefinition}[1]
{\begin{definition}{\it\e{(#1)}}\,\,}
	{\end{definition}}

\def\implies{\Rightarrow}

\newenvironment{repeatresult}[2]
{\vskip0.5em\par\textsc{#1} #2.\em}
{\vskip1em}

\newenvironment{repproposition}[1]{\begin{repeatresult}{Proposition}{#1}}{\end{repeatresult}}
\newenvironment{reptheorem}[1]{\begin{repeatresult}{Theorem}{#1}}{\end{repeatresult}}
\newenvironment{replemma}[1]{\begin{repeatresult}{Lemma}{#1}}{\end{repeatresult}}

\def\appendix{\par
	\section*{APPENDIX}
	\setcounter{section}{0}
	\setcounter{subsection}{0}
	\def\thesection{\Alph{section}} }

\def\eqdef{\mathrel{\stackrel{\textsf{\tiny def}}{=}}}
\def\I{\mathcal{I}}

\def\B{\mathcal{B}}
\def\e#1{\emph{#1}}

\newcommand{\algname}[1]{{\sf #1}}
\def\myrulewidth{2.80in}
\def\therule{\rule{\myrulewidth}{0.2pt}}

\def\myrulewidthwide{4in}
\def\therulewide{\rule{\myrulewidthwide}{0.2pt}}

\newenvironment{algseries}[2]
{\centering\begin{figure}[#1]\begin{center}\def\thecaption{\caption{#2}}
			\begin{tabular}{p{\myrulewidth}}\therule\end{tabular}
		}
		{\thecaption\end{center}\end{figure}}

\newenvironment{algserieswide}[2]
{\centering\begin{figure}[#1]\begin{center}\def\thecaption{\caption{#2}}
			\begin{tabular}{p{\myrulewidthwide}}\therulewide\end{tabular}\vskip0.2em}
		{\thecaption\end{center}\end{figure}}

\newenvironment{insidealg}[2]
{\normalsize
	\begin{insidecode}{#1}{#2}{Algorithm}}
	{\end{insidecode}
}

\newenvironment{insidealgwide}[2]
{\normalsize\begin{insidecodewide}{#1}{#2}{Algorithm}}
	{\end{insidecodewide}}

\newenvironment{insidecode}[3]
{
	\begin{tabular}{p{\myrulewidth}}
		\multicolumn{1}{c}{\rule{0mm}{3mm}{\bf #3} $\algname{#1}(\mbox{#2})$\vspace{-0.6em}}\\
		\therule\vskip-0.8em\therule
		\vspace{-1em}
		\begin{algorithmic}[1]}
		{\end{algorithmic}
		\vskip-0.4em\therule
\end{tabular}}

\newenvironment{insidecodewide}[3]
{
	\begin{tabular}{p{\myrulewidthwide}}
		\multicolumn{1}{c}{\rule{0mm}{3mm}{\bf #3} $\algname{#1}(\mbox{#2})$\vspace{-0.6em}}\\
		\therulewide\vskip-0.8em\therulewide
		\vspace{-1em}
		\begin{algorithmic}[1]}
		{\end{algorithmic}
		\vskip-0.3em\therulewide
\end{tabular}}

\newcommand{\T}{{\mathcal{T}}}

\newcommand{\jointreeMapFunction}{\chi}

\def\M{\mathcal{M}}

\newcommand{\edges}{\texttt{E}}
\newcommand{\Hedges}{\mathcal{E}}
\newcommand{\nodes}{\texttt{V}}
\newcommand{\parent}{\texttt{pa}}
\newcommand{\child}{\texttt{ch}}
\newcommand{\leftchild}{\texttt{leftCh}}
\newcommand{\rightchild}{\texttt{rightCh}}
\newcommand{\branch}{\texttt{br}}
\DeclareMathSymbol{\mlq}{\mathord}{operators}{``}
\DeclareMathSymbol{\mrq}{\mathord}{operators}{`'}

\newcommand{\tw}{{\tt \mathrm{tw}}}

\newcommand{\efftw}{{\tt \mathrm{eff{-}tw}}}

\newcommand{\bTheta}{\boldsymbol{\Theta}}

\def\FA{\varphi}
\newcommand{\FAs}[1]{F^{\jointreeMapFunction(#1)}}

\definecolor{mygreen}{rgb}{0,0.6,0}
\definecolor{mygray}{rgb}{0.5,0.5,0.5}
\definecolor{mymauve}{rgb}{0.58,0,0.82}
\definecolor{cadmiumgreen}{rgb}{0.0, 0.42, 0.24}

\def\gq1{{\geq}1}

\newcommand{\second}{{\tt second}}

\newenvironment{hproof}{%
	\renewcommand{\proofname}{Proof Overview}\proof}{\endproof}

\declaretheoremstyle[%
		spaceabove=-2pt,%
		spacebelow=2pt,%
		headfont=\normalfont\itshape,%
		postheadspace=1em,%
		qed=\qedsymbol%
		]{mystyle} 
\declaretheorem[name={Proof},style=mystyle,unnumbered]{prf}

%% file: Introduction.tex
\newcommand{\blue}[1]{\textcolor{blue}{#1}}

\section{Introduction}
A \e{hypergraph} $\H$ is a finite collection $\Hedges(\H)$ of sets called \e{hyperedges} over a finite set $\nodes(\H)$ of \e{nodes}. A
\e{transversal} (or \e{hitting set}) of $\H$ is a set $\sch \subseteq \nodes(\H)$ that has a non-empty intersection
with every hyperedge of $\H$. A transversal $\sch$ is minimal if no proper subset of
$\sch$ is a transversal of $\H$. The problem of enumerating the set of minimal transversals of a hypergraph (called \textsc{Trans-Enum} for short) is a fundamental combinatorial problem with a wide range of applications in Boolean algebra~\cite{DBLP:journals/dam/EiterMG08},\eat{model-based diagnosis~\cite{DBLP:journals/ai/Reiter87}} computational biology~\cite{vera-licona_ocsana_2013}, drug-discovery~\cite{DBLP:journals/bmcsb/Vazquez09}, and many more~(see surveys~\cite{DBLP:conf/jelia/EiterG02,DBLP:journals/siamdm/Gainer-DewarV17,DBLP:journals/siamcomp/EiterG95}).\eat{ )\cite{DBLP:journals/tcs/CosmadakisKP20,DBLP:journals/jgaa/KavvadiasS05,DBLP:journals/jcss/BlasiusFLMS22,DBLP:journals/siamdm/Gainer-DewarV17}).} 
\eat{
. . Enumerating the minimal transversals of a hypergraph has applications in various domains, making it a widely studied problem. It has applications in Boolean switching theory , model-based diagnosis, computational biology, and data management.

 is of interest in a wide variety of domains, and has been extensively studied in a a wide range of fields that include 
~\cite{DBLP:journals/tcs/CosmadakisKP20,DBLP:journals/jgaa/KavvadiasS05,DBLP:journals/jcss/BlasiusFLMS22,DBLP:journals/siamdm/Gainer-DewarV17}. Many important problems in various domains can be translated to \textsc{Trans-Enum}. Some examples include \e{Monotone dualization}, which receives as input a monotone Boolean formula in Conjunctive Normal Form (CNF), and converts it to Disjunctive Normal Form (DNF)~\cite{DBLP:journals/dam/EiterMG08}. Fault diagnostic procedures for complicated systems are also reduced to \textsc{Trans-Enum}~\cite{DBLP:journals/ai/Reiter87}. Algorithms for \textsc{Trans-Enum} have also been studied and applied in the field of computational biology~\cite{vera-licona_ocsana_2013}, and for the purpose of drug-discovery~\cite{DBLP:journals/bmcsb/Vazquez09}. 
Another example is listing the \e{Minimal Unsatisfiable Subsets} (MUS) of clauses in a CNF~\cite{DBLP:journals/constraints/LiffitonPMM16,DBLP:conf/ismvl/Silva10}, which are used to provide human-understandable explanations to solutions of Constraint Satisfaction Problems (CSP)~\cite{DBLP:conf/ecai/0001GCG20,DBLP:conf/ijcai/Gamba0G21}, planning decisions~\cite{DBLP:conf/aaai/VasileiouP021}, and to reason about inconsistent formulas and their repairs~\cite{DBLP:conf/ijcai/0001M20}.}

The role that \textsc{Trans-Enum} in data management applications can be traced back to the 1987 paper by Mannila and R{\"{a}}ih{\"{a}}~\cite{DBLP:conf/vldb/MannilaR87}. This stems from the fact that
finding minimal or maximal solutions (with respect to some property) is an essential task in many data management applications. In many cases, such tasks are easily translated, or even equivalent, to \textsc{Trans-Enum}.
These include, for example, the discovery of inclusion-wise minimal Unique Column Combinations (UCCs), which are natural candidates for primary keys~\cite{DBLP:journals/pvldb/BirnickBFNPS20,DBLP:journals/jcss/BlasiusFLMS22}, and play an important role in data profiling and query optimization (see surveys by Papenbrock, Naumann, and
coauthors~\cite{DBLP:series/synthesis/2018Abedjan},~\cite{DBLP:journals/vldb/KossmannPN22}). Eiter and Gottlob have proved that UCC discovery and \textsc{Trans-Enum} are, in fact, equivalent~\cite{DBLP:journals/siamcomp/EiterG95}. Other data profiling tasks that can be reduced to \textsc{trans-enum} include discovering maximal frequent itemsets~\cite{DBLP:journals/tods/GunopulosKMSTS03}, functional dependencies~\cite{DBLP:conf/vldb/MannilaR87,9835531,DBLP:journals/pvldb/0001N18}, and \e{multi-valued dependencies}~\cite{DBLP:conf/sigmod/KenigMPSS20,DBLP:conf/pods/KenigW23}.
\eat{
Algorithms for discovering functional dependencies proceed by finding inclusion-wise minimal sets of attributes that functionally determine a given attribute~\cite{DBLP:conf/vldb/MannilaR87,9835531,DBLP:journals/pvldb/0001N18}. Algorithms for discovering maximal frequent itemsets proceed by finding maximal sets of items whose frequency in the data is above some threshold~\cite{DBLP:journals/tods/GunopulosKMSTS03}. The ideas of~\cite{DBLP:journals/tods/GunopulosKMSTS03} were recently applied to the task of mining \e{multi-valued dependencies}~\cite{DBLP:conf/sigmod/KenigMPSS20}. Common to all of these tasks, and many others, is that they can be reduced to \textsc{Trans-Enum}.
}Another data-management task that directly translates to \textsc{Trans-Enum} is that of enumerating \e{database repairs}, which are maximal sets of tuples that exclude some combinations of tuples based on logical patterns (e.g., functional dependencies, denial constraints etc.)~\cite{DBLP:conf/cikm/ChomickiMS04,DBLP:series/synthesis/2011Bertossi,DBLP:journals/tcs/KimelfeldLP20,DBLP:journals/pvldb/LivshitsHIK20}. The forbidden combinations of tuples are represented as the hyperedges in a conflict hypergraph~\cite{DBLP:journals/iandc/ChomickiM05,DBLP:conf/cikm/ChomickiMS04}, and the
problem of enumerating the repairs is reduced to that of enumerating the minimal transversals
of the conflict hypergraph~\cite{DBLP:journals/pvldb/LivshitsHIK20}.  

An \e{edge cover} of a hypergraph $\H$ is a subset of hyperedges $\D\subseteq \Hedges(\H)$ such that every vertex $v\in \nodes(\H)$ is included in at least one hyperedge in $\D$ \eat{(i.e., $\nodes(\H)=\bigcup_{e\in \D}e$)}. An edge cover $\D$ is minimal if no proper subset of $\D$ is an edge cover. Minimal edge covers are used to generate \e{covers of query results}, which are succinct, lossless representations of query results~\cite{DBLP:conf/icdt/KaraO18}. Our results for \textsc{trans-enum} carry over to the problem of enumerating minimal edge-covers (\textsc{cover-enum} for short).

\eat{
It is hard to overstate the role that \textsc{Trans-Enum} plays in data management applications, which can be traced back to the 1987 paper by Mannila and R{\"{a}}ih{\"{a}} and~\cite{DBLP:conf/vldb/MannilaR87}. 
Finding minimal or maximal solutions (with respect to some property) is an essential task in many data management applications. These include, for example, the discovery of inclusion-wise minimal Unique Column Combinations (UCCs), which are natural candidates for primary keys~\cite{DBLP:journals/pvldb/BirnickBFNPS20,DBLP:journals/jcss/BlasiusFLMS22}, \blue{and play an important role in data profiling and query optimization (see surveys by Papenbrock, Naumann, and
	coauthors )}. Eiter and Gottlob have proved that the UCC discovery and \textsc{Trans-Enum} are, in fact, equivalent~\cite{DBLP:journals/siamcomp/EiterG95}.
Algorithms for discovering functional dependencies proceed by finding inclusion-wise minimal sets of attributes that functionally determine a given attribute~\cite{DBLP:conf/vldb/MannilaR87,9835531,DBLP:journals/pvldb/0001N18}, and are therefore translated to \textsc{Trans-Enum}.

Eiter and Gottlob [EG95] showed that the unique column combinations of
a database can be enumerated in output-polynomial time if and only if this is
possible for the hitting sets of a hypergraph.

 Common to all of these tasks, and many others, is that they can be reduced to the \e{transversal hypergraph} or \e{hitting set} generation problem that has been extensively studied~\cite{DBLP:journals/tcs/CosmadakisKP20,DBLP:journals/jgaa/KavvadiasS05,DBLP:journals/jcss/BlasiusFLMS22,DBLP:journals/siamdm/Gainer-DewarV17}. 
}

\eat{
Finding minimal or maximal solutions (with respect to some property) is an essential task in many areas of Artificial Intelligence and Logic~\cite{DBLP:journals/amai/Ben-EliyahuD96}.
Most famous is the problem of \textsc{Dualization} that receives as input the prime CNF of a monotone Boolean function, and returns all of its prime implicants~\cite{DBLP:journals/dam/EiterMG08}. 
Another example is that of listing the \e{Minimal Unsatisfiable Subsets} (MUS) of clauses in a CNF~\cite{DBLP:journals/constraints/LiffitonPMM16,DBLP:conf/ismvl/Silva10}, which are used to provide human-understandable explanations to solutions of Constraint Satisfaction Problems (CSP)~\cite{DBLP:conf/ecai/0001GCG20,DBLP:conf/ijcai/Gamba0G21}, planning decisions~\cite{DBLP:conf/aaai/VasileiouP021}, and to reason about inconsistent formulas and their repairs~\cite{DBLP:conf/ijcai/0001M20}.
These problems, and many more, can be reduced to the \e{transversal hypergraph} or \e{hitting-set} generation problem that has been extensively studied~\cite{DBLP:journals/tcs/CosmadakisKP20,DBLP:journals/jgaa/KavvadiasS05,DBLP:journals/jcss/BlasiusFLMS22,DBLP:journals/siamdm/Gainer-DewarV17}.  
}

\eat{
A \e{hypergraph} $\H$ is a finite collection $\Hedges(\H)$ of sets called \e{hyperedges} over a finite set $\nodes(\H)$ of \e{nodes}. A
\e{transversal} (or \e{hitting set}) of $\H$ is a set $\sch \subseteq \nodes(\H)$ that has a non-empty intersection
with every hyperedge of $\H$. A transversal $\sch$ is minimal if no proper subset of
$\sch$ is a hitting set of $\H$. The problem of enumerating the set of minimal transversals of a hypergraph is called \textsc{Trans-Enum} for short.
}

A hypergraph can have exponentially many minimal
hitting sets, ruling out any polynomial-time algorithm. Hence, the yardstick used to measure the efficiency of an algorithm for \textsc{Trans-Enum} is the \e{delay} between the output of successive solutions. An enumeration algorithm runs in \e{polynomial delay} if the delay is polynomial in the size of the input.
Weaker notions are \e{incremental polynomial delay} if the delay between the $i$th and the $(i{+}1)$st solutions is bounded by a polynomial of the size of the input plus $i$, and \e{polynomial total time} if the total computation time is polynomial in the sizes of both the input and output.\eat{
\revone{It is well-known that the \eat{In a seminal paper, Fredman and Khachiyan~\cite{DBLP:journals/jal/FredmanK96} proved that} minimal hitting sets can be enumerated in \e{quasi-polynomial total time}~\cite{DBLP:journals/jal/FredmanK96}}.} While there exists a \e{quasi-polynomial total time} algorithm~\cite{DBLP:journals/jal/FredmanK96}, whether \textsc{Trans-Enum} has a polynomial total time algorithm is a major open question~\cite{DBLP:journals/dam/EiterMG08}.
In the absence of a tractable algorithm for general inputs, previous work has focused on special classes of hypergraphs. In particular, \textsc{Trans-Enum} admits a polynomial total time algorithm when restricted to \e{acyclic} hypergraphs\eat{and hypergraphs of bounded degeneracy~\cite{DBLP:journals/siamcomp/EiterGM03}}, and hypergraphs with bounded hyperedge size~\cite{DBLP:journals/ppl/BorosEGK00}.\eat{, and those whose hitting sets have bounded cardinality~\cite{DBLP:journals/jcss/BlasiusFLMS22}.}

In this paper, we consider the \textsc{trans-enum} problem parameterized by the \e{treewidth} of the input hypergraph. The \e{treewidth} of a graph is an important graph complexity measure that occurs as a parameter in many \e{Fixed Parameter Tractable} (FPT) algorithms~\cite{DBLP:series/txtcs/FlumG06}. An FPT algorithm for a parameterized problem solves the problem in time $O(f(k)n^c)$, where $n$ is the size of the input, $f(k)$ is a computable function that depends only on the parameter $k$, and $c$ is a fixed constant.
The treewidth of a hypergraph is defined to be the treewidth of its associated node-hyperedge \e{incidence graph}~\cite{DBLP:journals/siamcomp/EiterGM03} (formal definitions in Section~\ref{sec:Preliminaries}). Our main result is that for a hypergraph $\H$, with $n$ vertices, $m$ hyperedges, and treewidth $w$, we can enumerate the minimal hitting sets of $\H$ with fixed-parameter-linear-delay $O((n+m)w)$, following an FPT preprocessing-phase that takes time $O((n{+}m)k5^k)$, where $w\leq k\leq 2w$. The memory requirement of our algorithm is in $O((n{+}m)5^k)$.

Previously, Eiter et al.~\cite{DBLP:journals/siamcomp/EiterGM03} showed that 
the minimal transversals of an $n$-vertex hypergraph $\H$ with treewidth $w$ can be enumerated with delay $O(||\H||n^{w+1})$, where $||\H||$ is the size of the hypergraph (formal definition in Section~\ref{sec:Preliminaries}), and $w$ its treewidth. \eat{Their algorithm takes advantage of the fact that the \e{degeneracy}{\footnote{A graph $G$ is $k$-degenerate if there is an ordering $v_1,\dots,v_n$ of its vertices such that for every $i\in [1,n]$, vertex $v_i$ has at most $k$ neighbors in $\set{v_1,\dots,v_{i-1}}$.}of a hypergraph is less than or equal to its treewidth.}}\eat{However, bounded treewidth implies bounded degeneracy, and is thus a stronger notion.}When parameterized by the treewidth of the input hypergraph, our result provides a significant improvement: the minimal transversals can be enumerated in fixed-parameter-linear delay $O(||\H||\cdot w)$, following an FPT preprocessing phase.\eat{, thus taking full advantage of the hypergraph's bounded-treewidth property.} Ours also improves on the memory consumption which, for the algorithm of Eiter et al.~\cite{DBLP:journals/siamcomp/EiterGM03}, scales with the output size, while ours is polynomial for bounded-treewidth hypergraphs.
\paragraph*{Overview of approach, and new techniques.}
We build on the well known relationship between the problem of enumerating minimal hypergraph transversals (\textsc{trans-enum}) and  enumerating minimal dominating sets in graphs. A subset of vertices $S \subseteq \nodes(G)$ is a \e{dominating set} of a graph $G$ if every vertex in $\nodes(G){\setminus} S$ has a neighbor in $S$. A dominating set is minimal if no proper subset of $S$ is a dominating set of $G$. We denote by $\textsc{Dom-Enum}$ the problem of enumerating the minimal dominating sets of $G$. There is a tight connection between \textsc{Dom-Enum} and \textsc{trans-enum}; Kant\'{e} et al.~\cite{DBLP:journals/siamdm/KanteLMN14} showed that \textsc{Dom-Enum} can be solved in polynomial-total-time if and only if \textsc{trans-enum} can as well. 
Building on this result, we show in Section~\ref{sec:TransEnumReduction} that \textsc{trans-enum} over a hypergraph with treewidth $w$ can be reduced to \textsc{Dom-Enum} over a graph with treewidth $w'\in \set{w,w+1}$. Therefore, we focus on an algorithm for enumerating minimal dominating sets, where the delay is parameterized by the treewidth. In that vein, we mention the recent work by Rote~\cite{DBLP:journals/corr/abs-1903-04517}, that presented an algorithm for enumerating the minimal dominating sets of trees.
\eat{
A known characterization of minimal dominating sets is that every vertex $v$ in a minimal dominating set $S$ of a graph $G$ has a \e{private neighbor}; a vertex $u \in \nodes(G)\setminus S$ such that $u$ is a neighbor of $v$, but of no other vertex in $S$. Therefore, we can view every minimal dominating set $S$ as partitioning the vertex set of $G$ into three sets: $S$, vertices outside $S$ that have a single neighbor in $S$ (i.e., private neighbors), and  vertices outside $S$ that have at least two neighbors in $S$. 
}

The enumeration algorithm has two phases: preprocessing and enumeration. In preprocessing, we construct a data structure that allows us to efficiently determine whether graph $G$ has a minimal dominating set $S$ meeting specified constraints. During enumeration, the algorithm queries this structure repeatedly, backtracking on a negative answer. This method of enumeration is known as \e{Backtrack Search}~\cite{DBLP:journals/networks/ReadT75}. Our algorithm leverages the treewidth parameter $w$ in two ways: constructing the data structure in fixed-parameter-tractable time and space dependent on $w$, and querying it such that each query answers in time $O(w)$. Overall, for an $n$-vertex graph, this yields a delay of $O(nw)$.
\eat{
The enumeration algorithm has two phases: preprocessing and enumeration. In the preprocessing phase, we construct a data structure that allows us to efficiently determine whether the graph $G$ has a minimal dominating set $S$ that meets certain constraints. During enumeration, the algorithm repeatedly queries this structure, and backtracks if the answer is negative. This method of enumeration is known as \e{Backtrack Search}~\cite{DBLP:journals/networks/ReadT75}. 
The idea behind our algorithm is to take advantage of the treewidth parameter $w$ on two levels: first, to build this data structure in time and space whose dependence on the treewidth is fixed-parameter-tractable. Second, the enumeration algorithm queries this structure in a way that ensures that the time to answer each query is in $O(w)$. Overall, for an $n$-vertex-graph, this leads to a delay of $O(nw)$. 
}
Central to our approach is a novel data structure based on the \e{disjoint branch tree decomposition}. 
\eat{Intuitively, this is a tree decomposition with disjoint branches~\cite{DBLP:journals/ipl/Duris12}.}Graphs with such a tree decomposition were first characterized by Gavril, and are called \e{rooted directed path graphs}~\cite{DBLP:journals/dm/Gavril75}. The class of hypergraphs that have a disjoint branch tree decomposition strictly include the class of \e{$\gamma$-acyclic} hypergraphs, and are strictly included in the class of \e{$\beta$-acyclic} hypergraphs~\cite{DBLP:journals/ipl/Duris12}. Hypergraphs that have a disjoint-branch tree decomposition have proved beneficial for both schema design in relational databases~\cite{DBLP:journals/jacm/Fagin83,DBLP:journals/ipl/Duris12}, and query processing in probabilistic databases~\cite{DBLP:conf/pods/BeameBGS15,DBLP:conf/sum/KenigGS13,DBLP:conf/sum/KenigG15,DBLP:conf/pods/KenigKPS17}. To our knowledge, ours is the first to make use of this structure for enumeration. 
\eat{
\revth{Dynamic programming over tree decompositions proceeds by breaking down a complex problem into smaller subproblems, each associated with a node (or bag) of the tree decomposition. These subproblems are solved independently and then combined according to the tree structure, typically in a bottom-up manner. The primary challenge, therefore, lies in merging the solutions of multiple subtrees when processing each node, which can be complex and computationally intensive. \e{Disjoint-branch tree decompositions} simplify this process by ensuring that the children of every node in the tree decomposition have pairwise disjoint sets of vertices.
This structure offers distinct advantages in handling problems with specific neighborhood constraints on vertices, such as the $\textsc{Dom-Enum}$ problem. For example, if a vertex $v$ in a bag $u$ of the tree decomposition is configured to have at least two neighbors in the dominating set (i.e., the solution), then enforcing this constraint becomes much easier under the (disjoint-branch) assumption that all of $v$'s neighbors belong to a distinct subtree rooted at one of $u$'s children. For precise details on how we apply this property in the context of enumerating minimal dominating sets, see Section~\ref{sec:convertToDBJT} and Example~\ref{eq:nonDBJT}.} 
}
Dynamic programming over tree decompositions breaks down complex problems into smaller subproblems associated with the nodes (or bags) of the tree decomposition. These subproblems are solved independently and combined in a bottom-up manner according to the tree structure. The challenge lies in merging solutions of multiple subtrees at each node, which can be complex and computationally intensive. \e{Disjoint-branch tree decompositions} simplify this by ensuring each node's children have disjoint vertex sets, easing the enforcement of specific constraints placed on the vertices.
This is especially important in the case of $\textsc{Dom-Enum}$ where the constraints imposed on the vertices have the form of ``vertex $v$ has at least two neighbors in the solution set''. Intuitively, it is easier and more efficient to ensure such constraints hold when the structure partitions $v$'s neighbors into disjoint sets.\eat{
, such as in the $\textsc{Dom-Enum}$ problem. For instance, if vertex $v$ in bag $u$ has at least two distinct neighbors in the dominating set, this constraint is simplified assuming every neighbor of $v$ belongs to a distinct subtree rooted at $u$'s children.} For the detailed application of this property for $\textsc{Dom-Enum}$ see Section~\ref{sec:convertToDBJT} and Example~\ref{eq:nonDBJT}.

Our enumeration algorithm does not assume that the (incidence graph) of the input hypergraph has a disjoint branch tree decomposition~\cite{DBLP:journals/ipl/Duris12}. \eat{In particular, we do not assume that the input hypergraph belongs to the class of hypergraphs that have a disjoint branch tree decomposition~\cite{DBLP:journals/ipl/Duris12}.}
Therefore, as part of the preprocessing phase, we construct a disjoint branch tree decomposition directly from the tree decomposition of the input graph. This new structure contains ``vertices'' that are not in the original input graph, and whose sole purpose is to facilitate the enumeration. In Section~\ref{sec:convertToDBJT} we show that
given a tree decomposition (that is not disjoint-branch) whose width is $w$, we can construct the required disjoint-branch data-structure in time $O(nw^2)$, with the guarantee that the width of the resulting structure is, effectively, at most $2w$. The main results of this paper are the following.

\eat{
The set of minimal dominating sets of a graph $G$ are subsets of its vertices. We can think of these subsets of vertices as corresponding to the leaves of a binary tree where each level $i$ corresponds to the choice of selection of vertex $v_i$ in the solution. All the solutions in the subtree rooted at $v_i$'s left child exclude $v_i$, while all the solutions in the subtree rooted at $v_i$'s right child include $v_i$. Traversing this tree depth-first allows the enumeration of all the solutions. In full generality, the delay of this algorithm can be exponential. To obtain polynomial delay, no branch should be searched in vain, that is, any explored internal node must be the prefix of at least one final solution. This method of enumeration is known as \e{Backtrack Search}.

Going through such a graph by a depth-first search allows the enumeration of all the solutions.
In full generality, the delay of this algorithm can be exponential. In order to obtain a polynomial delay, it
is important that no branch is searched in vain, that is, any explored internal node must be the prefix of
at least one final solution. T

Such a tree has an exponential number of leaves, a large portion of which are not minimal dominating sets. To use this method for the enumeration 
}
\eat{
\begin{theorem}
	Let $G$ be a graph with $n$ vertices and treewidth $w$, let $\sigma,\rho \subseteq \N$ be finite or cofinite sets, comprised of a continuous range of integers, and let $s$ be the fixed number of states required to
	represent partial solutions of the specific $[\sigma,\rho]$-domination problem.
	Following a preprocessing phase that takes time $O(nws^{2w})$, the minimal $[\sigma,\rho]$-dominating sets can be enumerated with delay $O(nw)$ and total space $O(ns^{2w})$. Ranked enumeration can be performed with delay $O(n^2)$ and total space $O(2^n)$. \eat{For Top-$K$ algorithms that return the top-$K$ solutions in ranked order, the space required is $O(K)+O(ns^{2w})$.}
\end{theorem}
}
\begin{theorem}
	\label{thm:DomEnumResult}
	Let $G$ be an $n$-vertex graph whose treewidth is $w$.\eat{
	If $G$ is a rooted directed path graph (Definition~\ref{def:disjointBranchTD}) we let $k=w+1$; otherwise, $k=2w$.} Following a preprocessing phase that takes time $O(nw8^{2w})$, the minimal dominating sets of $G$ can be enumerated with delay $O(nw)$ and total space $O(n8^{2w})$.
\end{theorem}
The translation from the \textsc{Trans-Enum} (\textsc{cover-enum}) problem in the hypergraph $\H$ to the \textsc{Dom-Enum} problem in the bipartite incidence-graph of $\H$ (see Section~\ref{sec:TransEnumReduction}) imposes certain restrictions on the set of vertices that can belong to the solution set (i.e., the minimal transversal or minimal edge-cover). As a corollary of Theorem~\ref{thm:DomEnumResult}, we show the following (in Section~\ref{sec:encodingTransEnum}).
\begin{corollary}
	\label{corr:TansToDomEnum}
	Let $\H$ be a hypergraph with $n$ vertices, $m$ hyperedges, and whose treewidth is $w$. Following a preprocessing phase that takes time $O((m+n)w5^{2w})$, the minimal transversals (minimal edge-covers) of $\H$ can be enumerated with delay $O((m+n)w)$ and total space $O((m+n)5^{2w})$.
\end{corollary}

\eat{
\noindent\textbf{Other related work.} Courcelle’s meta-theorem~\cite{DBLP:books/el/leeuwen90/Courcelle90,DBLP:journals/jal/ArnborgLS91} states that every decision and optimization graph problem definable in Monadic Second Order Logic (MSO) is FPT
(and even linear in the input size) when parameterized by the treewidth of the input
structure. That is, can be solved in time $O(n\cdot f(||\phi||,w))$ where $\phi$ is a formula in MSO of size $||\phi||$, and $G$ is an $n$-vertex graph with treewidth $w$.\eat{Since the minimal $[\sigma,\rho]$-domination problem is expressible in monadic second-order logic, the meta-theorem applies.}\eat{Arnborg et al.~\cite{ARNBORG1991308} later extended this result to counting problems, and Flum et al. [11] to enumeration problems. Furthermore,} Bagan~\cite{DBLP:conf/csl/Bagan06}, Courcelle~\cite{DBLP:journals/dam/Courcelle09}, and Amarilli et al.~\cite{DBLP:conf/icalp/AmarilliBJM17} extend this result to enumeration problems, and showed that the delay between successive solutions to problems defined in MSO logic is fixed-parameter linear, when parameterized by the treewidth of the input structure.
The problem with translating Courcelle's meta-theorem (and its extensions) to algorithms, is that the constants hidden in the $O$-notation for the running-time bound cannot be bounded by an elementary function\footnote{A function is elementary if it can be written in a ``closed'' form over the complex numbers, the elementary functors $-,+,\times$, and a constant number of exponentials and logarithms~\cite{DBLP:journals/csr/LangerRRS14}.}. In fact, Frick and Grohe~\cite{FRICK20043} prove that $f$ is non-elementary even for trees (i.e., when $w=1$).
Langer et al.~\cite{DBLP:journals/csr/LangerRRS14} present concrete experiments where this problem surfaces.
To summarize, it was already stated in~\cite{Grohe99,ParameterizedAlgs} that while approaches based on Courcelle's Theorem are very useful for classifying \e{which} problems are FPT, ``designing efficient dynamic programming algorithms on tree-decompositions requires constructing the algorithm explicitly''~\cite{ParameterizedAlgs}. Similar statements can be found in the text-book by Niedermeier~\cite{DBLP:books/ox/Niedermeier06}, and in other papers~\cite{DBLP:journals/tocl/GottlobPW10,Pichler2010,DBLP:journals/csr/LangerRRS14}. 
}

\eat{
, and the non-ranked enumeration algorithm in \cref{sec:nonRankedEnumeration}. In \cref{sec:rankedEnum}, we briefly describe the ranked enumeration algorithm, which builds on the ideas presented in the paper, but is more involved. Due to space constraints, we defer the details, pseudo-code, and proofs of this algorithm to \cref{app:rankedEnum}.
}

\eat{
It has been shown that counting maximal independent sets is FPT for bounded-tw graphs. 
Van Rooij et al.~\cite{DBLP:journals/corr/abs-1806-01667,DBLP:conf/esa/RooijBR09} showed an algorithm for counting dominating sets that runs in time $O(n^33^w)$ (Theorem 3.3). This algorithm effectively computes the number of dominating sets of every size $k\in[1,n]$. Alber and Niedermeier presented an algorithm for counting dominating sets that runs in time $O(n4^w)$ ~\cite{JochenNiedermeier2002}. 
The counting algorithm proceeds as follows. Each \e{bag} in the \e{tree decomposition} holds a set of configurations associated with its vertices. The algorithm computes the number of MISs associated with these configurations by dynamic programming, or message passing, that proceeds from the leaves towards the root of the tree decomposition. The counts associated with the root configurations are then combined to get the overall count of the MISs in the graph. 
}

\eat{

A transversal, or hitting set, of a hypergraph $\H(\nodes(\H),\mathcal{E}(\H))$  is a set $A \subseteq \nodes(\H)$ such that for every $e \in \mathcal{E}(\H)$ it holds that $A\cap e \neq \emptyset$. A \e{minimal transversal} does not contain any other transversal as a subset.\eat{An algorithm for enumerating the minimal transversals of a hypergraph $\H$ runs in polynomial delay if it outputs the minimal transversals of $\H$ without repetitions and the delay between the two outputs is in $O(p(||\H||))$ where $p$ is a polynomial function, and $||\H||=|\nodes(\H)|+\sum_{e \in \mathcal{E}(\H)}|e|$~\cite{DBLP:conf/fct/KanteLMN11}.}
The problem of enumerating the minimal transversals of a hypergraph is well studied~\cite{DBLP:conf/fct/KanteLMN11,DBLP:journals/siamcomp/EiterGM03}, but the existence of a polynomial-total-time algorithm remains an open question.\eat{
The complexity of enumerating minimal transversals remains open to this day. The best known algorithm, by Fredman and
Khachiyan~\cite{DBLP:journals/jal/FredmanK96}, runs in time $O(M^{\log M})$ where $M$ is the size of the input $||\H||$ plus output.}
The problem of enumerating the minimal dominating sets of a graph can be reduced to that of enumerating the minimal transversals of a hypergraph~\cite{DBLP:conf/fct/KanteLMN11}. For the latter problem, Eiter et al.~\cite{DBLP:journals/siamcomp/EiterGM03} presented an enumeration algorithm whose delay is $O(kn^{k+1})$, where $n$ is the number of vertices in the hypergraph $\H$, and
$k$ is the degeneracy of the \e{incidence-graph} of $\H$.\eat{
A graph $G$ is $k$-degenerate (has degeneracy $k$) if every subgraph of $G$ has a vertex of degree at most $k$.  It is easily seen that a graph of bounded-tw $w$ is $w$-degenerate (i.e., since it has an \e{elimination order} where every vertex removed has at most $w$ neighbors).}
Since every bounded-tw graph has bounded degeneracy, the result of Eiter et al.~\cite{DBLP:journals/siamcomp/EiterGM03} can be translated to the task of enumerating the minimal dominating sets of a bounded-tw, $n$-vertex graph, with the same delay $O(wn^{w+1})$~\cite{DBLP:conf/fct/KanteLMN11} (see Appendix~\ref{sec:introApp} for the details regarding this connection). For bounded-tw graphs, our result provides a significant improvement: we show that minimal dominating sets can be enumerated with linear delay $O(nw)$ following an FPT preprocessing phase.
Furthermore, our approach generalizes to $[\sigma,\rho]$-dominating sets and allows to enumerate the minimal solutions in ranked order. 
}
\eat{
 
Van Rooij et al.~\cite{DBLP:conf/birthday/Rooij20,DBLP:journals/corr/abs-1806-01667} present dynamic-programming algorithms over tree decompositions that solve the existence, optimization, and counting variant of the $[\sigma,\rho]$-domination problem over bounded-treewidth graphs. They do not consider the minimal variant of these problems. The algorithms of Van Rooij et al.~\cite{DBLP:conf/birthday/Rooij20,DBLP:journals/corr/abs-1806-01667} achieve significant improvement in runtime by incorporating algebraic transformations in the dynamic programming algorithm. In this work, we build on the encoding strategy used in~\cite{DBLP:conf/birthday/Rooij20} to address the problem of enumerating minimal $[\sigma,\rho]$-dominating sets in both ranked and arbitrary order. In~\cref{sec:CountingApp} we show how a variation of our approach can be used for counting minimal $[\sigma,\rho]$-dominating sets.
}

%% file: relationtoMetaTheoremOfCourcelle.tex
\noindent {\bf Relation to the meta-theorem of Courcelle~\cite{DBLP:books/el/leeuwen90/Courcelle90}.}
The seminal meta-theorem by Courcelle~\cite{DBLP:books/el/leeuwen90/Courcelle90,DBLP:journals/jal/ArnborgLS91} states that every decision and optimization graph problem definable in Monadic Second Order Logic (MSO) is FPT
(and even linear in the input size) when parameterized by the treewidth of the input
structure. MSO is the extension of First Order logic that allows quantification (i.e., $\forall,\exists$) over sets. Bagan~\cite{DBLP:conf/csl/Bagan06}, Courcelle~\cite{DBLP:journals/dam/Courcelle09}, and Amarilli et al.~\cite{DBLP:conf/icalp/AmarilliBJM17} extended this result to enumeration problems, and showed that the delay between successive solutions to problems defined in MSO logic is fixed-parameter linear, when parameterized by the treewidth of the input structure.

Courcelle-based algorithms for model-checking, counting, and enumerating the assignments satisfying an MSO formula $\varphi$ have the following form. Given a TD $(\T,\jointreeMapFunction)$ of a graph $G$, it is first converted to a labeled binary tree $T'$ over a fixed alphabet that depends on the width of the TD. 
The binary tree $T'$ has the property that $G\models \varphi$ if and only if $T' \models \varphi'$, where $\varphi'$ is derived from $\varphi$ (e.g., by the algorithm in~\cite{DBLP:journals/jacm/FlumFG02}). Then, the formula $\varphi'$ is converted to a finite, deterministic tree automaton (FTA) using standard techniques~\cite{DBLP:journals/mst/ThatcherW68,DBLP:journals/jcss/Doner70}. The generated FTA accepts the binary tree $T'$ if and only if $T'\models \varphi'$. Since simulating the deterministic FTA on $T'$ takes linear time, this proves that Courcelle-based algorithms run in time $O(n\cdot f(||\varphi||, w))$, where $n=|\nodes(G)|$, $w$ is the width of the TD $(\T,\jointreeMapFunction)$, and $||\varphi||$ is the length of $\varphi$.

The main problem with Courcelle-based algorithms is that the function $f(||\varphi||, w)$ is an iterated exponential of height $O(||\varphi||)$. The problem stems from the fact that every quantifier alternation in the MSO formula (i.e., $\exists \forall$ and $\forall \exists$) induces a power-set construction during the generation of the FTA. Therefore, an MSO formula with $k$ quantifier alternations, will result in a deterministic FTA whose size is a $k$-times iterated exponential. Various lower-bound results have shown that this blowup is unavoidable~\cite{DBLP:conf/dagstuhl/Reinhardt01} even for model-checking over trees. According to Gottlob, Pichler and Wei~\cite{DBLP:journals/ai/GottlobPW10}, the problem is even more severe on bounded treewidth structures because the original (possibly simple) MSO-formula $\varphi$ is first transformed
into an equivalent MSO-formula $\varphi'$ over trees, which is 
more complex than the original formula $\varphi$, and in general will have additional quantifier alternations.

MSO-logic is highly expressive, and can express complex graph properties. In particular, the condition that
$X\subseteq \nodes(G)$ is a minimal dominating set in an undirected graph $G$ can be specified using MSO logic. You can see the exact formulation (based on Proposition~\ref{prop:DSMin}) in Section~\ref{sec:courcelle}.
Specifically, the MSO formula for a minimal dominating set has two quantifier alternations. By the previous, the size of an FTA representing this formula will be at least doubly exponential. Previous analysis and experiments that attempted to craft the FTA directly (i.e., without going through the generic MSO-to-FTA conversion) have shown that this blowup materializes even for very 
simple MSO formulas~\cite{DBLP:journals/csr/LangerRRS14,DBLP:journals/entcs/KneisL09,DBLP:journals/ai/GottlobPW10,WeyerThesis}. Moreover, the result by Frick and Grohe~\cite{FRICK20043} (see Theorem~\ref{thm:FrickAndGrohe} in Section~\ref{sec:courcelle}) essentially proves that no general algorithm can perform better than the FTA-based approach. 

Despite evidence pointing to the contrary, it may be the case that with some considerable effort our algorithm can be formulated using an FTA with comparable runtime guarantees. The only instance we are aware of where a generic approach achieves running times comparable to those of a specialized algorithm is the algorithm of Gottlob, Pichler, and Wei~\cite{DBLP:journals/tocl/GottlobPW10} that translates an MSO formula to a \e{monadic Datalog} program (not an FTA). However, there too, the monadic Datalog programs were crafted manually, and used to directly encode the dynamic programming over the TD~\cite{DBLP:journals/tocl/GottlobPW10,DBLP:journals/csr/LangerRRS14}.

Because of these limitations, practical algorithms over bounded-treewidth structures often avoid the MSO-to-FTA conversion and directly encode dynamic programming strategies on the tree decomposition~\cite{DBLP:journals/ai/GottlobPW10,DBLP:journals/siamcomp/LokshtanovMS18}. These specialized algorithms are not only more efficient, but are also transparent, allowing for exact runtime analysis, and enable proving optimality guarantees by optimizing the dependence on treewidth~\cite{DBLP:journals/siamcomp/LokshtanovMS18,DBLP:conf/soda/BasteST20,DBLP:journals/tcs/DubloisLP22}. See Section~\ref{sec:courcelle} for further details.

\eat{

***********************

The algorithm derived from Courcelle's meta-theorem translates the MSO formula to a finite tree
automaton (FTA) that acts on the tree decomposition of the graph.

While the generality and theoretical guarantees of Courcelle-based approaches are appealing, the size of the resulting FTA is not bounded by an elementary function of the MSO formula. Frick and Grohe~\cite{FRICK20043} prove that unless $P=NP$, the complexity of testing MSO properties on trees is 
not bounded by any elementary function of the size of the MSO formula, making algorithms based on Courcelle's Theorem generally impractical~\cite{DBLP:conf/esa/Lampis10}. 
The problem stems from the state-explosion of the FTA generated by the procedure of converting the MSO formula to an FTA, and has been observed in practice even for simple MSO formulae~\cite{DBLP:journals/tocl/GottlobPW10}.\eat{
	In fact, Frick and Grohe~\cite{FRICK20043} prove that this tower of exponentials is unavoidable unless $P = NP$,
	and holds even for trees.} 
The impracticality of algorithms based on Courcelle's meta-theorem have been discussed in textbooks in parameterized algorithms~\cite{ParameterizedAlgs,Grohe99}, stating that approaches based on Courcelle’s Theorem are restricted to classifying \e{which} problems are FPT, and cannot be used for designing efficient algorithms. Similar statements
can be found in the text-book by Niedermeier~\cite{DBLP:books/ox/Niedermeier06}, and in other papers~\cite{DBLP:journals/tocl/GottlobPW10,DBLP:journals/csr/LangerRRS14,Pichler2010}. Importantly, optimizing the dependence on treewidth (i.e., obtaining a small-growing function $f(w)$) is an active area of research in parameterized complexity~\cite{DBLP:conf/icalp/BodlaenderCKN13,DBLP:conf/soda/BasteST20,DBLP:journals/tcs/DubloisLP22,DBLP:conf/focs/CyganNPPRW11} (see also Section~\ref{sec:lowerBounds} of the Appendix). 

To summarize, previous work has established that Courcelle's meta-theorem is an important theoretical tool for classifying which problems are FPT when parameterized by treewidth. 
In particular, our algorithm for enumerating minimal hitting sets in fixed-parameter-linear delay, following an FPT-preprocessing phase, cannot be derived from known meta-theorems, and is provably superior in both runtime and space complexity to current algorithms.
}

%% file: prelims.tex
\section{Preliminaries and Notation}
\label{sec:Preliminaries}
Let $G$ be an undirected graph with vertices $\nodes(G)$ and edges $\edges(G)$, where $n=|\nodes(G)|$, and $m=|\edges(G)|$. We assume that $G$ does not have self-loops.
Let $v\in \nodes(G)$. We let $N_G(v)\eqdef\set{u\in\nodes(G): (u,v)\in \edges(G)}$ denote the neighborhood of $v$, dropping the subscript $G$ when it is clear from the context. We denote by $G[T]$ the subgraph of $G$ induced by a subset of vertices $T\subseteq \nodes(G)$. Formally, $\nodes(G[T])=T$, and $\edges(G[T])=\{(u,v)\in \edges(G):$$ \set{u,v}\subseteq T\}$. 

Let $U\subseteq \nodes(G)$, and let $F$ be a finite set of labels. We denote by $\FA:U\rightarrow F$ an assignment of labels to the vertices $U$. For example, let $U=\set{v_1,v_2}$, $F=\set{\alpha,\beta}$, and $\FA$ be an assignment where $\FA(v_1)=\alpha, \FA(v_2)=\alpha$. If $t\in \nodes(G){\setminus}U$, we denote by $\FA'\eqdef (\FA,t\gets \beta)$ the assignment that results from $\FA$ by extending it with the assignment of the label $\beta$ to vertex $t$. That is, $\FA'(v)=\FA(v)$ for every $v\in U$, and $\FA'(t)=\beta$. Finally, for $U'\subseteq U$, we denote by $\FA[U']: U'\rightarrow F$ the assignment where $\FA[U'](w)=\FA(w)$ for all $w\in U'$.
\begin{definition}
	\label{def:DS}
A subset of vertices $D\subseteq \nodes(G)$ is a \e{dominating set} of $G$ if every vertex  $v\in \nodes(G){\setminus} D$ has a neighbor in $D$. A dominating set $D$ is \e{minimal} if $D'$ is no longer a dominating set for every $D'\subset D$. We denote by $\sch(G)$ the minimal dominating sets of $G$.
\end{definition}
\eat{
\def\propDSMin{
		A dominating set $D\subseteq \nodes(G)$ is minimal if and only if for every $u\in D$, one of the following holds: (i) $N(u)\cap D=\emptyset$, or (ii) there exists a vertex $v\in N(u){\setminus} D$, such that $N(v)\cap D=\set{u}$.
}
\begin{proposition}
	\label{prop:DSMin}
	\propDSMin
\end{proposition}
}
\eat{
\begin{proof}
	Let $D$ be a minimal dominating set, and let $u\in D$. Since $D'\eqdef D{\setminus}\set{u}$ is no longer dominating, there is a vertex $v\in \nodes(G){\setminus} D'$ such that $N(v)\cap D'=\emptyset$. If $v=u$, then $N(u)\cap D'= N(u)\cap D=\emptyset$. If $v\neq u$, then $v\in \nodes(G){\setminus} D$. Since $N(v)\cap D'=\emptyset$, and $N(v)\cap D\neq \emptyset$, we conclude that $N(v)\cap D=\set{u}$.
	
	Now, let $D\subseteq \nodes(G)$ be a dominating set for which the conditions of the proposition hold, and let $D'\eqdef D{\setminus} \set{u}$ for a vertex $u\in D$. If $D'$ is a dominating set, then $N(u)\cap D'=N(u)\cap D\neq \emptyset$. Also, $|N(v)\cap D'|\geq 1$ for every $v\in N(u){\setminus}D'$. In particular, $|N(v)\cap D|\geq 2$ for every $v\in N(u){\setminus}D$. But this brings us to a contradiction.  
\end{proof}
}
\eat{
\begin{citeddefinition}{\cite{DBLP:journals/dam/MurakamiU14}}\textsc{(Private Neighbor)}
	\label{def:pn}
Let $D\subseteq \nodes(G)$.
A vertex $v\in \nodes(G){
\setminus} D$ is a \e{private neighbor} of $D$ if there exists a vertex $u\in D$ such that $N(v)\cap D =\set{u}$. In this case we say that $v$ is $u$'s \e{private neighbor}.
\end{citeddefinition}
Another way to state Proposition~\ref{prop:DSMin} is to say that $D$ is a minimal dominating set if and only if $D$ is dominating, and for every $u\in D$ that does not have a private neighbor it holds that $N(u)\cap D=\emptyset$.\eat{
for every $u\in D$ it holds that $N(u)\cap D=\emptyset$ or $u$ has a private neighbor in $\nodes(G){\setminus} D$. When $N(u)\cap D=\emptyset$, we say that $u$ is its own private neighbor. When $u$ has a private neighbor in $\nodes(G){\setminus}D$, we say that $u$ has an \e{external private neighbor}.}
\ifranked We assume a weight function $w:\nodes(G)\rightarrow \real$ on the vertices of the graph, and for a subset $D\subseteq\nodes(G)$ denote $w(D)\eqdef \sum_{v\in D}w(v)$.
\fi
}
A \e{hypergraph} $\H$ is a pair $(\nodes(\H), \Hedges(\H))$, where $\nodes(\H)$, called the vertices of $\H$, is a finite set and $\Hedges(\H)\subseteq 2^{\nodes(\H)}{\setminus} \set{\emptyset}$ are called the hyperedges of $\H$. The \e{dual} of a hypergraph $\H$, denoted $\H^d$, is the hypergraph with vertices $\nodes(\H^d)\eqdef \set{y_e: e\in \Hedges(\H)}$, and hyperedges $\Hedges(\H^d)\eqdef \set{f_v: v\in \nodes(\H)}$, where $f_v\eqdef\set{y_e \in \nodes(\H^d) : v \in e}$.
\begin{citeddefinition}{\cite{DBLP:journals/siamcomp/EiterGM03}}\textsc{(Incidence Graph)}
	\label{def:incidenceGraph}
	The \e{incidence graph} of a hypergraph $\H$, denoted $I(\H)$, is the undirected, bipartite graph with vertex
	set $\nodes(I(\H))\eqdef \nodes(\H)\cup \set{y_e : e \in \Hedges(\H)}$ and edge set $\edges(I(\H))\eqdef \set{(x,y_e) : x \in \nodes(\H), e \in \Hedges(\H), \text{ and } x \in e}$. 
\end{citeddefinition}
Observe that a hypergraph $\H$ and its dual $\H^d$ have the same incidence graph.
The size of a hypergraph $\H$, denoted $||\H||$, is $|\nodes(\H)|+\sum_{e \in \Hedges(\H)}|e|$. 
\eat{
A \e{dominating set} of a graph $G$ is a subset of vertices $D \subseteq \nodes(G)$, such that every vertex not in $D$ is adjacent to at least one member in $D$. A dominating set $D$ is minimal if no proper subset of $D$ is a dominating set.} A \e{transversal} of a hypergraph $\H$ is a subset  $\M \subseteq \nodes(\H)$ that intersects every hyperedge $e\in \Hedges(\H)$. That is, $\M \cap e \neq \emptyset$ for every $e \in \Hedges(\H)$. It is a \e{minimal transversal} if it does not contain any transversal as a proper subset.
It is well known that a transversal $\M \subseteq \nodes(\H)$ is minimal if and only if for every $u \in \M$, there exists a hyperedge $e_u \in \Hedges(\H)$, such that $e_u\cap \M=\set{u}$~\cite{DBLP:journals/dam/MurakamiU14}. An \e{edge cover} of a hypergraph $\H$ is a subset $\D\subseteq \Hedges(\H)$ such that $\bigcup_{e\in \D}e=\nodes(\H)$.  It is a \e{minimal edge cover} if it does not contain any edge cover as a proper subset. The following shows that \textsc{trans-enum} and \textsc{cover-enum} are equivalent. 
\def\transversalCover{
	A subset $\D\subseteq \Hedges(\H)$ is a minimal edge cover of $\H$ if and only if the set $\set{y_e: e\in \D}$ is a minimal transversal of $\H^d$.
}
\begin{proposition}
	\label{prop:transversalCover}
	\transversalCover
\end{proposition}
\eat{
	\begin{proof}
		Recall that $\nodes(\H^d)\eqdef\set{y_e: e\in \Hedges(\H)}$, and $\nodes(\H^d)\eqdef \set{f_v: v\in \nodes(\H)}$, where $f_v \eqdef \set{y_e\in \nodes(\H^d): v\in e}$.
		Let $f_v\in \Hedges(\H^d)$. Since $\D$ is an edge cover of $\H$, then there exists a hyperedge $e\in \D$ such that $v\in e$. By definition, this means that $y_e\in f_v$, and hence $y_e\in f_v\cap \D$. Therefore, $\D$ is a transversal of $\H^d$. We now show that $\D$ is a minimal transversal of $\H^d$. Suppose that it is not, and that $\D'\eqdef \D{\setminus}\set{y_e}$ is a transversal of $\H^d$, where $y_e\in \D$. That is, $f_v\cap \D'\neq \emptyset$ for all $v\in \nodes(\H)$. But this means that $\bigcup_{y_e\in \D'}e=\nodes(\H)$, and hence $\D'\subset \D$ is an edge cover; a contradiction.
		
		For the other direction, we first show that $\D$ is an edge cover of $\H$. Let $v\in \nodes(\H)$. Since $\D$ is a transversal of $\H^d$, then $f_v \cap \D \neq \emptyset$ for all $v\in \nodes(\H)$. Therefore, there exists a vertex $y_e\in \D$ such that $y_e\in f_v$. By construction, this means that $v\in e$. In particular, $v\in \bigcup_{y_e\in \D}e$, and hence $\D$ is an edge cover of $\H$. If $\D$ is not minimal, then there exists a $\D'\subset \D$ that is an edge cover of $\H$. By the previous, this means that $\D'\subset \D$ is a transversal of $\H^d$. But this is a contradiction to the minimality of $\D$.
	\end{proof}
}

\begin{definition}[Tree Decomposition (TD)]\label{def:TreeDecomposition}
	A \e{Tree Decomposition (TD)} of a graph $G$ is a pair $\left(\T,\jointreeMapFunction\right)$
	where $\T$ is an undirected tree with nodes $\nodes(\T)$, edges $\edges(\T)$,
	and $\jointreeMapFunction$ is a
	function that maps each $u \in \nodes(\T)$ to a
	set of vertices $\jointreeMapFunction(u) \subseteq \nodes(G)$, called a \e{bag}, such
	that:
	\begin{enumerate}
		\item $\bigcup_{u \in \nodes(\T)}\jointreeMapFunction(u)=\nodes(G)$.
		\item For every $(s,t){\in} \edges(G)$ there is a node $u{\in} \nodes(\T)$ such that $\set{s,t} {\subseteq} \jointreeMapFunction(u)$.
		\item For every $v\in \nodes(G)$ the set $\setof{u\in \nodes(\T)}{v \in \jointreeMapFunction(u)}$ is connected in $\T$; this is called the \e{running intersection property}.
	\end{enumerate}
	The \e{width} of $\left(\T,\jointreeMapFunction\right)$ is the size of the largest bag minus one (i.e., $\max\set{|\jointreeMapFunction(u)|: u \in \nodes(\T)}-1$). The \e{tree-width} of $G$, denoted $\tw(G)$, is the minimal width of a TD for $G$.
\end{definition}
\eat{
A TD $(\T,\jointreeMapFunction)$ is \e{simple} if $\jointreeMapFunction(u_1) \not\subseteq \jointreeMapFunction(u_2)$ for any two distinct nodes $u_1,u_2 \in \nodes(\T)$.}
There is more than one way to define the treewidth of a hypergraph~\cite{DBLP:journals/siamcomp/EiterGM03}. In this work, the treewidth of a hypergraph $\H$ is defined to be the treewidth of its incidence graph $I(\H)$~\cite{DBLP:journals/siamcomp/EiterGM03}.

Let $\left(\T,\jointreeMapFunction\right)$ be a TD of a graph $G$. Root the tree at node $r\in \nodes(\T)$ by directing the edges  $\edges(\T)$ from $r$ to the leaves. We say that such a TD is \e{rooted}; in notation $(\T_r,\jointreeMapFunction)$.
For a node $u\in \nodes(\T_r)$, we denote by $\parent(u)$ its parent in $\T_r$ (where $\parent(u)\eqdef\texttt{nil}$ if $u$ is the root node), by $\T_u$ the subtree of $\T_r$ rooted at node $u$, by $V_u\eqdef\bigcup_{t \in \nodes(\T_u)}\jointreeMapFunction(t)$, and by $G_u$ the graph induced by $V_u$ (i.e., $G_u\eqdef G[V_u]$). By this definition, $G_r=G$. We will use `vertices' to refer to the vertices of the graph $G$, and `nodes' to refer to the vertices of the TD.

\begin{citeddefinition}{Disjoint-Branch TD (DBTD)~\cite{DBLP:journals/ipl/Duris12,DBLP:journals/dm/Gavril75}}
	\label{def:disjointBranchTD}
	A rooted TD $(\T_r,\jointreeMapFunction)$ is \e{disjoint branch} if, for every $u\in \nodes(\T_r)$, with children $u_1,\dots,u_k$, it holds that $\jointreeMapFunction(u_i)\cap \jointreeMapFunction(u_j)=\emptyset$ for all $i\neq j$. A TD $(\T,\jointreeMapFunction)$ is \e{disjoint branch} if it has a node $r \in \nodes(T)$ such that $(\T_r,\jointreeMapFunction)$ is disjoint branch.
\end{citeddefinition}
First characterized by Gavril~\cite{DBLP:journals/dm/Gavril75}, 
graphs that have a disjoint branch TD are called \e{rooted directed path graphs}. In~\cite{DBLP:journals/dam/Courcelle12}, disjoint branch TDs are called \e{special tree decompositions}.
For the purpose of enumeration, our algorithm constructs a disjoint-branch TD whose width is at most twice the treewidth of the input graph $G$.
It is common practice to formulate dynamic programming algorithms over TDs that are \e{nice}~\cite{DBLP:books/sp/Kloks94}.
\begin{citeddefinition}{Nice Tree Decomposition~\cite{DBLP:books/sp/Kloks94}}
	\label{def:niceTD}
	A nice TD is a rooted TD $(\T,\jointreeMapFunction)$ with root node $r$\eat{where $\jointreeMapFunction(r)=\emptyset$}, in which each node $u\in \nodes(\T)$ is one of the following types:
	\begin{itemize}
		\item Leaf node: a leaf of $\T$ where $\jointreeMapFunction(u)=\emptyset$.
		\item Introduce node: has one child node $u'$ where $\jointreeMapFunction(u)=\jointreeMapFunction(u')\cup \set{v}$, and $v\notin \jointreeMapFunction(u')$.
		\item Forget node: has one child node $u'$ where $\jointreeMapFunction(u)=\jointreeMapFunction(u')\setminus \set{v}$, and $v\in\jointreeMapFunction(u')$.
		\item Join node: has two child nodes $u_1$, $u_2$ where $\jointreeMapFunction(u)=\jointreeMapFunction(u_1)=\jointreeMapFunction(u_2)$. 
	\end{itemize}
\end{citeddefinition}
Given a TD (Definition~\ref{def:TreeDecomposition}) of a graph $G$ of width $w$, one can in time $O(w^2\cdot |\nodes(G)|)$ compute a nice tree decomposition of $G$ of width $w$ that has at most $O(w|\nodes(G)|)$ nodes~\cite{ParameterizedAlgs,DBLP:books/sp/Kloks94}.

\eat{
\subsection{Enumeration Algorithms}
\label{subsec:prelimsEnumeration}
An \e{enumeration problem} arises when the set of solutions to a problem is too large (e.g., exponential in that of the input) to wait for it to be computed and returned in its entirety. An \e{enumeration algorithm} lists all solutions to the problem without repetitions. Very often, solutions are \e{ranked} by their estimated desirability, and then we require \e{ranked enumeration}, which means that the solutions are returned in the order of their ranks (e.g., by increasing weight).\eat{A \e{Top-$K$} algorithm is a ranked enumeration algorithm that returns the $K$ highest-ranking solutions.}
The yardstick used to measure the efficiency of enumeration algorithms is the \e{delay} between the output of successive solutions. An enumeration algorithm runs in \e{polynomial delay} if the delay is polynomial in the size of the input.
Weaker notions are \e{incremental polynomial delay} if the delay between the $i$th and the $(i+1)$st solutions is bounded by a polynomial of the size of the input plus $i$, and \e{polynomial total time} if the total computation time is polynomial in the sizes of both the input and the output.

In this paper, we view the treewidth $w$ of a graph as a parameter, and present an algorithm for enumerating the minimal dominating sets that, following an FPT-preprocessing phase that takes time $O(nw3^{2w})$, enumerates the minimal dominating sets in fixed-parameter-linear delay $O(nw)$. Importantly, our enumeration algorithm can enumerate the minimal dominating sets under \e{inclusion} and \e{exclusion constraints}. That is, for a pair of disjoint subsets $I,U \subseteq \nodes(G)$, we can enumerate only the minimal dominating sets of $G$ that include the vertices of $I$, and exclude the vertices of $U$, with the same complexity guarantees.

}

%% file: ReductionFromTransEnumToDomEnum.tex
\section{From \textsc{Trans-Enum} to \textsc{Dom-Enum}}
\label{sec:TransEnumReduction}
\eat{
Let $\mathcal{A}$ denote an algorithm that receives as input a graph $G$, and enumerates the minimal $[\sigma,\rho]$-dominating sets of $G$.\eat{We let $\texttt{setup}(\mathcal{A})$ and $\texttt{delay}(\mathcal{A})$ denote the runtime complexities of the set up phase and delay of $\A$ respectively.}
We assume that $\mathcal{A}$ has the additional property that any vertex $v \in \nodes(G)$ can be designated to belong to (be excluded from) the solution set. That is, for a fixed subset $X \subseteq \nodes(G)$, algorithm $\mathcal{A}$ can enumerate all minimal $[\sigma,\rho]$-dominating sets $D$ where $D \cap X=\emptyset$ ($D \supseteq X$) without affecting its setup or delay runtime bounds. We call such an enumeration algorithm \e{constraint-aware}. We prove the following.
\begin{theorem}
\label{thm:enumH}
Let $\mathcal{A}$ be an algorithm that enumerates the minimal dominating sets of a graph $G$ with delay $O(|G|\cdot f(w))$, following a preprocessing phase that takes time $O(|G|\cdot g(w))$, where $f$ and $g$ are computable functions, and $w=\tw(G)$.
If $\A$ is constraint-aware, then it enumerates the minimal transversals of a hypergraph $\H$ with delay $O(||\H||\cdot f(k+1))$, following a preprocessing phase that takes time $O(||\H||\cdot g(k+1))$ where $k=\tw(I(\H))$.
\end{theorem}
}

In this section, we show how the \textsc{trans-enum} problem over a hypergraph $\H$ with treewidth $w$ is translated to the \textsc{dom-enum} problem over a graph $G$ with treewidth $w'\leq w+1$. Due to the equivalence between \textsc{trans-enum} and $\textsc{cover-enum}$ (Proposition~\ref{prop:transversalCover}), the translation also carries over to the $\textsc{cover-enum}$ problem.
We associate with the hypergraph $\H$ a tripartite graph $B(\H)$ that is obtained from the incidence graph $I(\H)$ (Definition~\ref{def:incidenceGraph}) by adding a new vertex that is made adjacent to all vertices in $\nodes(\H)$. Formally, the nodes of $B(\H)$ are defined $\nodes(B(\H))\eqdef \nodes(I(\H))\cup \set{v}$ where $v\notin \nodes(I(\H))$, and $\edges(B(\H))\eqdef \edges(I(\H))\cup \set{(u,v):u\in \nodes(\H)}$. Note that for every vertex in $\set{y_e : e \in \Hedges(\H)}$, it holds that $N_{B(\H)}(y_e)=e$, $N_{B(\H)}(v)=\nodes(\H)$, and $v\notin N_{B(\H)}(y_e)$, for every $e\in \Hedges(\H)$. The construction of $B(\H)$ takes time $O(||\H||)$, and is clearly polynomial. 
\def\thmdomenum{
		A set $\M {\subset} \nodes(\H)$ is a minimal transversal of $\H$ if and only if $\M {\cup} \set{v}$ is a minimal dominating set of $B(\H)$.
}
\begin{theorem}
	\label{thm:domenum}
\thmdomenum
\end{theorem}
\eat{
\begin{proof}	
	If $\M$ is a transversal, then 
	every vertex in $\set{y_e : e \in \Hedges(\H)}$ is adjacent to at least one vertex in $\M$. By the assumption of the Theorem, $\nodes(\H)\setminus \M\neq \emptyset$. Since the vertices in $\nodes(\H)\setminus \M$ are dominated by $v$, we conclude that $\M \cup \set{v}$ is a dominating set of $B(\H)$. Since $\M$ is a minimal transversal, then for any $\M' \subset \M$, it holds that $e \cap \M'=\emptyset$ for some hyperedge $e \in \Hedges(\H)$. But this means that $N(y_e) \cap \M'=\emptyset$ for $y_e \in \nodes(B(\H))$. Hence, $\M\cup \set{v}$ is a minimal dominating set of $\B(\H)$.
	
	If $\M\cup \set{v}$ is a minimal dominating set of $B(\H)$, then $\M$ must be a minimal transversal of $\H$. If not, then let $\M'\subset \M$ be a transversal of $\H$. But then, by the previous, $\M' \cup \set{v}\subset \M\cup\set{v}$ is a dominating set of $B(\H)$; a contradiction.
\end{proof}
}
\eat{
\begin{lemma}
	\label{lem:transIsDS}
	Let $\M$ be a transversal of $\H$. Then $\M \cup \set{v}$ is a dominating set of $B(\H)$.
\end{lemma}
\begin{proof}
Since $\M$ is a transversal, then every vertex in $\set{y_e:e\in \Hedges(\H)}$ is adjacent to at least one vertex in $\M$. Every vertex in $\nodes(\H)\setminus \M$ is dominated by $v$.
\end{proof}
\begin{lemma}
	\label{lem:reductionToEnumDS}
	Let $\H$ be a hypergraph and let $D$ be a minimal dominating set of $B(\H)$, such that $D \subseteq \nodes(\H)\cup \set{v}$. Then $D\setminus \set{v}$ is a minimal transversal of $\H$.
\end{lemma}
\begin{proof}
	Since $D \cap \set{y_e \mid e\in \Hedges(\H)}=\emptyset$, every vertex $y_e$ must be incident to a vertex in $D$. Hence, $D\setminus \set{v}$ is a transversal of $\H$. If $D$ is a minimal dominating set, then by Lemma~\ref{lem:transIsDS} we conclude that $D\setminus \set{v}$ is a minimal transversal.
\end{proof}
An immediate consequence of Lemmas~\ref{lem:transIsDS} and~\ref{lem:reductionToEnumDS} is that a subset $\M \subset \nodes(\H)$ is a minimal transversal of $\H$ if and only if $\M\cup \set{v}$ is a minimal dominating set of $B(\H)$.
}
\eat{
\begin{corollary}
	If $\nodes(\H)$ is not a minimal transversal of $\H$, then $\sch(B(\H))=\nodes(\H) \cup \set{\M{\cup}\set{v}: \M \text{ is a minimal transversal of }\H}$. Otherwise, $\sch(B(\H))=\nodes(\H)$.
\end{corollary}
\begin{proof}
First, observe that $\nodes(\H)$ is a minimal dominating set of $B(\H)$ (i.e., $\nodes(\H)\in \sch(B(\H))$).
If $\nodes(\H)$ is not a minimal transversal of $\H$, then $\M\subsetneq \nodes(\H)$ for every minimal transversal $\M$ of $\H$. The first part of the corollary then immediately follows from Theorem~\ref{thm:domenum}. Otherwise, $\nodes(\H)$ is the unique minimal transversal of $\H$. But this means that every minimal dominating $D\in \sch(B(\H))$ must include $\nodes(\H)$ (i.e., $D\supseteq \nodes(H)$). By minimality, we have that $D=\nodes(\H)$, which makes it the unique minimal dominating set of $B(\H)$.
\end{proof}
}
If $\M=\nodes(\H)$ is a minimal transversal of $\H$, then it is the unique minimal transversal of $\H$. Consequently, $\nodes(\H)$ is the unique minimal dominating set of $B(\H)$ that is contained in $\nodes(\H)\cup \set{v}$.
Hence, the problem of enumerating the minimal transversals of $\H$ is reduced to that of enumerating the minimal dominating sets of $B(\H)$ that exclude the vertex-set $\set{y_e : e\in \Hedges(\H)}$, and include the vertex $v$. Symmetrically, the problem of enumerating the minimal edge covers of $\H$ is reduced to enumerating the minimal dominating sets of the tripartite graph $C(\H)$ that results from $I(\H)$ by adding a new vertex $v\notin \nodes(\I(\H))$ as a neighbor to all vertices $\set{y_e:e\in \Hedges(\H)}$, and enumerating the minimal dominating sets of $C(\H)$ that exclude the vertex-set $\set{u : u \in  \nodes(\H)}$, and include the vertex $v$. 
In Section~\ref{sec:encodingTransEnum}, we show how these restrictions are encoded.

Suppose that $\H$ is a hypergraph with treewidth $k$. That is, the treewidth of the incidence graph $\tw(I(\H))=k$~\cite{DBLP:journals/siamcomp/EiterGM03}. 
Next, we show that the treewidth of $B(\H)$ and $C(\H)$ is at most $k+1$. This claim is directly implied by the following lemma.
\def\lemmaaddvertexwidth{
		Let $(\T,\jointreeMapFunction)$ be a tree-decomposition for the graph $G$, and let $U\subseteq \nodes(G)$. Let $G'$ be the graph that results from $G$ by adding a vertex $v\notin \nodes(G)$ as a neighbor to all vertices $U$. That is, $\nodes(G')=\nodes(G)\cup \set{v}$ and $\edges(G')=\edges(G)\cup \set{(v,w): w\in U}$. 
	If the width of $(\T,\jointreeMapFunction)$ is $k$, then $G'$ has a tree-decomposition whose width is $k'\leq k+1$.
}
\begin{lemma}
	\label{lemma:addvertexwidth}
	\lemmaaddvertexwidth
\end{lemma}
\eat{
\blue{
	\begin{proof}
		Let $(\T,\jointreeMapFunction')$ be the tree that results from $(\T,\jointreeMapFunction)$ by adding vertex $v$ to all bags of $(\T,\jointreeMapFunction)$. In other words, for every $u\in \nodes(\T)$, we let $\jointreeMapFunction'(u)\eqdef\jointreeMapFunction(u)\cup \set{v}$. Since $(\T,\jointreeMapFunction)$ is a tree-decomposition for $G$, then clearly $(\T,\jointreeMapFunction')$ meets the three conditions in Definition~\ref{def:TreeDecomposition} with respect to $G'$. Therefore, $(\T,\jointreeMapFunction')$ is a tree-decomposition for $G'$. Since the size of every bag increased by exactly one vertex, the width of $(\T,\jointreeMapFunction')$ is $k+1$.
	\end{proof}
}
}
\eat{
\blue{
\begin{citedlemma}{\cite{DBLP:books/daglib/0023091}}
	Let $(\T,\jointreeMapFunction)$ be a tree-decomposition for the graph $G$, and let $x,y\in \nodes(G)$ be non-adjacent vertices of $G$. Let $G'$ be the graph that results from $G$ by adding the edge $(x,y)$. If the width of $(\T,\jointreeMapFunction)$ is $k$, then $G'$ has a tree-decomposition whose width is $k'\leq k+1$.
\end{citedlemma}}\blue{
\begin{proof}
	We consider two cases. First, suppose that there is a bag $u\in \nodes(\T)$ such that $x,y \in \jointreeMapFunction(u)$. In this case, $(\T,\jointreeMapFunction)$ is also a tree-decomposition for $G'$, because for every edge of $G'$, and for the edge $(x,y)\in \edges(G')$ in particular, there exists a node in $\nodes(\T)$ that contains its endpoints.
	Otherwise, let $u,w\in \nodes(\T)$ such that $x\in \jointreeMapFunction(u)$, and $y\in \jointreeMapFunction(w)$. By item (1) of Definition~\ref{def:TreeDecomposition}, such nodes must exist, and by our assumption they are distinct. Add $y$ to the bag $\jointreeMapFunction(u)$. This addition may have violated the running intersection property (item (3) in Definition~\ref{def:TreeDecomposition}). To fix this, add the vertex $y$ to all nodes on the (single) path between $u$ and $w$ in $\T$.
	The resulting tree meets all conditions of Definition~\ref{def:TreeDecomposition} with respect to the graph $G'$, and its width has increased by at most one.
\end{proof}
}
}

The reduction from \textsc{Trans-Enum} to \textsc{Dom-Enum} of Kant\'{e} et al.~\cite{DBLP:journals/siamdm/KanteLMN14} does not preserve the bounded-treewidth property, and generates a graph whose treewidth is $O(n)$, regardless of the treewidth of the hypergraph, namely $\tw(I(\H))$. Our reduction is similar to theirs, but guarantees that the increase to the treewidth is at most one. That is, in our reduction $\tw(B(\H))\leq \tw(I(\H))+1$ (and $\tw(C(\H))\leq \tw(I(\H))+1$), hence preserving the bounded treewidth property. \eat{Additionally, our reduction relies on the fact that we can enumerate the minimal dominating sets of a graph under exclusion constraints.}

%% file: AlgorithmOverview.tex
\section{Overview of Algorithm}
\label{sec:overview}
In this section, we provide a high-level view of the algorithm for enumerating the minimal dominating sets of a graph $G$, which applies the following simple characterization.
\def\propDSMin{
	A dominating set $D\subseteq \nodes(G)$ is minimal if and only if for every $u\in D$, one of the following holds: (i) $N(u)\cap D=\emptyset$, or (ii) there exists a vertex $v\in N(u){\setminus} D$, such that $N(v)\cap D=\set{u}$.
}
\begin{proposition}
	\label{prop:DSMin}
	\propDSMin
\end{proposition}
\eat{
	\begin{proof}
		Let $D$ be a minimal dominating set, and let $u\in D$. Since $D'\eqdef D{\setminus}\set{u}$ is no longer dominating, there is a vertex $v\in \nodes(G){\setminus} D'$ such that $N(v)\cap D'=\emptyset$. If $v=u$, then $N(u)\cap D'= N(u)\cap D=\emptyset$. If $v\neq u$, then $v\in \nodes(G){\setminus} D$. Since $N(v)\cap D'=\emptyset$, and $N(v)\cap D\neq \emptyset$, we conclude that $N(v)\cap D=\set{u}$.
		
		Now, let $D\subseteq \nodes(G)$ be a dominating set for which the conditions of the proposition hold, and let $D'\eqdef D{\setminus} \set{u}$ for a vertex $u\in D$. If $D'$ is a dominating set, then $N(u)\cap D'=N(u)\cap D\neq \emptyset$. Also, $|N(v)\cap D'|\geq 1$ for every $v\in N(u){\setminus}D'$. In particular, $|N(v)\cap D|\geq 2$ for every $v\in N(u){\setminus}D$. But this brings us to a contradiction.  
	\end{proof}
}
\begin{citeddefinition}{\cite{DBLP:journals/dam/MurakamiU14}}\textsc{(Private Neighbor)}
	\label{def:pn}
	Let $D\subseteq \nodes(G)$.
	A vertex $v\in \nodes(G){
		\setminus} D$ is a \e{private neighbor} of $D$ if there exists a vertex $u\in D$ such that $N(v)\cap D =\set{u}$. In this case we say that $v$ is $u$'s \e{private neighbor}. By \e{private neighbors of $D$} we refer to the vertices in $\nodes(G){\setminus}D$ that are private neighbors of $D$.
\end{citeddefinition}
Another way to state Proposition~\ref{prop:DSMin} is to say that $D$ is a minimal dominating set if and only if $D$ is dominating, and for every $u\in D$ that does not have a private neighbor, it holds that $N(u)\cap D=\emptyset$.
In other words, if $D\in \sch(G)$, then every vertex $u\in D$ either has a private neighbor $v\in \nodes(G){\setminus}D$ where $N(v)\cap D=\set{u}$, or $N(u)\cap D=\emptyset$. Therefore, the minimal dominating set $D\in \sch(G)$ partitions the vertices $\nodes(G)$ into four categories (or labels): 
\begin{enumerate}
	\item Vertices in $D$ that have a private neighbor in $\nodes(G){\setminus}D$ are labeled $[1]_\sigma$.
	\item Vertices $v\in D$ that have no private neighbor in $\nodes(G){\setminus}D$, (i.e., $N(v)\cap D=\emptyset$), are labeled $\sigma_I$.
	\item Vertices $v\in \nodes(G){\setminus} D$ that are private neighbors of $D$, (i.e., $|N(v)\cap D|=1$), are labeled $[1]_\omega$.
	\item Vertices $v\in \nodes(G){\setminus} D$ that are not private neighbors of $D$, (i.e., $|N(v)\cap D|\geq 2$), are labeled $[2]_\rho$.
\end{enumerate}
It is easy to see that distinct minimal dominating sets induce distinct labelings to the vertices of $G$. Therefore, we view $D$ as assigning labels to vertices according to their category.
\begin{example}
	Consider the graph $G$ in Figure~\ref{fig:G}, and the minimal dominating set $D=\set{d,f,i}$. Vertices $a,b,c,e$, and $h$ have exactly one neighbor in $D$, and hence are assigned the label $[1]_\omega$; since $N(g)\cap D=\set{f,i}$, then $g$ is assigned the label $[2]_\rho$. The private neighbors of $d$ are $\set{b,c,e}$, and the private neighbors of $f$ are $\set{a,h}$. Therefore, vertices $d$ and $f$ are assigned label $[1]_\sigma$. Finally, $i$ is assigned label $\sigma_I$ because it has no private neighbors (i.e., $|N(g)\cap D|\geq 2$), and $N(i)\cap D=\emptyset$. 
\end{example}
Next, we refine our labeling as follows. Let $\nodes(G)=\set{v_1,\dots,v_n}$, where the indices represent a complete ordering of $\nodes(G)$ (in Section~\ref{sec:ordering} we discuss this ordering and its properties). We let $V_i\eqdef \set{v_1,\dots,v_i}$ denote the first $i\leq n$ vertices in the ordering. 
\begin{definition}[Labels Induced by minimal dominating set]
	\label{def:induceLabels}
	Let $D\in \sch(G)$ be a minimal dominating set. We define the \e{labels induced by $D$ on $V_i\eqdef \set{v_1,\dots,v_i}$} as follows.
	\begin{enumerate}
\eat{	\item \sout{	\e{Vertices in $D\cap V_i$ that have a private neighbor in $\nodes(G){\setminus}D$.}
		Vertex $v\in V_i$ is labeled $[0]_\sigma$ if and only if $v\in D$, and $v$ has a private neighbor in $V_i{\setminus}D$; it is labeled $[1]_\sigma$ if and only if $v\in D$, it has a private neighbor in $\nodes(G){\setminus}(V_i\cup D)$, and no private neighbor in $V_i{\setminus}D$.} }
		\item \e{Vertices in $D\cap V_i$ that have a private neighbor in $\nodes(G){\setminus}D$.} Vertex $v\in V_i$ is labeled $[0]_\sigma$ if and only if $v\in D$, $v$ has a private neighbor in $V_i{\setminus}D$, and no private neighbors in $\nodes(G){\setminus}V_i$; it is labeled $[1]_\sigma$ if and only if $v\in D$, and it has a private neighbor in $\nodes(G){\setminus}V_i$.
		\item A vertex $v\in D\cap V_i$ is labeled $\sigma_I$ if and only if $v$ has no private neighbors in $\nodes(G){\setminus}D$, and $N(v)\cap D=\emptyset$.
		\item \e{Vertices $v\in V_i{\setminus} D$ that are private neighbors of $D$, and hence $|N(v)\cap D|=1$.} 
		Vertex $v\in V_i{\setminus} D$ is labeled $[0]_\omega$ if and only if $v$ has a single neighbor in $V_i\cap D$.
		Otherwise, $v$ has a single neighbor in $(\nodes(G){\setminus}V_i)\cap D$, and is labeled $[1]_\omega$. 
		\item \e{Vertices $v\in V_i{\setminus} D$ where $|N(v)\cap D|\geq 2$.}
		Let $j\eqdef \min\set{2,|N(v)\cap V_i \cap D|}$. Then $v$ has at least $2-j$ neighbors in $D\cap (\nodes(G){\setminus}V_i)$, and is labeled $[2-j]_\rho$.
	\end{enumerate}
\end{definition}
Overall, we have eight labels that we partition to three groups: $F_\sigma \eqdef \set{[0]_\sigma,[1]_\sigma,\sigma_I}$, $F_\omega\eqdef  \set{[0]_\omega,[1]_\omega}$, and $F_\rho\eqdef \set{[0]_\rho,[1]_\rho,[2]_\rho}$. We let $F\eqdef F_\sigma \cup F_\omega \cup F_\rho$. 
\begin{definition}[Minimal Dominating set consistent with assignment; extendable assignment]
	\label{def:consistentExtendable}
Let $\theta^i: V_i \rightarrow F$ denote a labeling to $V_i$. A minimal dominating set $D\in \sch(G)$ \e{is consistent with} $\theta^i$ if, for every $v\in V_i$, the label that $D$ induces on $v$ is $\theta^i(v)$ (Definition~\ref{def:induceLabels}). We say that the labeling $\theta^i$ is \e{extendable} if there exists a minimal dominating set $D\in \sch(G)$ that is consistent with $\theta^i$.
\end{definition}
\begin{example}
	Consider the path $v_1-v_2-v_3$. The labeling $\theta^1:\set{ v_1 {\gets} \sigma_I}$ is extendable because the minimal dominating set $D_1=\set{v_1,v_3}$ is consistent with $\theta^1$; vertex $v_1 \in D_1$, $v_1$ does not have any private neighbors in $\nodes(G){\setminus}D_1=\set{v_2}$ (because $v_2$ has two neighbors in $D_1$), and $N(v_1)\cap D_1=\emptyset$.
	On the other hand, the assignment $\theta^2: \set{ v_1 {\gets} [0]_\sigma, v_2 {\gets} [0]_\omega}$ is not extendable because the only minimal dominating set that contains $v_1$, and excludes $v_2$, must contain $v_3$, thus violating the constraint that $v_2$ have no dominating neighbors among $\nodes(G){\setminus}V_2=\set{v_3}$, or a single neighbor in the minimal dominating set (item (3) of Definition~\ref{def:induceLabels}).
	Also $\theta^1: \set{v_1 {\gets} [1]_\sigma}$ is not extendable for the same reason; $v_1$ must contain a private neighbor, which can only be $v_2$. If $v_2$ is not part of the solution, then $v_3$ must be included in it. But then, $v_2$ is no longer a private neighbor because both its neighbors (i.e., $v_1$ and $v_3$) are dominating, thus violating the constraint $\theta^1(v_1)=[1]_\sigma$.\eat{ (item (1) of Definition~\ref{def:induceLabels}).}
\end{example}
We denote by $\theta^0$ the empty assignment, which is vacuously extendable.
For every $i\in [1,n]$, we denote by $\bTheta^i$ the set of extendable assignments $\theta^i:\set{v_1,\dots,v_i}\rightarrow F$. Formally:
\begin{align}
	\label{eq:extendable}
	\bTheta^i \eqdef \set{\theta^i: \set{v_1,\dots,v_i}\rightarrow F : \theta^i \text{ is extendable}} && \forall i\in [0,n]
\end{align}
Let $\theta^i:V_i \rightarrow F$ be an assignment. For brevity, we denote by $\nodes_\sigma(\theta^i)$, and $\nodes_\omega(\theta^i)$ the vertices in $V_i$ that are assigned a label in $F_\sigma$ and $F_\omega$, respectively. Formally, $\nodes_\sigma(\theta^i)\eqdef \set{v\in V_i: \theta^i(v) \in F_\sigma}$, $\nodes_\omega(\theta^i)\eqdef \set{v\in V_i: \theta^i(v) \in F_\omega}$, and $\nodes_\rho(\theta^i)\eqdef \set{v\in V_i: \theta^i(v) \in F_\rho}$.

\begin{lemma}
	\label{lem:bijection}
	The following holds: $\sch(G)=\set{\nodes_{\sigma}(\theta^n) : \theta^n \in \bTheta^n}$.
\end{lemma}
\begin{proof}
	By bidirectional inclusion.
	Let $\theta^n\in \bTheta^n$ be an extendable labeling. Then there exists a $D\in \sch(G)$ that is consistent with $\theta^n$. By Definition~\ref{def:induceLabels} (items (1) and (2)), this means that $D\supseteq \nodes_\sigma(\theta^n)$. If $v\notin \nodes_\sigma(\theta^n)$, then by the consistency of $D$, it holds that $v\notin D$. Therefore, $D\subseteq \nodes_\sigma(\theta^n)$, and thus $D=\nodes_\sigma(\theta^n)$. Hence, $\nodes_\sigma(\theta^n)\in \sch(G)$.
	
	Now, let $D\in \sch(G)$. We build the labeling $\theta^n: \nodes(G)\rightarrow F$ as follows. For every $v\in \nodes(G){\setminus}D$, such that $|N(v)\cap D|\geq 2$ we assign the label $[0]_\rho$. For every $v\in \nodes(G){\setminus}D$, such that $|N(v)\cap D|= 1$ we assign the label $[0]_\omega$. Since $D$ is dominating, then the the remaining (unlabeled) vertices must belong to $D$. We assign the label $[0]_\sigma$ to every $v\in D$ that has a private neighbor in $\nodes(G){\setminus}D$ (i.e., $v$ has a neighbor assigned $[0]_\omega$). By Proposition~\ref{prop:DSMin}, the remaining vertices in $D$ do not have any neighbors in $D$, and hence are assigned $\sigma_I$. By construction, $D$ is consistent with $\theta^n$, and hence $\theta^n \in \bTheta^n$. Also by construction, $\nodes_\sigma(\theta^n)=D$. Therefore, $D\in \set{\nodes_\sigma(\theta^n):\theta^n \in \bTheta^n}$.
\end{proof}
According to Lemma~\ref{lem:bijection}, we describe an algorithm for enumerating the set $\set{\nodes_\sigma(\theta^n): \theta^n\in \bTheta^n}$.
We fix a complete order $Q=\langle v_1,v_2,\dots,v_n\rangle$ over $\nodes(G)$.
In the preprocessing phase, the algorithm constructs a data structure $\M$ that receives as input a labeling $\theta^i:V_i\rightarrow F$, and in time $O(\tw(G))$ returns $\true$ if and only if $\theta^i$ is extendable\eat{(i.e., $\theta^i \in \bTheta^i$, Definition~\ref{def:consistentExtendable})}. Here, $V_i=\set{v_1,\dots,v_i}$ are the first $i$ vertices of $Q$. \eat{We make the assumption that $\theta^i$, the input to $\M$, is the result of augmenting an extendable labeling $\theta^{i-1}:V_{i-1}\rightarrow F^{i-1}$ with a label $c_i\in F$ for $v_i$. The precise details of this procedure are described in procedure \algname{IncrementLabeling} (Figure~\ref{fig:IncrementLabeling}) that is deferred to the Appendix.}

The outline of the enumeration algorithm is presented in Figure~\ref{fig:EnumDSOutline} (ignore the algorithm in Figure~\ref{fig:IsExtendible} for now). It is initially called with the empty labeling $\theta^0$, which is vacuously extendable, and the first vertex in the order $v_1$. If $i=n+1$, then the set $\nodes_\sigma(\theta^n)$ is printed in line~\ref{line:printS}. The algorithm then iterates over all eight labels in $F$, and for each label $c_i\in F$ attempts to \e{increment} the extendable labeling $\theta^{i-1}:V_{i-1}\rightarrow F$, to one or more labelings $\theta^i:V_i\rightarrow F$, where $\theta^i(v_i)=c_i$. Assigning label $c_i$ to $v_i$ triggers an update to the labels of $v_i$'s neighbors in $V_{i-1}$ (i.e., $N(v_i)\cap V_i$). The details of this update are specified in the procedure \algname{IncrementLabeling} described in Section~\ref{sec:incrementLabeling}.\eat{
This process has the side affect of updating the label of $v_i$'s neighbors in $V_{i-1}$ (i.e., $N(v_i)\cap V_i$). This process is quite straightforward, and the technical details are presented in Algorithm  \algname{IncrementLabeling} (Figure~\ref{fig:IncrementLabeling}) in the Appendix.} The enumeration algorithm then utilizes the data structure $\M$, built in the preprocessing phase. In line~\ref{line:isExtendable}, the algorithm queries the data structure $\M$, whether the resulting labeling $\theta^i$ is extendable. If and only if this is the case, does the algorithm recurs with the pair $(v_{i+1},\theta^i)$ in line~\ref{line:recurseShort}. 
\eat{Algorithm \algname{EnumDS} belongs to the family of \e{backtracking} enumeration algorithms~\cite{DBLP:journals/networks/ReadT75}.} Since the algorithm is initially called with the empty labeling, which is extendable, it follows that it only receives as input pairs $(\theta^{i-1},v_i)$ where $\theta^{i-1}\in \bTheta^{i-1}$ is extendable. From this observation, we prove the following in the Appendix.
\def\correctnessthm{
	For all $i\in \set{1,\dots,n+1}$, \algname{EnumDS} is called with the pair $(\theta^{i-1},v_i)$ if and only if $\theta^{i-1}\in \bTheta^{i-1}$ is extendable.
}
\begin{theorem}
	\label{thm:correctness}
	\correctnessthm
\end{theorem}
A corollary of Theorem~\ref{thm:correctness} is that \algname{EnumDS} is called with the pair $(\theta^n,v_{n+1})$ if and only if $\theta^n\in \bTheta^n$. From Lemma~\ref{lem:bijection}, we get that  \algname{EnumDS} prints a subset $D\subseteq \nodes(G)$ if and only if $D\in \sch(G)$.
\begin{proposition}
	\label{prop:general}
	Let $Q=\langle v_1,\dots, v_n \rangle$ be a complete order over $\nodes(G)$ such that the following holds for every $i\in \set{1,\dots,n}$:
	\begin{enumerate}
		\item \algname{IncrementLabeling} runs in time $O(\tw(G))$ for every pair $(\theta^{i-1},v_i)$ where $\theta^{i-1}:V_{i-1}\rightarrow F$, and returns a constant number of labelings $\theta^i: V_i \rightarrow F$.
		\item \algname{IsExtendable} runs in time $O(\tw(G))$.
	\end{enumerate}
	Then the delay between the output of consecutive items in algorithm \algname{EnumDS} is $O(n\cdot \tw(G))$.
\end{proposition}
In the rest of the paper, we describe an order $Q=\langle v_1,\dots,v_n\rangle$, and a data structure $\M$ that meet the conditions of Proposition~\ref{prop:general}. We also show how $\M$ can be built in time $O(n|F|^{2\tw(G)})$. 
\def\charsigma{\enquote*{\text{$\sigma$}}}
\def\Icharsigma{\enquote*{\text{$\sigma_I$}}}
\def\charomega{\enquote*{\text{$\omega$}}}
\def\charrho{\enquote*{\text{$\rho$}}}

\renewcommand{\algorithmicrequire}{\textbf{Input:}}
\eat{
\begin{algserieswide}{H}{Algorithm for enumerating minimal dominating sets in FPL-delay when parameterized by treewidth. \label{fig:EnumDSOutline}}
	\begin{insidealgwide}{EnumDS}{$\theta^{i-1}$,$v_i$}		
		\REQUIRE{An extendable labeling $\theta^{i-1}:V_{i-1}\rightarrow F$, and $v_i$, the $i$th vertex in the order $Q$.}
		\IF{$i=(n+1)$}
		\STATE print $\nodes_\sigma(\theta^{i-1})$ \label{line:printS}	
		\RETURN
		\ENDIF
		\FORALL{$c_i \in F$}
		\STATE $\theta^i \gets \algname{IncrementLabeling}(\theta^{i-1},c_i)$
		\IF{$\algname{IsExtendable}(\M,\theta^i)$} \label{line:isExtendable}
		\STATE $\algname{EnumDS}(\theta^i,v_{i+1})$ \label{line:recurseShort}
		\ENDIF
		\ENDFOR
	\end{insidealgwide}
\end{algserieswide}
}

\begin{flushleft} 
	\noindent 
	\begin{minipage}[t]{0.42\textwidth}
		\begin{algseries}{H}{Algorithm for enumerating minimal dominating sets in FPL-delay when parameterized by treewidth. \label{fig:EnumDSOutline}}
			\begin{insidealg}{EnumDS}{$\theta^{i-1}$,$v_i$}		
				\footnotesize
				\REQUIRE{An extendable labeling $\theta^{i-1}\in \bTheta^{i-1}$, and $v_i$, the $i$th vertex in the order $Q$.}
				\IF{$i=(n+1)$}
				\STATE print $\nodes_\sigma(\theta^{i-1})$ \label{line:printS}	
				\RETURN
				\ENDIF
				\FORALL{$c_i \in F$}
				\FORALL{$\theta^i \in \algname{IncrementLabeling}(\theta^{i-1},c_i)$}
				\IF{$\algname{IsExtendable}(\M,\theta^i)$} \label{line:isExtendable}
				\STATE $\algname{EnumDS}(\theta^i,v_{i+1})$ \label{line:recurseShort}
				\ENDIF
				\ENDFOR
				\ENDFOR
			\end{insidealg}
		\end{algseries}
	\end{minipage}
	\hfill
	\begin{minipage}[t]{0.42\textwidth}
		\begin{algseries}{H}{Returns \texttt{true} if and only if $\theta^i\in \bTheta^i$ is extendable. \label{fig:IsExtendible}}
			\begin{insidealg}{IsExtendable}{$(\T,\jointreeMapFunction)$,$\theta^{i}$}		
				\footnotesize
				\REQUIRE{$(\T,\jointreeMapFunction)$: a nice DBTD that is the output of the preprocessing phase of Section~\ref{sec:PreprocessingForEnumeration}. \\ $\theta^i:V_i\rightarrow F$: output of $\algname{IncrementLabeling}(\theta^{i-1},c_i)$ where $\theta^{i-1}:V_{i-1}\rightarrow F$ is extendable and $c_i\in F$.}
				\eat{
					\IF{\sout{$\theta^i(v_i)=\bot$}} \label{line:ifReturnFalse}
					\RETURN \sout{\texttt{false}} \label{line:returnFalse}
					
				}
				\STATE $\FA_i \gets \theta^i[\jointreeMapFunction(B(v_i))]$  \label{line:buildProjections}
				\RETURN $\M_{B(v_i)}(\FA_i)$  \label{line:query}
			\end{insidealg}
		\end{algseries}
	\end{minipage}
\end{flushleft}

\subsection{Ordering Vertices for Enumeration}
\label{sec:ordering}
Let $\left(\T,\jointreeMapFunction\right)$ be 
a TD rooted at node $u_1$. Define $P\eqdef \langle u_1,\dots,u_{|\nodes(\T)|} \rangle$ to be a depth first order of $\nodes(\T)$. For every $i > 1$, $\parent(u_i)$ is a node $u_j\in\nodes(\T)$, with $j<i$ closer to the root $u_1$. Define $B:\nodes(G)\rightarrow \set{1,\dots,|\nodes(\T)|}$ to map every $v\in \nodes(G)$ to the earliest bag $u\in \nodes(\T)$ in the order $P$, such that $v \in \jointreeMapFunction(u)$. 
By the running intersection property of TDs, the mapping $B:\nodes(G)\rightarrow \set{1,\dots,|\nodes(\T)|}$ is well defined; every vertex $v\in \nodes(G)$ is assigned a single node in $\nodes(\T)$.
Let $Q=\langle v_1,\dots,v_n \rangle$ be a complete order of $\nodes(G)$ that is consistent with $B$, such that if $B(v_i)<B(v_j)$ then $v_i\prec_Q v_j$. There can be many such complete orders consistent with $B$, and we choose one arbitrarily.
With some abuse of notation, we denote by $B(v_i)$ both the identifier of the node in the TD, and the bag $\jointreeMapFunction(B(v_i))$. The enumeration algorithm \algname{EnumDS} in Figure~\ref{fig:EnumDSOutline} follows the order $Q$.
\begin{example}
	Consider the graph $G$ in Figure~\ref{fig:G}, and its nice TD in Figure~\ref{fig:niceTD}. Then $B(c)=u$, and $a\prec g\prec b\prec c \prec d \prec e$ in the order $Q$. Also, see the mapping $B$ applied to the vertices of $G$.
\end{example}
\def\orderNbrLem{
	Let $v_i$ be the $i$th vertex in $Q=\langle v_1,\dots,v_n\rangle$.
	Then $N(v_i)\cap \set{v_1,\dots,v_{i-1}}\subseteq B(v_i)$.
}
Lemma~\ref{lem:orderNbr} establishes an important property of the order $Q$, that is crucial for obtaining a delay of $O(nw)$, where $w$ is the width of the TD. Proof is deferred to Section~\ref{sec:orderingAppendix} of the Appendix.
\begin{lemma}
	\label{lem:orderNbr}
	\orderNbrLem
\end{lemma}
\eat{
\begin{proof}
	Let $v_k\in N(v_i)\cap \set{v_1,\dots,v_{i-1}}$.
	Since $v_k \prec_Q v_i$, then $B(v_k)\leq B(v_i)$. If $B(v_k)= B(v_i)$, then $v_k\in B(v_i)$, and we are done. Otherwise,  $B(v_k)< B(v_i)$, and hence $B(v_k)\notin \nodes(\T_{B(v_i)})$.
	Since $v_i \in N(v_k)$, then by Definition~\ref{def:TreeDecomposition} there exists a node $l\in \nodes(\T)$ such that $v_i,v_k\in \jointreeMapFunction(l)$. Since $B(v_i)$ is the first bag (in the depth-first-order $P$) that contains $v_i$, then $l \geq B(v_i)$.
	By the running intersection property, $v_k$ appears in every bag on the path between nodes $B(v_k)$ and $l$ in $\T$, and all bags that contain $v_i$ are descendants of $B(v_i)$ in $\T$. Therefore, node $l$ belongs to $\T_{B(v_i)}$. Hence, $v_k\in \jointreeMapFunction(l)$ belongs to a bag $\jointreeMapFunction(l)$ in the subtree rooted at $B(v_i)$, and to a bag $B(v_k)\notin \nodes(\T_{B(v_i)})$. By the running intersection property, $v_k\in B(v_i)$.
\end{proof}
}
\def\orderNbrLemGreater{
	Let $(\T,\jointreeMapFunction)$ be a DBTD (Definition~\ref{def:disjointBranchTD}), and let $Q=\langle v_1,\dots,v_n\rangle$ be a complete vertex ordering compatible with $\T$.
	Let $v_i$ be the $i$th vertex in $Q$.
	Then $N(v_i)\cap \set{v_{i+1},\dots,v_n}= N(v_i)\cap (\nodes_{B(v_i)}{\setminus}B(v_i))$.
}
\eat{
An immediate corollary from Lemma~\ref{lem:orderNbr} is that if $(\T,\jointreeMapFunction)$ is a TD for $G$ with width $w$, then $|B(v_i)|\leq w+1$ for all $i\in \set{1,2,\dots,n}$, and hence $|N(v_i)\cap V_i|\leq |B(v_i)|-1\leq w$. By Proposition~\ref{prop:incrementruntime}, this means that when applied with the vertex-order $Q$ as defined above, \algname{IncrementLabeling} runs in time $O(\tw(G))$ as required by item (1) of Proposition~\ref{prop:general}. 
}

\eat{
The \e{degeneracy} of an undirected graph $G$ is the smallest integer $k$ such that every subgraph of $G$ contains a vertex with at most $k$ neighbors, and is a well-known measure of its sparsity.
An undirected graph $G$ has degeneracy $k$ if and only if its vertices $\nodes(G)$ can be ordered $\langle v_1,v_2,\dots,v_n \rangle$ such that $|N_G(v_i)\cap \set{v_1,\dots,v_{i-1}}|\leq k$, for every $i\in \set{1,\dots,n}$. It is well-known that the degeneracy of a graph is at most equal to its treewidth. Indeed, by Lemma~\ref{lem:orderNbr}, the order $Q$ described in this section exhibits this property.
}

\subsection{The \algname{IncrementLabeling} Procedure}
\label{sec:incrementLabeling}
Let $\theta^i: V_i \rightarrow F$ be an assignment to the first $i$ vertices of $\nodes(G)$ according to the complete order $Q$ over $\nodes(G)$ defined in Section~\ref{sec:ordering}. Recall that $\nodes_\sigma(\theta^i)$, $\nodes_\omega(\theta^i)$, and $\nodes_\rho(\theta^i)$ are the vertices in $V_i$ that are assigned a label in $F_\sigma$, $F_\omega$, and $F_\rho$ by $\theta^i$, respectively. Likewise, for every $a\in F$, we denote by $\nodes_{a}(\theta^i)$ the vertices in $V_i$ that are assigned label $a$ in $\theta^i$ (e.g., $\nodes_{\sigma_I}(\theta^i)$ and $\nodes_{[0]_\omega}(\theta^i)$  are the vertices in $V_i$  assigned labels $\sigma_I$ and $[0]_\omega$ respectively in $\theta^i$).
The following proposition follows almost immediately from Definition~\ref{def:induceLabels}, and the definition of an extendable assignment (Definition~\ref{def:consistentExtendable}). 
\def\incrementProposition{
	Let $\theta^i: V_i \rightarrow F$ be an extendable assignment. For every $v_j \in V_i$ it holds that:
	\begin{enumerate}
		\item If $\theta^i(v_j)=\sigma_I$, then $\theta^i(w)\in F_\rho$ for every $w\in N(v_j)\cap V_i$. 
		\item If $\theta^i(v_j)=[x]_\omega$ where $x\in \set{0,1}$, then $|N(v_j)\cap \nodes_{\sigma_I}(\theta^i)|=0$, and $|N(v_j)\cap \nodes_{\sigma}(\theta^i)|=1-x$.
		\item If $\theta^i(v_j)=[1]_\sigma$, then $|N(v_j)\cap \nodes_{\sigma_I}(\theta^i)|=0$, and$|N(v_j)\cap \nodes_{[1]_\omega}(\theta^i)|=0$.
		\item If $\theta^i(v_j)=[0]_\sigma$, then $|N(v_j)\cap \nodes_{\sigma_I}(\theta^i)|=0$, $|N(v_j)\cap \nodes_{[1]_\omega}(\theta^i)|=0$, and $|N(v_j)\cap \nodes_{[0]_\omega}(\theta^i)|{\geq} 1$.
		\item If $\theta^i(v_j)=[x]_\rho$ where $x\in \set{1,2}$, then $|N(v_j)\cap \nodes_{\sigma}(\theta^i)|=2-x$;  and if $\theta^i(v_j)=[0]_\rho$, then $|N(v_j)\cap \nodes_{\sigma}(\theta^i)|\geq 2$.
	\end{enumerate}
}
\begin{proposition}
\label{prop:increment}
\incrementProposition
\end{proposition}
\eat{
\begin{hproof}
	Since $\theta^i\in \bTheta^i$ is extendable, then there is a minimal dominating set $D\in \sch(G)$ that is consistent with $\theta^i$. By Definition~\ref{def:induceLabels}, it holds that $\nodes_\sigma(\theta^i)\subseteq D$, for every $v_j\in \nodes_{\sigma_I}(\theta^i)$ it holds that $N(v_j)\cap D=\emptyset$, for every $v_j\in \nodes_\omega(\theta^i)$ it holds that $|N(v_j)\cap D|=1$, and for every $v_j\in \nodes_\rho(\theta^i)$ that $|N(v_j)\cap D|\geq 2$.
	
	Suppose, by contradiction, that item~$(1)$ does not hold. Then there is a vertex $w\in N(v_j)\cap V_i$ such that $\theta^i(w)\in F_\omega \cup F_\sigma$. If $\theta^i(w)\in F_\sigma$, then $N(v_j)\cap D\neq \emptyset$, which means that $D$ is not consistent with $\theta^i$, a contradiction.
	If $\theta^i(w)\in F_\omega$, then by Definition~\ref{def:induceLabels}, it holds that $N(w)\cap D=\set{v_j}$. But this means that $v_j$ has a private neighbor in $\nodes(G){\setminus}D$, which again brings us to a contradiction. 
	The rest of the items of the proposition are proved in a similar fashion, as a direct consequent of Definition~\ref{def:induceLabels}, and are omitted.
\end{hproof}
}
\algname{IncrementLabeling} receives as input a labeling $\theta^{i-1}:V_{i-1}\rightarrow F$, which we assume to be extendable (see \algname{EnumDS} in Figure~\ref{fig:EnumDSOutline}), and a label $c_i\in F$. It generates at most two new assignments $\theta^i:V_i \rightarrow F$, where (1) $\theta^i(w)=\theta^{i-1}(w)$ for all $w\notin N(v_i)$, and (2) For all $w\in N(v_i)\cap V_i$ it holds that $\theta^{i-1}(w)\in \nodes_a(\theta^{i-1})$ if and only if $\theta^i(w)\in  \nodes_a(\theta^{i})$ where $a\in \set{\sigma_I,\sigma,\omega,\rho}$, and (3) $\theta^i(v_i)=c_i$. 
The procedure updates the labels of vertices $w\in N(v_i)\cap V_i$ based on the value $c_i$, so that the conditions of Proposition~\ref{prop:increment} are maintained in $\theta^i$. 
If the conditions of Proposition~\ref{prop:increment} cannot be maintained following the assignment of label $c_i$ to $v_i$, then \algname{IncrementLabeling} returns an empty set, indicating that $\theta^{i-1}$ cannot be incremented with the assignment $\theta^i(v_i)\gets c_i$ while maintaining its extendability. 
\def\incrementLabelingLemma{
		Procedure \algname{IncrementLabeling} receives as input an extendable assignment 
	$\theta^{i-1}: V_{i-1}\rightarrow F$ to the first $i-1$ vertices of $\nodes(G)$ according to the complete order $Q$, and a value $c_i \in F$. There exist at most two extendable assignments $\theta^i: V_i \rightarrow F$ where (1) $\theta^i(w)=\theta^{i-1}(w)$ for all $w\notin N(v_i)$, and (2) For all $w\in N(v_i)\cap V_i$ it holds that $\theta^{i-1}(w)\in \nodes_a(\theta^{i-1})$ if and only if $\theta^i(w)\in  \nodes_a(\theta^{i})$ where $a\in \set{\sigma_I,\sigma,\omega,\rho}$, and (3) $\theta^i(v_i)=c_i$, that will be returned by \algname{IncrementLabeling}. If no such extendable assignment exists, the procedure will return the empty set.
}

The pseudocode of \algname{IncrementLabeling} is deferred to Section~\ref{sec:AppendixOverview} in the Appendix. In what follows, we illustrate with an example.
If $c_i\in F_\sigma$, and $w\in N(v_i)\cap V_i$ where $\theta^{i-1}(w)=[1]_\rho$, then $w$'s label is updated to $\theta^i(w)=[0]_\rho$ (item (5) of Proposition~\ref{prop:increment}). If $c_i=[0]_\sigma$, and $\theta^{i-1}(w)=[1]_\omega$, then $w$'s label is updated to $\theta^i(w)=[0]_\omega$ ((2) of Proposition~\ref{prop:increment}). On the other hand, if $\theta^{i-1}(w)=[0]_\omega$, then by item (2) of Proposition~\ref{prop:increment}, it means that $|N(w)\cap \nodes_\sigma(\theta^{i-1})|=1$. Therefore, if $c_i\in F_\sigma$, then $|N(w)\cap \nodes_\sigma(\theta^{i})|= 2$, thus violating item (2) of Proposition~\ref{prop:increment}  with respect to $\theta^i$. In this case, the procedure willreturn an empty set, indicating that $\theta^{i-1}$ {\bf cannot} be augmented with $v_i \gets c_i$, while maintaining the conditions of Proposition~\ref{prop:increment}, which are required for the labeling to be extendable (Definition~\ref{def:consistentExtendable}). 
\begin{lemma}
	\label{lem:incrementLabeling}
\incrementLabelingLemma
\end{lemma}
\def\incrementLabelingRuntime{
	The runtime of procedure \algname{IncrementLabeling} with input $(\theta^{i-1},c_i)$ where $v_j \prec_Q v_i$ for all $j\in \set{1,\dots,i-1}$ is $O(\tw(G))$.
}
\begin{lemma}
	\label{lem:incrementLabelingRuntime}
	\incrementLabelingRuntime
\end{lemma}
\eat{
\begin{proposition}
	\label{prop:incrementruntime}
The runtime of \algname{IncrementLabeling} with input $(\theta^{i-1},c_i)$ is $O(|N(v_i) \cap V_i|)$.
\end{proposition}
\begin{corollary}
	\label{corr:incrementruntimeQ}
The runtime of procedure \algname{IncrementLabeling} with input $(\theta^{i-1},c_i)$ where $v_j \prec_Q v_i$ for all $j\in \set{1,\dots,i-1}$ is $O(\tw(G))$.
\end{corollary}
\begin{proof}
If $V_i$ are the first $i$ vertices of $Q=\langle v_1,\dots,v_n\rangle$ defined in Section~\ref{sec:ordering}, then by Lemma~\ref{lem:orderNbr} it holds that $N(v_i) \cap V_i \subseteq B(v_i)$. Since $|B(v_i)|\leq \tw(G)+1$, then $|N(v_i)\cap B(v_i)|\leq \tw(G)$. The claim immediately follows from Proposition~\ref{prop:incrementruntime}.
\end{proof}
}
The proofs of Lemmas~\ref{lem:incrementLabeling} and~\ref{lem:incrementLabelingRuntime} are deferred to Section~\ref{sec:AppendixOverview} in the Appendix. 
Proposition~\ref{prop:general} defines two conditions sufficient for obtaining FPL-delay of $O(n\cdot \tw(G))$. In Section~\ref{sec:ordering} we introduced the ordering $Q=\langle v_1,\dots,v_n \rangle$, which guarantees that the runtime of \algname{IncrementLabeling} is in $O(\tw(G))$ (Lemma~\ref{lem:incrementLabelingRuntime}), thus meeting the first condition of Proposition~\ref{prop:general}. We now proceed to the second.

\tikzset{
	solid node/.style={circle,draw,inner sep=1.5,fill=black},
	hollow node/.style={circle,draw,inner sep=1.5}
	ex node/.style={circle,inner sep=1.5}
}

\begin{figure*}
	\begin{flushleft}
		\begin{subfigure}[t]{0.08\textwidth}
				\includegraphics[width=\textwidth]{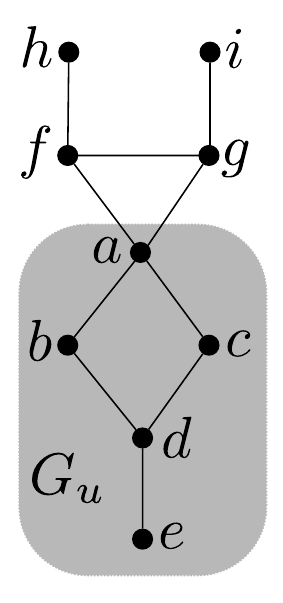}
				\caption{$G$,$G_u$}
				\label{fig:G}
		\end{subfigure}
	\hfill
	\begin{subfigure}[t]{0.55\textwidth} 
			\scalebox{0.8}{
		\begin{tikzpicture}[node distance=0.4cm,every node/.style={scale=0.475},x=1.28cm, y=1cm,font=\footnotesize]
			\tikzset{vertex/.style = {draw, circle, font=\fontsize{11}{12}}}
			\tikzset{doublevertex/.style = {draw, double, circle, font=\fontsize{11}{12}}}
			
			\node[vertex, label=above:$B(h)$] (h) at (0,0) {$h$};
			\node[vertex, label=above:$B(f)$] (hf) at (0.5,0) {$hf$};
			\node[vertex] (f) at (1,0) {$f$};
			\node[vertex, label=above:$B(a)$] (fa) at (1.5,0) {$fa$};
			\node[vertex, label=above:$B(g)$] (fag) at (2,0) {$fag$};
			\node[vertex, label=below:$t$] (ag) at (2.5,0) {$ag$};
			\node[vertex, label=below:$t_0$] (agg) at (3,-1) {$ag$};
			\node[vertex, label=below:$t_1$] (aga) at (3,1) {$ag$};

			\node[vertex] (g) at (3.5,-1) {$g$};
			\node[vertex, label=above:$B(i)$] (gi) at (4,-1) {$gi$};
			\node[vertex] (i) at (4.5,-1) {$i$};
			
			\node[vertex] (a) at (3.5,1) {$a$};
			\node[vertex, label=above:$B(b)$] (ab) at (4,1) {$ab$};
			\node[doublevertex, label=below:$u$, label=above:$B(c)$] (abc) at (4.5,1) {$abc$};
			\node[vertex] (bc) at (5,1) {$bc$};
			\node[vertex, label=above:$B(d)$] (bcd) at (5.5,1) {$bcd$};
			\node[vertex] (cd) at (6,1) {$cd$};
			\node[vertex] (d) at (6.5,1) {$d$};
			\node[vertex, label=above:$B(e)$] (de) at (7,1) {$de$};
			\node[vertex] (e) at (7.5,1) {$e$};
			\eat{
			\begin{scope}[on background layer]
				\node[draw=red,inner sep=10pt, ellipse,fit=(ab) (abc) (bc) (bcd) (cd) (d) (de) (e)] {};
			\end{scope}
		}
			
			\draw[-{Stealth[length=1mm]}] (h) -- (hf);
			\draw[-{Stealth[length=1mm]}] (hf) -- (f);
			\draw[-{Stealth[length=1mm]}] (f) -- (fa);
			\draw[-{Stealth[length=1mm]}] (fa) -- (fag);
			\draw[-{Stealth[length=1mm]}] (fag) -- (ag);
			
			\draw[-{Stealth[length=1mm]}] (ag) -- (agg);
			\draw[-{Stealth[length=1mm]}] (ag) -- (aga);
			
			\draw[-{Stealth[length=1mm]}] (agg) -- (g);
			\draw[-{Stealth[length=1mm]}] (g) -- (gi);
			\draw[-{Stealth[length=1mm]}] (gi) -- (i);
	
			\draw[-{Stealth[length=1mm]}] (aga) -- (a);
			\draw[-{Stealth[length=1mm]}] (a) -- (ab);
			\draw[-{Stealth[length=1mm]}] (ab) -- (abc);
			\draw[-{Stealth[length=1mm]}] (abc) -- (bc);
			\draw[-{Stealth[length=1mm]}] (bc) -- (bcd);
			\draw[-{Stealth[length=1mm]}] (bcd) -- (cd);
			\draw[-{Stealth[length=1mm]}] (cd) -- (d);
			\draw[-{Stealth[length=1mm]}] (d) -- (de);
			\draw[-{Stealth[length=1mm]}] (de) -- (e);
			\draw (4.45,0.5) -- (7.6,0.5) node [midway, fill=white] {$\T(u)$};
			
			\draw[thick, dotted, -{Stealth[length=1mm]}] (ag) -- (a);
			\draw[thick, dotted, -{Stealth[length=1mm]}] (ag) -- (g);
		\end{tikzpicture}
	}
		\caption{Disregarding the dotted edges, this is a nice TD $(\T,\jointreeMapFunction)$ for $G$. The nice disjoint branch TD for $G$ results from adding the dotted edges, and removing nodes $t_0$, $t_1$, and their adjacent edges.
		The subtree $\T_u$ corresponds to the induced graph $G_u\eqdef G[\nodes_u]$, where $\nodes_u=\set{a,b,c,d,e}$.  An order $Q$ corresponding to this nice TD is $Q=\langle h,f,a,g,i,b,c,d,e\rangle$.}
		\label{fig:niceTD}
	\end{subfigure}
	\hfill
	\begin{subfigure}[t]{0.3\textwidth} 
		\scalebox{0.6}{
		{\footnotesize
		\begin{forest}
			[$a$,circle,draw
				[{$[0]_\sigma$}	[$b$,circle,draw
												[{$[0]_\omega$}[$c$,circle,draw[{$[0]_\omega$}[{$\set{a,e}$}]]]]
												[{$[1]_\rho$}[$c$, circle, draw[{$[1]_\rho$}[{$\set{a,d}$}]]]]]]
				[{$[0]_\omega$}[$b$,circle,draw
												[{$[0]_\sigma$}[$c$,circle,draw[{$[1]_\omega$}[{$\set{b,d}$}]]]]
												[{$[1]_\omega$}[$c$,circle,draw[{$[0]_\sigma$}[{$\set{c,d}$}]][{$[1]_\omega$}[{$\set{d}$}]]]]]]
			]
		\end{forest}
	}
}
	\caption{Trie for factor $\M_u:F^{\jointreeMapFunction(u)}\rightarrow \set{0,1}$ for node $u$ in the TD in Figure~\ref{fig:niceTD}. \eat{The assignment $\FA_u$ where $\FA_u(a)=[0]_\sigma$, and $\FA_u(b)=\FA_u(c)=[1]_\rho$ appears in the trie and hence $\M_u(\FA_u)=1$.}The order used to construct the trie is $a\prec b \prec c$ which is compatible with the order $Q$ derived from the TD in Figure~\ref{fig:niceTD}.}
	\label{fig:trie}
\end{subfigure}
	\end{flushleft}
\caption{An undirected graph $G$, its nice TD, and example of a trie representing a factor of the TD.}
\end{figure*}

%% file: preprocessing.tex
\eat{
\subsection{Labeling for \textsc{trans-enum} and proof of Corollary~\ref{corr:TansToDomEnum}}
\label{sec:encodingTransEnum}
In Section~\ref{sec:TransEnumReduction}, we reduced \textsc{trans-enum} on hypergraph $\H$ to \textsc{dom-enum} over the tripartite graph $B\eqdef B(\H)$, where $\nodes(B)=\set{v}\cup \nodes(\H) \cup \set{y_e: e\in\Hedges(\H)}$. Importantly, by Theorem~\ref{thm:domenum}, we are only interested in minimal dominating sets $D$ of $B$, where $D\subsetneq \set{v}\cup \nodes(\H)$, and $D\neq \nodes(\H)$. If no such dominating set exists, then $\nodes(\H)$ is the unique minimal dominating set of $B$ contained in $\set{v}\cup \nodes(\H)$, which in turn means that $\nodes(\H)$ is the unique minimal transversal of $\H$.

Vertices in $\set{y_e: e\in\Hedges(\H)}$ are never dominating, and can only be assigned labels in $F_\omega \cup F_\rho$. If $w\in \nodes(\H)$, then by construction $N(w)\subseteq \set{v}\cup \set{y_e: e\in \Hedges(\H)}$. Consequently, vertices $w\in \nodes(\H)$ can have at most one neighbor in the minimal dominating set; namely the vertex $v$. Therefore, vertices in $\nodes(\H)$ can only be assigned labels in $F_\sigma \cup F_\omega$. Finally, the only dominating set of $B$ that excludes $v$, is $\nodes(\H)$. Since we are interested only in minimal transversals that are strictly included in $\nodes(\H)$, then $v$ must be dominating, and can only be assigned labels in $F_\sigma$.
Since $|F_\sigma \cup F_\omega|=|F_\omega \cup F_\rho|=5$, and $|F_\sigma|=3$, we get that every vertex of $B$ can be assigned one of $5$ labels (as opposed to $8$ as in the general case). 

\blue{
	For every vertex $t\in \nodes(G)$, we denote by $F_t\subseteq F$ the set of labels that may be assigned to vertex $t$. In the general case, $F_t=F$ for all $t\in \nodes(G)$. For enumerating the minimal transversals of a hypergraph $\H$, we reduced  the problem to that of enumerating the minimal dominating sets of the tripartite graph $B$ where, for all $t\in \nodes(B)$ such that $t\in \nodes(\H)$, it holds that $F_t=F_\sigma \cup F_\omega$, for all $t\in \set{y_e: e\in\Hedges(\H)}$, it holds that $F_t=F_\omega \cup F_\rho$, and $F_v=F_\sigma$. 
	In Section~\ref{sec:AppendixProofsPreproc}, we provide the detailed description of the preprocessing phase, and show that it takes time $O(nws^k)$, where $k$ is the width of the disjoint branch TD generated from the nice TD of the graph (see Section~\ref{sec:convertToDBJT}), $n=|\nodes(G)|$, and $s=\max_{v\in \nodes(G)}|F_v|$. In the general case, when no restrictions are placed on the label-assignments to vertices, then $s=|F|=8$. In the case of \textsc{trans-enum}, we get that $s=5$. 
	In Section~\ref{sec:convertToDBJT}, we show that if $\tw(G)=w$, then $G$ has a nice disjoint-branch TD whose \e{effective width} is at most $k=2w$. This proves Corollary~\ref{corr:TansToDomEnum}. }
}

\eat{
\renewcommand{\algorithmicrequire}{\textbf{Input:}}
\renewcommand\algorithmicensure{\textbf{Output:}}
\begin{algseries}{H}{Algorithm for enumerating minimal dominating sets in FPL-delay when parameterized by treewidth. \label{fig:EnumDSOutline}}
	\begin{insidealg}{IncrementLabeling}{$\theta^{i-1}$,$c_i$}		
		\REQUIRE{A labeling $\theta^{i-1}:V_{i-1}\rightarrow F^{i-1}$, and $c_i \in F$}
		\ENSURE{A labeling $\theta^{i}:V_{i}\rightarrow F^{i}$}
		\STATE $K_i\gets N(v_i)\cap V_i$
		\STATE $N_i\gets K_i\cap \nodes_\sigma(\theta^{i-1})$, $W_i\gets K_i \cap \nodes_\omega(\theta^{i-1})$
		\STATE $\theta^i \gets(\theta^{i-1},c_i)$
		\IF{$c_i \in F_\sigma$}
			\FORALL{$v\in K_i$ s.t. $\theta^{i-1}(v)=[x]_\rho$}
				\STATE $\theta^{i}(v)\gets [\max\set{0,x-1}]_\rho$
			\ENDFOR
		\ENDIF
		\IF{$c_i=\sigma_I$ AND $(N_i\neq \emptyset \text{ OR } W_i\neq \emptyset)$}
			\STATE $\theta^i(v_i)\gets \bot$
		\ENDIF
		\IF{$c_i\in \set{[0]_\sigma,[1]_\sigma}$}
			\IF{$\exists w\in K_i$ s.t. $\theta^{i-1}(w)\in \set{\sigma_I,[0]_\omega}$}
				\STATE $\theta^i(v_i)\gets \bot$
			\ELSIF{$\left(c_i=[1]_\sigma \text{,}W_i\neq \emptyset\right)$ OR 		$\left(c_i=[0]_\sigma \text{,}W_i= \emptyset\right)$ }
					\STATE $\theta^i(v_i)\gets \bot$
			\ELSE
				\FORALL{$w\in K_i$ s.t. $\theta^{i-1}(w)=[1]_\omega$}
					\STATE $\theta^i(w)\gets [0]_\omega$
				\ENDFOR
			\ENDIF		
		\ENDIF
		\IF{$c_i\in \set{[0]_\omega,[1]_\omega}$}
			\IF{$\left(\exists w\in K_i \text{ s.t. }\theta^{i-1}(w)=\sigma_I\right) \text{ OR } |N_i|\geq 2$}
				\STATE $\theta^i(v_i)\gets \bot$
			\ELSIF{$(|N_i|=\set{v}, c_i=[0]_\omega)$}
				\STATE $\theta^i(v)\gets [0]_\sigma$
			\ELSIF{$(N_i=\emptyset, c_i=[0]_\omega) \text{ OR } (N_i \neq \emptyset, c_i=[1]_\omega)$}
				\STATE $\theta^i(v_i)\gets \bot$
			\ENDIF
		\ENDIF
		\IF{$c_i =[x]_\rho$ AND $x\neq \max\set{0,2-|N_i|}$}
			\STATE $\theta^i(v_i)\gets \bot$
		\ENDIF
		\end{insidealg}
	\end{algseries}
}

\section{Preprocessing for Enumeration}
\label{sec:PreprocessingForEnumeration}
The data structure generated in the preprocessing-phase is a \e{nice} disjoint-branch tree-decomposition.
\begin{definition}[Nice Disjoint-Branch Tree Decomposition (Nice DBTD)]
	\label{def:niceDBTD}
	A nice DBTD is a rooted TD $(\T,\jointreeMapFunction)$ with root node $r$, in which each node $u\in \nodes(\T)$ is one of the following types:
	\begin{itemize}
		\item Leaf node: a leaf of $\T$ where $\jointreeMapFunction(u)=\emptyset$.
		\item Introduce node: has one child node $u'$ where $\jointreeMapFunction(u)=\jointreeMapFunction(u')\cup \set{v}$, and $v\notin \jointreeMapFunction(u')$.
		\item Forget node: has one child node $u'$ where $\jointreeMapFunction(u)=\jointreeMapFunction(u'){\setminus} \set{v}$, and $v\in\jointreeMapFunction(u')$.
		\item Disjoint Join node: has two child nodes $u_1$, $u_2$ where $\jointreeMapFunction(u)=\jointreeMapFunction(u_1)\cup \jointreeMapFunction(u_2)$ and $\jointreeMapFunction(u_1){\cap} \jointreeMapFunction(u_2){=}\emptyset$. 
	\end{itemize}
\end{definition}
Figure~\ref{fig:niceTD} presents an example of a nice DBTD.
\eat{
While the nice DBTD structure is new, transforming a DBTD (Definition~\ref{def:disjointBranchTD}) to a nice DBTD in polynomial time, while maintaining the width of the TD is similar to transforming a TD to a nice TD (Chapter 13 in~\cite{DBLP:books/sp/Kloks94}). For completeness, we show how this can be done in Section of the Appendix.
Therefore, the proof of Lemma~\ref{lem:buildDBTD} is deferred to the Appendix.
}
The nodes of the generated nice DBTD are associated with \e{factor tables} (described in Section~\ref{subsec:encoding}). 
Every tuple in the factor table associated with node $u\in \nodes(\T)$ is an assignment $\FA_u: \jointreeMapFunction(u) \rightarrow F$ to the vertices of  $\jointreeMapFunction(u)$, where $F$ is the set of labels described in Section~\ref{sec:overview} (see Definition~\ref{def:induceLabels}).
\eat{The preprocessing phase proceeds by dynamic programming over a \e{nice} DBTD.}

\eat{
Due to space restrictions, we defer the detailed description of the dynamic programming algorithm to Appendix~\ref{sec:AppendixProofsPreproc}.}
Since the input graph $G$ is given with a nice TD that is not necessarily disjoint-branch, we describe, in Section~\ref{sec:convertToDBJT}, how to transform its nice TD into one that is nice and disjoint-branch (Definition~\ref{def:niceDBTD}). This involves adding additional ``vertices'' to the bags of the tree, which are not vertices in the original graph. Essentially, some vertices in the graph $G$ will be represented by two or more ``vertices'' in the generated nice DBTD. The introduction of these new vertices is done in a way that maintains the consistency with respect to the vertex-assignments in the original (nice) TD. Importantly, our algorithm does not require the graph $G$ to be a rooted-directed-path graph~\cite{DBLP:journals/ipl/Duris12,DBLP:journals/dm/Gavril75}, and does not change it. The transformation from a nice TD to a nice DBTD comes at the cost of at most doubling the width of the nice TD. The necessity of this step, the details of the transformation, its correctness, and runtime analysis are presented in Section~\ref{sec:convertToDBJT} and Section~\ref{sec:buildDBJTAppendix} of the Appendix. After the transformation from a nice TD to a nice DBTD, the preprocessing phase proceeds by dynamic programming over the generated nice DBTD.

 For the sake of completeness, we show in Section~\ref{sec:buildDBJTAppendix} of the Appendix how to transform a DBTD (Definition~\ref{def:disjointBranchTD}) to a nice DBTD in polynomial time, while preserving its width. While the nice DBTD structure is new, transforming a DBTD to a nice DBTD is similar to transforming a TD to a nice TD (Chapter 13 in~\cite{DBLP:books/sp/Kloks94}). 

In Section~\ref{sec:compactRepresentation}, we describe how to represent the factor tables of Section~\ref{subsec:encoding} as \e{tries}. This allows us to test whether an assignment is extendable in time $O(\tw(G))$, and allows us to obtain fixed-parameter-linear delay of $O(n\cdot \tw(G))$ (see Proposition~\ref{prop:general}).
\eat{
We defer many of the technical details associated with the dynamic-programming to the Appendix and the full version of this paper\footnote{In the supplementary material.}. 
}
\eat{
To solve the existence (or optimization) variant of the minimal domination problem by dynamic programming on a TD $(\T,\jointreeMapFunction)$, every node $u\in \nodes(\T)$ is associated with a memoisation table, denoted $\mu_u$, that is indexed by the classes of partial solutions corresponding to the induced subgraph $G_u$ (recall that $G_u\eqdef G[V_u]$ where $V_u\eqdef \bigcup_{t \in \nodes(\T_u)}\jointreeMapFunction(t)$).
Every class is identified by an assignment of \e{labels} to the vertices $\jointreeMapFunction(u)$. 
The labels encode the subset of vertices of $\jointreeMapFunction(u)$ that are in the partial solution $D$, and the \e{neighborhood constraints} imposed on each vertex in $\jointreeMapFunction(u)$ with respect to $D$ in the induced subgraph $G_u$. 

In~\ref{sec:EncodingPreprcExistence} we describe the labels associated with the vertices, and how they represent the different classes of solutions. In~\ref{sec:dynamicProgExistence}, we describe the dynamic programming algorithm over a disjoint-branch nice TD (\ref{def:niceDBTD}). In~\ref{sec:convertToDBJT}, we show how to convert the encoding of the problem over a nice-TD (\ref{def:niceTD}) to an encoding over a disjoint-branch-nice TD. In~\ref{sec:compactRepresentation}, we show how to represent the memoisation tables as tries; a simple but crucial step for achieving linear delay. The algorithms described in Sections~\ref{sec:dynamicProgExistence}-\ref{sec:compactRepresentation} together make up the preprocessing phase.
}

\subsection{Factors in the TD}
\label{subsec:encoding}
\eat{
A minimal dominating set $D\subseteq \nodes(G)$ partitions the vertices $\nodes(G)$ into three categories: (1) vertices that belong to the solution set $D$, (2) vertices in $\nodes(G){\setminus} D$ that are private neighbors to vertices in $D$ (see Definition~\ref{def:pn}), and (3) vertices in $\nodes(G){\setminus} D$ that are not private neighbors. Since distinct dominating sets lead to distinct partitions, we view $D$ as assigning labels to vertices according to their category.
Based on Proposition~\ref{prop:DSMin}, we further divide the vertices in the solution set $D$ to those that have an external private neighbor in $\nodes(G){\setminus} D$, and to those that are their own private neighbor (i.e., $N(u)\cap D=\emptyset$). For every vertex in $\nodes(G){\setminus} D$, the algorithm tracks the number of neighbors it has in the solution set $D$, thus guaranteeing domination.}
Let $(\T,\jointreeMapFunction)$ be a TD of $G$, and $u\in \nodes(\T)$. Recall that $G_u\eqdef G[V_u]$ is the graph induced by $V_u\eqdef \bigcup_{t\in \nodes(\T_u)}\jointreeMapFunction(t)$. We define a \e{labeling} of the bag $\jointreeMapFunction(u)$ to be an assignment $\FA_u: \jointreeMapFunction(u)\rightarrow F$. There are $|F|^{|\jointreeMapFunction(u)|}$ possible labelings of $\jointreeMapFunction(u)$. In Section~\ref{sec:overview}, we defined what it means for a minimal dominating set $D\in \sch(G)$ to be consistent with a labeling. Next, we define what it means for 
a set of vertices $D_u \subseteq \nodes_u$ to be consistent with an assignment $\FA_u: \jointreeMapFunction(u)\rightarrow F$. \eat{We begin with defining 
the consistency of a set $D_u \subseteq \nodes_u$ with respect to a labeling $f_u: \nodes_u \rightarrow F$.}
\begin{definition}
	\label{def:induceLabelssubtree}
	Let $u\in \nodes(\T)$, and $\FA_u:\jointreeMapFunction(u) \rightarrow F$.
	A subset $D_u\subseteq \nodes_u$ is \e{consistent with $\FA_u$} if the following holds for every $v\in\jointreeMapFunction(u)$:
	\begin{enumerate}
		\item $v\in D_u$ if and only if $\FA_u(v)\in F_\sigma\eqdef \set{[0]_\sigma,[1]_\sigma,\sigma_I}$.
		\item If $\FA_u(v)=[j]_\omega$ where $j\in \set{0,1}$, then $|N(v)\cap D_u|\leq 1$, and $|N(v)\cap (D_u{\setminus}\jointreeMapFunction(u))|=j$ (i.e., $v$ has $j$ neighbors in $D_u{\setminus}\jointreeMapFunction(u)$). \eat{$v$ has exactly $j$ neighbors in $D_u{\setminus}\jointreeMapFunction(u)$ (i.e., $|N(v)\cap (D_u{\setminus}\jointreeMapFunction(u))|=j$).}
		\item  If $\FA_u(v)=[j]_\rho$ where $j\in \set{0,1,2}$, then $v$ has at least $j$ neighbors in $D_u{\setminus}\jointreeMapFunction(u)$. That is, $|N(v)\cap (D_u{\setminus}\jointreeMapFunction(u))|\geq j$.
		\item If $\FA_u(v)=[1]_\sigma$, then $v$ has a neighbor $w\in N(v)\cap (\nodes_u{\setminus}(D_u\cup \jointreeMapFunction(u))$, such that $N(w)\cap D_u=\set{v}$. In addition, none of $v$'s neighbors in $\jointreeMapFunction(u)$ are assigned label $[1]_\omega$.
		\item If $\FA_u(v)=[0]_\sigma$, then for every $w\in N(v){\cap}(\nodes_u{\setminus}(D_u{\cup}\jointreeMapFunction(u)))$, it holds that $|N(w)\cap D_u|\geq 2$.
		\item If $\FA_u(v)=\sigma_I$, then every neighbor of $v$ in $V_u$ is assigned a label in $F_\rho$. \eat{$N(v)\cap D_u=\emptyset$, and every one of $v$'s neighbors in $V_u$ is assigned a label in $F_\rho$
		and every neighbor of $v$ in $\nodes_u{\setminus}\jointreeMapFunction(u)$ has at least two neighbors in $D_u$. That is, for every $w\in N(v)\cap \nodes_u$ it holds that $|N(w)\cap D_u|\geq 2$.}
	\end{enumerate}
In addition, for every $v\in \nodes_u{\setminus}\jointreeMapFunction(u)$ one of the following holds:
	\begin{enumerate}
		\item $v\notin D_u$ and  $N(v)\cap D_u \neq \emptyset$ (i.e., $v$ is dominated by $D_u$).
		\item $v\in D_u$, and has a private neighbor in $V_u{\setminus}D_u$. That is, there exists a $w\in N(v)\cap (\nodes_u{\setminus}D_u)$ such that $N(w)\cap D_u=\set{v}$.
		\item $v\in D_u$, does not have a private neighbor in $V_u{\setminus}D_u$, and $N(v)\cap D_u=\emptyset$.
	\end{enumerate}
\end{definition}
\eat{
Let $u\in \nodes(\T)$, and $\FA_u: \jointreeMapFunction(u) \rightarrow F$. We say that a labeling $f_u:\nodes_u \rightarrow F$ is \e{an extension of} $\FA_u$ if $f_u(v)=\FA_u(v)$ for every $v\in \jointreeMapFunction(u)$, and $f_u(v)\in \set{\sigma_I,[1]_\sigma,[1]_\omega,[2]_\rho}$ otherwise (i.e., $v\in \nodes_u{\setminus}\jointreeMapFunction(u)$).
\begin{definition}
	\label{def:induceLabelsJT}
	Let $u\in \nodes(\T)$, and $\FA_u: \jointreeMapFunction(u) \rightarrow F$.
	The subset $D_u\subseteq \nodes_u$ is \e{consistent with $\FA_u$} if $D_u$ is consistent with a labeling $f_u:\nodes_u\rightarrow F$ that is an extension of $\FA_u$.
\end{definition}
}
We denote by
$\M_u: F^{\jointreeMapFunction(u)} \rightarrow \set{0,1}$ a mapping that assigns every labeling $\FA_u:\jointreeMapFunction(u)\rightarrow F$ a Boolean indicating whether there exists a subset $D_u\subseteq V_u$ that is consistent with $\FA_u$ according to Definition~\ref{def:induceLabelssubtree}. That is,
$\M_u(\FA_u)=1$ if and only if there exists a set $D_u\subseteq V_u$ that is consistent with $\FA_u:\jointreeMapFunction(u) \rightarrow F$ according to Definition~\ref{def:induceLabelssubtree}. The factors of the TD are the mappings $\set{\M_u: u\in \nodes(\T)}$.

\eat{
Let $D$ be any minimal dominating set of $G$, and let $D_u\eqdef D\cap V_u$. In Definition~\ref{def:induceLabels}, we described the eight labels $F=F_\sigma \cup F_\omega \cup F_\rho$ that $D$ induces on $\nodes(G)$. Similarly, we define the labels that $D_u\subseteq \nodes_u$ induces on the vertices of $\nodes_u$.
}

\eat{
\begin{itemize}[itemsep=0mm]
	\item $\sigma_I$: vertices belonging to the set $D'$, which form an independent set (i.e., do not have neighbors in $D'$); these vertices are their own private neighbors, and do not have any external private neighbors.
	\item $[0]_{\sigma_I}$: vertices belonging to the set $D'$, do not have neighbors in $D'$, and have been matched to a private neighbor (i.e., outside of $D$).
	\item $[1]_\sigma$: vertices belonging to the set $D'$, have a neighbor in $D'$, and have not been matched to a private neighbor.
	\item $[0]_\sigma$: vertices belonging to the set $D'$, have a neighbor in $D'$, and have been matched to a private neighbor.
	\item $[1]_\omega$: vertices in $\nodes(G)\setminus D$ that are a private neighbor to some vertex in $D$ (i.e., they have a single neighbor in $D$), but are not yet dominated. That is, they have no neighbors in $D'$.
	\item $[0]_\omega$: vertices in $\nodes(G)\setminus D$ that are a private neighbor to some vertex in $D$, and have already been dominated. That is, they have a single neighbor in $D'$.
	\item $[j]_\rho$ where $j\in \set{0,1,2}$: vertices in $\nodes(G)\setminus D$ that are dominated by at least two vertices in $D$. The state $[j]_\rho$ indicates that the vertex requires domination by $j$ more vertices from $D\setminus D'$.
\end{itemize}
}
\eat{
\begin{enumerate}[itemsep=0mm]
	\item $\sigma$: vertices belonging to the minimal dominating set $D$. 
	\item $\rho$: vertices in $\nodes(G)\setminus D$ with at least two neighbors in $D$.
	\item $\omega$: vertices in $\nodes(G)\setminus D$, with exactly one neighbor in $D$. Every $\omega$-vertex is a private neighbor to its (unique) neighbor in $D$.
\end{enumerate}

We let $F\eqdef F_\sigma \cup F_\rho \cup F_\omega$ where  $F_\sigma\eqdef\set{\sigma_I,[0]_\sigma,[1]_\sigma}$, $F_\omega\eqdef\set{[0]_\omega,[1]_\omega}$, and $F_\rho \eqdef\set{[0]_\rho,[1]_\rho,[2]_\rho}$. Observe that $|F|=8$.

Let $(\T,\jointreeMapFunction)$ denote a TD of $G$, and $u {\in} \nodes(\T)$. 
We define a \e{labeling} of the bag $\jointreeMapFunction(u)$ to be an assignment $\FA_u: \jointreeMapFunction(u)\rightarrow F$. There are $8^{|\jointreeMapFunction(u)|}$ possible labelings of $\jointreeMapFunction(u)$. As described in Definition~\ref{def:induceLabelsJT}, the labels assigned to vertices  $\jointreeMapFunction(u)$ represent their categories ($\sigma$, $\rho$, and $\omega$), and the \e{neighborhood constraints} imposed on them in the induced graph $G_u$.
We denote by
$\mu_u: F^{\jointreeMapFunction(u)} \rightarrow \set{0,1}$ a mapping that assigns every labeling $\FA_u:\jointreeMapFunction(u)\rightarrow F$, a Boolean indicating whether there exists a subset of vertices $D_u\subseteq V_u$ such that:
\begin{enumerate}
	\item For every $v\in V_u\setminus \jointreeMapFunction(u)$, it holds that the label induced by $D_u$ on $v$ belongs to the set $\set{\sigma_I,[1]_\sigma, [1]_\omega,[2]_\rho}$.
	\item The set $D_u$ meets the constraints implied by the labeling $\FA_u$.
\end{enumerate}
}
\begin{example}
Consider node $u$ marked (with a double circle) in Figure~\ref{fig:niceTD}, and an assignment $\FA_u:\jointreeMapFunction(u)\rightarrow F$.
Note that $N(a)\cap (\nodes_u{\setminus}\jointreeMapFunction(u))=N(a)\cap \set{d,e}=\emptyset$. According to Definition~\ref{def:induceLabelssubtree}, every consistent subset $D_u\subseteq \nodes_u$ can induce only one of $\set{\sigma_I,[0]_\omega,[0]_\sigma,[0]_\rho}$ to the vertex $a$. Equivalently, for any $\FA_u:\jointreeMapFunction(u)\rightarrow F$ where $\FA_u(a)\notin \set{\sigma_I,[0]_\omega,[0]_\sigma,[0]_\rho}$, it holds that $\M_u(\FA_u)=0$.

Consider the labeling where $\FA_u(a)=[0]_\sigma$. If $\FA_u(b)=[1]_\omega$, then by Definition~\ref{def:induceLabelssubtree} and the fact that $N(b)\cap (\nodes_u{\setminus}\jointreeMapFunction(u))=\set{d}$, in any $D_u\subseteq \nodes_u$ that is consistent with $\FA_u$, it must hold that $d\in D_u$. But then, $|N(b)\cap D_u|=2 > 1$, violating the constraint that $|N(b)\cap D_u|\leq 1$. Hence, no such consistent set $D_u$ exists, and $\M_u(\FA_u)=0$. If $\FA_u(a)=[0]_\sigma$, $\FA_u(b)=[0]_\omega$, and $\FA_u(c)=[1]_\rho$, then by Definition~\ref{def:induceLabelssubtree} and the fact that $N(c)\cap (\nodes_u{\setminus}\jointreeMapFunction(u))=\set{d}$, in any $D_u\subseteq \nodes_u$ that is consistent with $\FA_u$, it must hold that $d\in D_u$. 
Since $d\in D_u\cap N(b)$, then $D_u$ cannot be consistent with the assignment $\FA_u(b)=[0]_\omega$. Therefore, in this case as well $\M_u(\FA_u)=0$. 

Finally, consider the assignment where $\FA_u(a)=[0]_\sigma$, and $\FA_u(b)=\FA_u(c)=[0]_\omega$. In this case $\M_u(\FA_u)=1$ because the set $D_u=\set{a,e}$ is consistent with $\FA_u$ according to Definition~\ref{def:induceLabelssubtree}. Indeed, $N(b)\cap (D_u{\setminus}\jointreeMapFunction(u))=N(c)\cap (D_u{\setminus}\jointreeMapFunction(u))=\emptyset$, hence $D_u$ is consistent with label $[0]_\omega$ on vertices $b$ and $c$ as required. Since $N(a){\cap}(\nodes_u{\setminus}(D_u{\cup \jointreeMapFunction(u)})=\emptyset$, then $D_u$ is consistent with $\FA_u(a)=[0]_\sigma$.
We now consider vertices $\nodes_u{\setminus}\jointreeMapFunction(u)=\set{d,e}$. Observe that $N(d)\cap D_u=\set{e}$, and thus $d$ is dominated by $D_u$. Also, $e$ has a private neighbor in $\nodes_u{\setminus}D_u$. Indeed, $N(d)\cap D_u=\set{e}$. 
Figure~\ref{fig:trie} presents other assignments $\FA_u$ to the vertices of $\jointreeMapFunction(u)$ where $\M_u(\FA_u)=1$, along with the appropriate consistent subsets of vertices (i.e., according to Definition~\ref{def:induceLabelssubtree}).
\eat{

 where $\FA_u(a)=[1]_\sigma$, $\FA_u(b)=[1]_\omega$, and $\FA_u(c)=[2]_\rho$. If the set $D_u\subseteq \nodes_u$ is consistent with $\FA_u$, the following must hold (Definition~\ref{def:induceLabelsJT}). First, since $\FA_u(c)=[2]_\rho$, then $c$ must have at least two neighbors in $D_u$. Since $N(c)\cap \nodes_u=\set{a,d}$, this means that $d\in D_u$. But then, the constraint implied by $\FA_u(b)=[1]_\omega$ is violated in $D_u$ because $b$ should have exactly one neighbor in $D_u$ (namely, $a$ because $a\in D_u$). Hence, $\mu_u(\FA_u)=0$. 

On the other hand, if $\FA_u(a)=[1]_\sigma$, $\FA_u(b)=[1]_\omega$, and $\FA_u(c)=[1]_\omega$, then the set $D_u\eqdef \set{a,e}$ is consistent with $\FA_u$ (Definition~\ref{def:induceLabelsJT}). To see why, consider the labeling $f_u: \nodes_u \rightarrow F$ where $f_u(x)=\FA_u(x)$ where $x\in \set{a,b,c}$, and where $f_u(d)=[1]_\omega$ and $f_u(e)=[1]_\sigma$. By definition, $f_u$ is an extension of $\FA_u$. Also, we can verify that $D_u$ is consistent with $f_u$. 
First, $N(d)\cap D_u=\set{e}$, and hence the label to $d$ induced by $D_u$ is indeed $[1]_\omega$. Also, since $N(d)\cap D_u=\set{e}$, then $e$ has a private neighbor in $V_u$, and hence the label to $e$ induced by $D_u$ is $[1]_\sigma$. Second, the constraints implied by the assignment $\FA_u$ are met;
$N(b)\cap D_u=\set{a}$, and $N(c)\cap D_u=\set{a}$, thus meeting the constraints associated with the labels $\FA_u(b)=\FA_u(c)=[1]_\omega$. Consequently, vertices $b,c\in V_u{\setminus}D_u$ are private neighbors of vertex $a$, thus meeting the constraint $\FA_u(a)=[1]_\sigma$. Therefore, $\mu_u(\FA_u)=1$.
}
\end{example}

\subsection{From Nice to Disjoint-Nice TD}
\label{sec:convertToDBJT}
In this section, and in Section~\ref{sec:buildDBJTAppendix} of the Appendix, we show how the encoding of an instance of the $\textsc{Dom-Enum}$ problem in a nice TD $(\T,\jointreeMapFunction)$, of width $w$, can be represented using a nice disjoint branch TD $(\T',\jointreeMapFunction')$, whose width is at most $2w$, and how this representation is built in time $O(nw^2)$.
We start with an example that provides some intuition regarding the requirement for a nice DBTD. 
\def\DBJTBuildLemma{
	If a graph $G$ has a nice TD $(\T,\jointreeMapFunction)$ of width $w$, then it has a nice disjoint branch TD $(\T',\jointreeMapFunction')$ whose effective width (see~\eqref{eq:efftw}) is at most $2w$, that can be constructed in time $O(nw^2)$.
}

\def\DBJTBuildLemma_2{
	If a graph $G$ has a nice TD $(\T,\jointreeMapFunction)$ of width $w$, then there is an algorithm that in time $O(nw^2)$ constructs a nice disjoint-branch TD $(\T',\jointreeMapFunction')$ with the following properties:
	\begin{enumerate}[noitemsep]
		\item $\nodes(\T)\subseteq \nodes(\T')$, where for every node $u\in \nodes(\T)$, there is a bijection between $\jointreeMapFunction(u)$ and $\jointreeMapFunction'(u)$.
		\item For every node $u\in \nodes(\T)$, and every $\FA_u:\jointreeMapFunction(u)\rightarrow F$, it holds that $\M_u(\FA_u)=1$ if and only if $\M'_u(\FA'_u)=1$, where $\FA'_u:\jointreeMapFunction'(u)\rightarrow F$ and $\M'_u$ are the corresponding labeling and factor  of node $u$ in $(\T',\jointreeMapFunction')$.
		\item The effective width of $(\T',\jointreeMapFunction')$ is at most $2w$.
	\end{enumerate}

}

\begin{example}
	\label{eq:nonDBJT}
	Let $u\in \nodes(\T)$ be a regular (i.e., not disjoint) join node in a nice TD $(\T,\jointreeMapFunction)$, with children $u_0$ and $u_1$ (Definition~\ref{def:niceTD}). 
	Hence, $\jointreeMapFunction(u)=\jointreeMapFunction(u_0)=\jointreeMapFunction(u_1)$. For simplicity, suppose that $\jointreeMapFunction(u)=\set{v}$. 
	 Consider the assignment $\FA_u: \jointreeMapFunction(u)\rightarrow F$, where $\FA_u(v)=[1]_\omega$. By definition, a subset $D_u\subseteq \nodes_u$ that is consistent with $\FA_u(v)$ (Definition~\ref{def:induceLabelssubtree}) is such that $|D_u\cap N(v)|=1$. This means that
	 $D_u$ includes exactly one neighbor of $v$ in $G_{u_0}$ (and no neighbors of $v$ from $G_{u_1}$), or vice versa. Therefore, to infer whether there exists a subset $D_u\subseteq \nodes_u$ that is consistent with the assignment $\FA_u=\set{v\gets [1]_\omega}$ (i.e., $\M_u(\FA_{u})=1$), we need to assign \e{distinct} labels to $v$ in the two  sub-trees $\T_{u_0}$ and $\T_{u_1}$, corresponding to the induced subgraphs $G_{u_0}$ and $G_{u_1}$ respectively. Specifically, we need to assign $\FA_{u_0}(v)=[1]_\omega$ and $\FA_{u_1}(v)=[0]_\omega$, or $\FA_{u_0}(v)=[0]_\omega$ and $\FA_{u_1}(v)=[1]_\omega$. The transformation of Lemma~\ref{lem:DBJTLemma} represents $v$ as two distinct vertices $v_0,v_1$ in $\T_{u_0}$ and $\T_{u_1}$ respectively, while preserving the integrity of the assignment; that is, $|N(v)\cap D_u|=1$.
\end{example}

Due to space restrictions, the details of the algorithm that generates a nice DBTD from a nice TD $(\T,\jointreeMapFunction)$, proof of correctness and runtime analysis are deferred to Section~\ref{sec:buildDBJTAppendix} of the Appendix. Here, we provide some intuition. The algorithm transforms every join node $u \in \nodes(\T)$ with children $u_0$ and $u_1$ to a disjoint join node (see Definition~\ref{def:niceDBTD}). To that end, the algorithm adds a new node $u'$ to $\T$, such that $\jointreeMapFunction(u')\eqdef \jointreeMapFunction(u)\cup \bigcup_{v\in \jointreeMapFunction(u)}\set{v_0,v_1}$, where $v_0$ and $v_1$ are variables used to represent the neighbors of $v$ in $G[V_{u_0}]$ and $G[V_{u_1}]$ respectively. Every occurrence of $v$ in $\T_{u_0}$ is replaced with $v_0$, and every occurrence of $v$ in $\T_{u_1}$ is replaced with $v_1$. To maintain the niceness of the resulting TD, the algorithm adds $3|\jointreeMapFunction(u)|-2$ additional nodes (of type forget and introduce, see details in Section~\ref{sec:buildDBJTAppendix}).
Since $v_0$ and $v_1$ essentially represent vertex $v$ in $G[V_{u_0}]$ and $G[V_{u_1}]$ respectively, then certain constraints are placed on the assignments $\FA_{u'}: \jointreeMapFunction(u') \rightarrow F$ associated with node $u'$. For example, $\FA_{u'}(v)\in F_a$ if and only if $\FA_{u'}(v_i)\in F_a$ for $a\in \set{\sigma_I,\sigma,\rho ,\omega}$ and $i\in \set{0,1}$. Another such constraint, for example, is that  $\FA_{u'}(v)=[1]_\omega$ if and only if $\FA_{u'}(v_0)=[x_0]_\omega$, $\FA_{u'}(v_1)=[x_1]_\omega$, and $x_0+x_1=1$. The complete list of constraints are specified in Section~\ref{sec:buildDBJTAppendix}. We refer to this set of constraints as \e{local constraints}. Observe that local constraints can be verified in time $O(|\jointreeMapFunction(u')|)$.

Let $u\in \nodes(\T)$, and let $\kappa_u: \FAs{u} \rightarrow \set{0,1}$ be the set of local constraints on
$\FAs{u}$ that can be verified in time $O(|\jointreeMapFunction(u)|)$.\eat{For example, a local constraint may be that the pair of vertices $v_1,v_2 \in \jointreeMapFunction(u)$ belong to the same category (e.g., $\FA_u(v_1)\in F_\sigma$ iff $\FA_u(v_2)\in F_\sigma$).}
We say that an assignment $\FA_u\in \FAs{u}$ is \e{consistent} with $\kappa_u$ if $\kappa_u(\FA_u)=1$. For a node $u\in \nodes(\T)$, and local constraints $\kappa_u: \FAs{u} \rightarrow \set{0,1}$, we let $K_u\eqdef \set{ \FA_u:\jointreeMapFunction(u)\rightarrow F: \kappa_u(\FA_u)=1}$. For a given set of local constraints $\set{\kappa_u: u\in \nodes(\T)}$, and letting $s=|F|$, the \e{effective width} of $(\T,\jointreeMapFunction)$ is defined:
\begin{equation}
	\label{eq:efftw}
	\efftw(\T,\jointreeMapFunction)\eqdef  \max_{u\in \nodes(\T)}\lceil \log_s|K_u| \rceil
\end{equation}
Observe that in any dynamic programming algorithm over a TD, the runtime and memory consumption of the algorithm depends exponentially on the effective width of the TD.
\begin{lemma}
	\label{lem:DBJTLemma}
	\DBJTBuildLemma_2
\end{lemma}
Conditions (1) and (2) of Lemma~\ref{lem:DBJTLemma} guarantee that for every node $u\in \nodes(\T)$, and $\FA_u:\jointreeMapFunction(u)\rightarrow F$, a subset $D_u \subseteq V_u$ is consistent with $\FA_u$ (in $(\T,\jointreeMapFunction)$) according to Definition~\ref{def:induceLabelssubtree}, if and only if $D_u$ is consistent with $\FA'_u$ (in $(\T',\jointreeMapFunction')$) according to Definition~\ref{def:induceLabelssubtree} (following the appropriate name-change to the vertices). This guarantees that, following the preprocessing phase, the generated TD $(\T',\jointreeMapFunction')$ can be used for enumerating the minimal dominating sets of $G$. For an illustration, see Figure~\ref{fig:JTToDBJTIllustrationNew} in the Appendix.  
The proof of Lemma~\ref{lem:DBJTLemma} and the pseudocode detailing the construction of a nice disjoint branch TD is deferred to Section~\ref{sec:buildDBJTAppendix}.

In Section~\ref{sec:preprocessingDP}, we present the algorithm that receives as input a nice DBTD (Definition~\ref{def:niceDBTD}) $(\T,\jointreeMapFunction)$ and a set of local constraints $\set{\kappa_u: u \in \nodes(\T)}$, which computes the factor tables $\M_u:\FA_u\rightarrow \set{0,1}$ for every labeling $\FA_u:\jointreeMapFunction(u) \rightarrow F$, and every node $u\in \nodes(\T)$, by dynamic programming. The factor tables take into account the local constraints. That is, if $\M_u(\FA_u)=1$ then $\kappa_u(\FA_u)=1$. 
We define $s\eqdef \max_{v \in \nodes(G)}|F_v|$, where $F_v \subseteq F$ is the set of labels that may be assigned to $v$.
\def\DPCorrectnessLemma{
	Let $(\T,\jointreeMapFunction)$ be a nice DBTD whose width is $w$, and let $\set{\kappa_u: u \in \nodes(\T)}$ be a set of local constraints. There is an algorithm that in time $O(nws^{w})$ computes the factors of $(\T,\jointreeMapFunction)$ such that for every $u\in \nodes(\T)$, and every labeling $\FA_u:\jointreeMapFunction(u) \rightarrow F$, it holds that $\M_u(\FA_u)=1$ if and only if $\kappa_u(\FA_u)=1$, and there exists a subset $D_u \subseteq V_u$ that is consistent with $\FA_u$ according to Definition~\ref{def:induceLabelssubtree}.
}
\begin{lemma}
	\label{lem:DPCorrectnessLemma}
	\DPCorrectnessLemma
\end{lemma}

\input{fromNiceToDisjointNiceTD_Pic}

\eat{
\begin{hproof}
	Let $(\T,\jointreeMapFunction)$ be a nice TD with root node $r$.  We first set $(\T',\jointreeMapFunction')\gets (\T,\jointreeMapFunction)$, and root $\T'$ at node $r$.
	Let $u \in \nodes(\T')$  be the join node closest to the root that violates the disjoint property (breaking ties arbitrarily), with child nodes $u_0$ and $u_1$. For every $v\in \jointreeMapFunction(u)$, we can partition $N_{G_u}(v)$ as follows:
	\begin{align}
		\label{eq:distributeNbrsBody}
		N_{G_u}(v) {=} (N_{G_u}(v){\cap} \jointreeMapFunction(u)){\cup} \underbrace{(N_{G_{u_0}}(v){\setminus} \jointreeMapFunction(u))}_{\eqdef N_{G_u}^0(v)} {\cup}\underbrace{(N_{G_{u_1}}(v){\setminus} \jointreeMapFunction(u))}_{\eqdef N_{G_u}^1(v)}
	\end{align}
where we denote by $N_{G_u}^0(v)\eqdef (N_{G_{u_0}}(v){\setminus} \jointreeMapFunction(u))$, and $N_{G_u}^1(v)\eqdef (N_{G_{u_1}}(v){\setminus} \jointreeMapFunction(u))$. Clearly, $(N_{G_u}(v){\cap} \jointreeMapFunction(u))\cap N_{G_u}^i(v)=\emptyset$, for $i\in \set{0,1}$. We also note that $N_{G_u}^0(v)\cap N_{G_u}^1(v)=\emptyset$. If not, then there is a vertex $w\in N_{G_u}^0(v)\cap N_{G_u}^1(v)$. By the definition of TD, it must hold that $w\in \jointreeMapFunction(u)$, which brings us to a contradiction.

We update the tree $(\T',\jointreeMapFunction')$ as follows. For every $v\in \jointreeMapFunction(u)$ we create $v_0,v_1$ two fresh copies of $v$. 
	We replace every occurrence of $v$ in $\jointreeMapFunction'(u_0)$ ($\jointreeMapFunction'(u_1)$), and the subtree $\T'_{u_0}$ ($\T'_{u_1}$), with $v_0$ ($v_1$). \ifranked Additionally, we set the weight of the two fresh copies $v_0$ and $v_1$ to $0$ (i.e., $w(v_0)=w(v_1)=0$). The weight of the original vertex $v$ remains unchanged.	\fi
	Following this transition, the bags and subtrees associated with the children of $u$: $u_0$ and $u_1$ are disjoint. The cardinality of the bags in subtrees $\T'_{u_0}$ and $\T'_{u_1}$ remain unchanged\eat{(i.e., we only replaced $v$ with $v_0$ and $v_1$ in subtrees $\T'_{u_0}$ and $\T'_{u_1}$ respectively)}. Finally, we add $v_0$ and $v_1$ to $\jointreeMapFunction'(u)$. Letting $u'$ denote the updated node, we have that $\jointreeMapFunction'(u')=\jointreeMapFunction(u)\cup \bigcup_{v\in \jointreeMapFunction(u)}\set{v_0,v_1}$. 
	Now, we distribute the neighbors $N_{G_u}(v)$ among vertices $v,v_0,v_1\in \jointreeMapFunction'(u')$ according to~\eqref{eq:distributeNbrsBody}, as follows. 
For the new $v_i$ ($i{\in} \set{0,1}$), we restrict its neighbors in $G_{u'}$ to be its neighbors in $\nodes(G_{u_i}){\setminus}\jointreeMapFunction'(u')$. That is, $N_{G_{u'}}(v_i)\eqdef N^i_{G_u}(v)$. For $v\in \jointreeMapFunction(u)\subsetneq \jointreeMapFunction'(u')$, we restrict its neighbors to be only those in $\jointreeMapFunction(u)$. In other words:
	\begin{align}
		N_{G_{u'}}(v_i)\eqdef N_{G_{u_i}}(v_i){\setminus}\jointreeMapFunction'(u')	&&
		N_{G_{u'}}(v)&\eqdef N_{G_u}(v)\cap \jointreeMapFunction(u)\label{eq:DBJTProofBody1}
	\end{align} 
	Since $N_{G_{u'}}(v_i)=N_{G_{u_i}}(v){\setminus}\jointreeMapFunction(u) =N^i_{G_u}(v)$, we have that:
	\begin{align}
		\label{eq:DBJTProofBody2}
		N_{G_{u}}(v) {=} N_{G_{u'}}(v){\cup} N_{G_{u'}}(v_0) {\cup} N_{G_{u'}}(v_1)
	\end{align}
and, by definition, the three sets of vertices $N_{G_{u'}}(v)$,$N_{G_{u'}}(v_0)$, and $N_{G_{u'}}(v_1)$ are pairwise disjoint, as in~\eqref{eq:distributeNbrsBody}.

	Consider the assignment $\FA_{u'}\in F^{\jointreeMapFunction'(u')}$. We can express $\FA_{u'}=(\FA_u,\FA_{u_0},\FA_{u_1})$ where $\FA_u: \jointreeMapFunction(u)\rightarrow F$, and $\FA_{u_i}: \jointreeMapFunction'(u_i)\rightarrow F$ where $i\in \set{0,1}$. 
	In the Appendix, we define the set of local constraints $\kappa_{u'}$ that restrict the set of assignments $\FA_{u'}:\jointreeMapFunction'(u')\rightarrow F$. For example, for every triple $(v_0,v_1,v)$ where $v\in \jointreeMapFunction(u)$, all three vertices are assigned the same category: $\FA_{u'}(v_0),\FA_{u'}(v_1),\FA_{u'}(v)\in F_a$ where $a\in \set{\sigma, \rho, \omega}$. We show that with these constraints, we have that $\efftw(\T',\jointreeMapFunction')\leq 2w$ (see~\eqref{eq:efftw}). 

	Processing a single node $u\in \nodes(\T)$ takes time $O(nw^2)$, and involves updating the vertices in the nodes of the subtrees rooted at $u_0$ and $u_1$ (i.e., to $v_0$ and $v_1$ respectively), and redistributing the neighbors as in~\eqref{eq:DBJTProofBody1} and~\eqref{eq:DBJTProofBody2}. Since the TD contains $O(n)$ nodes, then\eat{
		Processing a node reduces by 1 the number of nodes that violate the disjoint branch property, hence} the overall runtime of the algorithm is in $O(n^2w^2)$.
	To maintain the ``niceness'' of the resulting disjoint-branch TD, we need to add at most $3n-1$ additional nodes to the nice TD: $2n-1$ forget nodes, and $n$ introduce nodes (see Figure~\ref{fig:JTToDBJTIllustration}). 
\end{hproof}
}

\subsection{Factor Sizes and Runtime for \textsc{trans-enum} (proof of Corollary~\ref{corr:TansToDomEnum})}
\label{sec:encodingTransEnum}
As previously described, for every vertex $v\in \nodes(G)$, we denote by $F_v\subseteq F$ the domain of $v$. That is, $F_v$ is 
the set of labels that may be assigned to $v$ in any labeling $\FA_u:\jointreeMapFunction(u)\rightarrow F$ where $u\in \nodes(\T)$ and $v\in \jointreeMapFunction(u)$. 
In Lemma~\ref{lem:DPCorrectnessLemma}, we establish that the dynamic programming over the nice DBTD takes time $O(nws^w)$ where $w$ is the width of the DBTD, $n=|\nodes(G)|$, and $s=\max_{v\in \nodes(G)}|F_v|$. 
\eat{
The factor $s^w$ in the runtime stems from the fact that for every $u\in \nodes(\T)$, it holds that $|\jointreeMapFunction(u)|\leq w+1$ (Definition~\ref{def:TreeDecomposition}), and that every vertex $v\in \jointreeMapFunction(u)$ can be assigned any one of the $s$ labels in $F$. Therefore, the factor of every node $u\in \nodes(\T)$ represents at most $|F|^{w+1}$ distinct labelings $\FA_u: \jointreeMapFunction(u)\rightarrow F$.}
In Section~\ref{sec:convertToDBJT}, we show that if $\tw(G)=k$, then $G$ has an embedding to a disjoint-branch TD whose \e{effective width} is at most $w=2k$. This allows us to prove the first part of Theorem~\ref{thm:DomEnumResult}.
Next, we show that for \textsc{trans-enum}, every vertex can be assigned only one of $s=5$ labels (i.e., as opposed to $|F|=8$ labels), leading to better runtimes of the preprocessing phase, thus proving Corollary~\ref{corr:TansToDomEnum}.

In Section~\ref{sec:TransEnumReduction}, we reduced \textsc{trans-enum} on hypergraph $\H$ to \textsc{dom-enum} over the tripartite graph $B\eqdef B(\H)$, where $\nodes(B)=\set{v}\cup \nodes(\H) \cup \set{y_e: e\in\Hedges(\H)}$. By Theorem~\ref{thm:domenum}, we are interested in minimal dominating sets $D\in \sch(B)$, where $D\subset \set{v}\cup \nodes(\H)$, and $D\neq \nodes(\H)$. 
In particular, $D\in \sch(B)$ where $D\cap \set{y_e: e\in\Hedges(\H)}=\emptyset$. Therefore, vertices in $\set{y_e: e\in\Hedges(\H)}$ can only be assigned labels in $F_\omega \cup F_\rho$. If $w\in \nodes(\H)$, then by construction $N(w)\subseteq \set{v}\cup \set{y_e: e\in \Hedges(\H)}$. Consequently, vertices $w\in \nodes(\H)$ can have at most one neighbor in the minimal dominating set; namely the vertex $v$. Therefore, vertices in $\nodes(\H)$ can only be assigned labels in $F_\sigma \cup F_\omega$. Finally, the only dominating set of $B$ that excludes $v$ and $\set{y_e: e\in\Hedges(\H)}$, is $\nodes(\H)$. Since we are interested only in minimal transversals that are strictly included in $\nodes(\H)$, then $v$ must be dominating, and can only be assigned labels in $F_\sigma$.
Overall, we get that vertices in $\set{y_e: e\in \Hedges(\H)}$ can only be assigned labels from $F_\omega \cup F_\rho$, vertices in $\nodes(\H)$ can only be assigned labels in $F_\sigma \cup F_\omega$, and the vertex $v$ can only be assigned labels in $F_\sigma$. 
Since $|F_\sigma \cup F_\omega|=|F_\omega \cup F_\rho|=5$, and $|F_\sigma|=3$, every vertex of $B$ can be assigned one of at most $5$ labels (as opposed to $8$ as in the  case of \textsc{dom-enum}). Therefore, for \textsc{trans-enum} the runtime of the preprocessing phase is in $O(nw5^w)$. This proves Corollary~\ref{corr:TansToDomEnum}.

\subsection{Compact Representation with Tries}
\label{sec:compactRepresentation}
A trie, also known as a prefix tree, is a tree-based data structure used for efficient retrieval of strings. It organizes the strings by storing characters as tree-nodes, with common prefixes shared among multiple strings. Each node represents a character, and its outgoing edges represent possible next characters. The root node represents an empty string, and each leaf node represents a complete string. 

In Lemma~\ref{lem:DPCorrectnessLemma} we establish the correctness and runtime guarantees of the algorithm that receives as input a nice DBTD $(\T,\jointreeMapFunction)$, and a set of local constraints $\set{\kappa_u: u\in \nodes(\T)}$, which computes the factor tables $\set{\M_u: u\in \nodes(\T)}$ where $\M_u: \FA_u \rightarrow \set{0,1}$.
We represent $\M_u$ as a trie where every root-to-leaf path in $\M_u$ represents an assignment $\FA_u: \jointreeMapFunction(u) \rightarrow F$, or a word in $F^{\jointreeMapFunction(u)}$, such that $\M_u(\FA_u)=1$. 
Let $\jointreeMapFunction(u)=\set{v_1,\dots,v_k}$ where the indices represent a complete order over $\jointreeMapFunction(u)$.
We view every assignment $\FA_u: \jointreeMapFunction(u)\rightarrow F$ as a string: $\FA_u(v_1)\FA_u(v_2)\cdots \FA_u(v_k)$ over the alphabet $F$.
In our construction, every string in the trie has precisely $k$ characters from $F$, where the first character represents the assignment to $v_1$, the second character the assignment to $v_2$, etc. The advantage of this representation is that given an assignment $\FA_u:\jointreeMapFunction(u) \rightarrow F$, we can check if $\M_u(\FA_u)=1$ in $O(k)$ time by traversing the trie from root to leaf to check whether the string $\FA_u(v_1)\cdots \FA_u(v_k)$ appears in the trie $\M_u$. In the tabular representation, this requires traversing all assignments in the table, which, in the worst case, takes time $O(|F|^k)$.
The order used to generate the tries is derived from the complete order $Q$ defined in Section~\ref{sec:ordering}. As we later explain, this is required to guarantee the FPL-delay of the enumeration. See Figure~\ref{fig:trie} for an example of the trie data structure.

%% file: fromNiceToDisjointNiceTD_Pic.tex
\begin{figure}[t]
	\centering
	\begin{subfigure}[t]{0.4\linewidth} 
		\begin{tikzpicture}[node distance=0.2cm,every node/.style={scale=0.475},x=1.28cm, y=0.7cm,font=\footnotesize]
			\tikzset{vertex/.style = {draw, ellipse, font=\fontsize{11}{12}}}
			\tikzset{doublevertex/.style = {draw, double, ellipse, font=\fontsize{11}{12}}}
			
			\node[vertex, label=left:$u_0$] (a0b0) at (-0.5,6) {$ab$};
			\node[vertex, label=right:$u_1$] (a1b1) at (0.5,6) {$ab$};
			\node[vertex, label=left:$u$] (u) at (0,7) {$ab$};
			\node[vertex] (abc) at (0,8) {$abc$};
			
			\draw[-{Stealth[length=1mm]}] (abc) -- (u);
			\draw[-{Stealth[length=1mm]}] (u) -- (a0b0);
			\draw[-{Stealth[length=1mm]}] (u) -- (a1b1);
			
			\draw[thick, dotted] (0,9)--(abc);
			
			\draw   (-0.5,5.5) coordinate(a0) --
			(-0.2,4.5) coordinate (c0) --
			(-0.8,4.5) coordinate (b0) -- cycle;
			\draw[thick, dotted] (-0.5,5.5) -- (-0.5,4.5);
			\draw[-{Stealth[length=1mm]}] (a0b0) -- (a0);
			
			\draw   (0.5,5.5) coordinate(a1) --
			(0.8,4.5) coordinate (c1) --
			(0.2,4.5) coordinate (b1) -- cycle;
			\draw[thick, dotted] (0.5,5.5) -- (0.5,4.5);
			\draw[-{Stealth[length=1mm]}] (a1b1) -- (a1);	
		\end{tikzpicture}
		\caption{Part of a nice TD $(\T,\jointreeMapFunction)$, with join node $u$.}
		\label{fig:origNiceTD}
	\end{subfigure} \hspace{0.3cm}
	\begin{subfigure}[t]{0.4\linewidth} 
		\centering
		\begin{tikzpicture}[node distance=0.2cm,every node/.style={scale=0.475},x=1.28cm, y=0.7cm,font=\footnotesize]
			\tikzset{vertex/.style = {draw, ellipse, font=\fontsize{11}{12}}}
			\tikzset{doublevertex/.style = {draw, double, ellipse, font=\fontsize{11}{12}}}			
			\node[vertex, label=left:$u_0$] (a0b0) at (-0.5,0) {$a_0b_0$};
			\node[vertex, label=right:$u_1$] (a1b1) at (0.5,0) {$a_1b_1$};
			\node[vertex, label=left:$u''$] (utt) at (0,1) {$a_0b_0a_1b_1$};
			\node[vertex, fill=yellow] (i1) at (0,2) {$ba_0b_0a_1b_1$};
			\node[vertex, label=left:$u'$] (ut) at (0,3) {$aba_0b_0a_1b_1$};
			\node[vertex,fill=pink] (f1) at (0,4) {$aba_0b_0a_1$};
			\node[vertex, fill=pink] (f2) at (0,5) {$aba_0b_0$};
			\node[vertex, fill=pink] (f3) at (0,6) {$aba_0$};
			\node[vertex, label=left:$u$] (u) at (0,7) {$ab$};
			\node[vertex] (abc) at (0,8) {$abc$};
			
			\draw[-{Stealth[length=1mm]}] (abc) -- (u);
			\draw[-{Stealth[length=1mm]}] (u) -- (f3);
			\draw[-{Stealth[length=1mm]}] (f3) -- (f2);
			\draw[-{Stealth[length=1mm]}] (f2) -- (f1);
			\draw[-{Stealth[length=1mm]}] (f1) -- (ut);
			\draw[-{Stealth[length=1mm]}] (ut) -- (i1);
			\draw[-{Stealth[length=1mm]}] (i1) -- (utt);
			\draw[-{Stealth[length=1mm]}] (utt) -- (a0b0);
			\draw[-{Stealth[length=1mm]}] (utt) -- (a1b1);
			
			\draw[thick, dotted] (0,9)--(abc);
			
			\draw   (-0.5,-0.5) coordinate(a0) --
			(-0.2,-1.5) coordinate (c0) --
			(-0.8,-1.5) coordinate (b0) -- cycle;
			\draw[thick, dotted] (-0.5,-0.55) -- (-0.5,-1.5);
			\draw[-{Stealth[length=1mm]}] (a0b0) -- (a0);
			
			\draw   (0.5,-0.5) coordinate(a1) --
			(0.8,-1.5) coordinate (c1) --
			(0.2,-1.5) coordinate (b1) -- cycle;
			\draw[thick, dotted] (0.5,-0.55) -- (0.5,-1.5);
			\draw[-{Stealth[length=1mm]}] (a1b1) -- (a1);
		\end{tikzpicture}
		\caption{Part of the disjoint branch TD $(\T',\jointreeMapFunction')$ that results from processing the join node $u$.}
		\label{fig:DBJTAfterTransformation}
	\end{subfigure}
	\caption{The result of processing join node $u\in \nodes(\T)$ in the nice TD $(\T,\jointreeMapFunction)$ in Figure~\ref{fig:origNiceTD}. Note nodes $u'$ and $u''$ in $(\T',\jointreeMapFunction')$ in Figure~\ref{fig:DBJTAfterTransformation}, the introduction of the three forget nodes (in pink), and the single introduce node (in yellow). Observe that $2\cdot |\jointreeMapFunction(u)|-1=3$, and $|\jointreeMapFunction(u)|-1=1$.}
\eat{
\caption{The result of converting join node $u\in \nodes(\T)$ in the nice TD $(\T,\jointreeMapFunction)$ in Figure~\ref{fig:origNiceTD} to the disjoint join node $u''$ in Figure~\ref{fig:DBJTAfterTransformation}. Note nodes $u'$ and $u''$ in $(\T',\jointreeMapFunction')$ in Figure~\ref{fig:DBJTAfterTransformation}, the introduction of the three forget nodes (in pink), and the single introduce node (in yellow). Observe that $2\cdot |\jointreeMapFunction(u)|-1=3$, and $|\jointreeMapFunction(u)|-1=1$.}}

\label{fig:JTToDBJTIllustrationNew}
\end{figure}

%% file: nonRankedEnum.tex
\section{Correctness and Analysis of The Enumeration Algorithm}
\label{sec:enuemration}
We prove correctness of algorithm \algname{EnumDS} presented in Section~\ref{sec:overview}, and show that its delay is $O(nw)$ where $n=|\nodes(G)|$ and $w=\tw(G)$. By Proposition~\ref{prop:general}, this requires showing that both \algname{IncrementLabeling} and \algname{IsExtendable} are correct and run in time $O(w)$. The former is established in Lemmas~\ref{lem:incrementLabeling} and~\ref{lem:incrementLabelingRuntime}. In this section, we present the \algname{IsExtendable} procedure (Figure~\ref{fig:IsExtendible}), prove its correctness, and show that it runs in time $O(w)$.
\eat{
In Section~\ref{sec:ordering}, we defined a complete order $Q=\langle v_1,\dots,v_n\rangle$ over the vertices $\nodes(G)$ that guides the enumeration algorithm \algname{EnumDS}; recall that $V_i\eqdef \set{v_1,\dots,v_i}$ denotes the first $i$ vertices in $Q$. In Section~\ref{sec:incrementLabeling}, we presented the \algname{IncrementLabeling} procedure, and showed that it runs in time $O(w)$ when called with the pair $(\theta^{i-1},c_i)$ where $\theta^{i-1}: V_{i-1}\rightarrow F$, and $c_i\in F$. In this section, we present the \algname{IsExtendable} procedure (Figure~\ref{fig:IsExtendible}), prove its correctness, and show that it runs in time $O(w)$.}The input to \algname{IsExtendable} is the pair $(\M,\theta^i)$ where $\M=(\T,\jointreeMapFunction)$ is the nice DBJT (Definition~\ref{def:niceDBTD}) that is the output of the preprocessing phase described in Section~\ref{sec:PreprocessingForEnumeration}, and $\theta^i:V_i\rightarrow F$ is an assignment which meets the conditions of Proposition~\ref{prop:increment}. \eat{results from calling \algname{IncrementLabeling} on an extendable assignment $\theta^{i-1}:V_{i-1}\rightarrow F$, and a label $c_i\in F$.}In Lemmas~\ref{lem:DBJTLemma} and~\ref{lem:DPCorrectnessLemma} we established that $\M=(\T,\jointreeMapFunction)$ has the property that 
every node $u\in \nodes(\T)$ is associated with a factor $\M_u: F^{\jointreeMapFunction(u)} \rightarrow \set{0,1}$ where $\M_u(\FA_u)=1$ if and only if there exists a subset $D_u \subseteq \nodes_u$ that is consistent with the assignment $\FA_u:\jointreeMapFunction(u) \rightarrow F$ according to Definition~\ref{def:induceLabelssubtree}. In Section~\ref{sec:compactRepresentation}, we showed how to check whether $\M_u(\FA_u)=1$ in time $O(|\jointreeMapFunction(u)|)=O(w)$.
\eat{
to represent the factors $\M_u: F^{\jointreeMapFunction(u)} \rightarrow \set{0,1}$ as tries, such that checking whether $\M_u(\FA_u)=1$ for an assignment $\FA_u: \jointreeMapFunction(u)\rightarrow F$ takes time $O(|\jointreeMapFunction(u)|)=O(w)$.
The pseudo-code of \algname{IsExtendable} is presented in Figure~\ref{fig:IsExtendible}.} Recall that $B:\nodes(G)\rightarrow \set{1,\dots,|\nodes(\T)|}$ maps every $v\in \nodes(G)$ to the first bag in $\T$, with respect to depth-first order of $\nodes(\T)$, that contains it.

\eat{
\renewcommand{\algorithmicrequire}{\textbf{Input:}}
\begin{algserieswide}{t}{Returns \texttt{true} if and only if $\theta^i\in \bTheta^i$ is extendable. \label{fig:IsExtendible}}
	\begin{insidealgwide}{IsExtendable}{$(\T,\jointreeMapFunction)$,$\theta^{i}$}		
		\REQUIRE{$(\T,\jointreeMapFunction)$: a nice DBTD that is the output of the preprocessing phase of Section~\ref{sec:PreprocessingForEnumeration}. \\ $\theta^i:V_i\rightarrow F$: output of $\algname{IncrementLabeling}(\theta^{i-1},c_i)$ where $\theta^{i-1}:V_{i-1}\rightarrow F$ is extendable and $c_i\in F$.}
		\IF{$\theta^i(v_i)=\bot$} \label{line:ifReturnFalse}
			\RETURN \texttt{false} \label{line:returnFalse}
		\ENDIF
		\STATE $\FA_i \gets \theta^i[\jointreeMapFunction(B(v_i))]$  \label{line:buildProjections}
		\RETURN $\M_{B(v_i)}(\FA_i)$  \label{line:query}
	\end{insidealgwide}
\end{algserieswide}
}
\def\thmCorrectness{
\algname{IsExtendable} returns \texttt{true} iff $\theta^i\in \bTheta^i$ is extendable, and runs in time $O(\tw(G))$.
}
\begin{theorem}
	\label{thm:mainCorrectness}
	\thmCorrectness
\end{theorem}

\eat{
\begin{theorem}
	\algname{IsExtendable} runs in time $O(w)$ where $w=\tw(G)$.
\end{theorem}
\begin{proof}
	The test in line~\ref{line:ifReturnFalse} takes $O(1)$. The construction of $\FA_i$ in line~\ref{line:buildProjections} takes time $O(|B(v_i)|)=O(w)$ because, by definition, $|B(v_i)|\leq w+1$. Since $\M_{B(v_i)}$ is represented as a trie, then by the properties of this data structure described in Section~\ref{sec:compactRepresentation}, querying for $\M_{B(v_i)}(\FA_i)$ in line~\ref{line:query} takes time $O(|B(v_i)|)=O(w)$.
\end{proof}
}
The proof of Theorem~\ref{thm:mainCorrectness} relies on Lemma~\ref{lem:BisInjective}, which uses the requirement that $(\T,\jointreeMapFunction)$ is nice, and on Lemma~\ref{lem:useDBTD}, which makes use of the fact that  $(\T,\jointreeMapFunction)$ is disjoint-branch. These are the main ingredients of the proof, that is deferred to Section~\ref{sec:AppendixEnumeration} of the Appendix. 

\def\lemBisInjective{
		Let $(\T,\jointreeMapFunction)$ be a nice TD, and let $v\in \nodes(G)$. For every $w \in B(v){\setminus}\set{v}$, it holds that $B(w) < B(v)$, and hence $B$ is injective.
}
\begin{lemma}
	\label{lem:BisInjective}
	\lemBisInjective
\end{lemma}
\eat{
\begin{proof}
	Since $B(v)$ is the node, closest to the root, that contains $v$, then by definition $v\notin \parent(B(v))$. Since $(\T,\jointreeMapFunction)$ is a nice TD, the node $\parent(B(v))$ can be neither a join node, nor an introduce node because in both of these cases $\jointreeMapFunction(\parent(B(v)))\supseteq \jointreeMapFunction(B(v))$. Consequently, $\parent(B(v))$ must be  a forget node, and hence $\jointreeMapFunction(B(v)){\setminus}\set{v}\subseteq \jointreeMapFunction(\parent(B(v)))$. This means that for every $w\in B(v){\setminus}\set{v}$, it holds that $B(w)\leq \parent(B(v))<B(v)$, and thus the only vertex that is mapped to $B(v)$ is $v$.
\end{proof}
}
\def\propertyDBDT{
		Let $(\T_r,\jointreeMapFunction)$ be a disjoint branch TD. Let $v\in \nodes(G)$. For any pair of distinct nodes $u_1,u_2\in \nodes(\T_r)$ such that $v\in \jointreeMapFunction(u_1)\cap \jointreeMapFunction(u_2)$, it holds that $u_1 \in \nodes(\T_{u_2})$ or $u_2 \in \nodes(\T_{u_1})$.
}
\eat{
\begin{proposition}
	\label{prop:propertyDBDT}
	\propertyDBDT
\end{proposition}

\begin{proof}
	Suppose that neither of the conditions hold. Let $u$ be the least common ancestor of $u_1$ and $u_2$ in $\T_r$. By the assumption, $u\notin\set{u_1,u_2}$, and contains two distinct child nodes $u_1'$ and $u_2'$ such that $u_1\in \nodes(\T_{u_1'})$ and $u_2\in \nodes(\T_{u_2'})$. By the running intersection property, $v\in \jointreeMapFunction(u)\cap \jointreeMapFunction(u_1')\cap \jointreeMapFunction(u_2')$, contradicting the fact that $(\T_r,\jointreeMapFunction)$ is disjoint branch.
\end{proof}
}
\def\lemuseDBTD{
		Let $(\T,\jointreeMapFunction)$ be a nice DBTD (Definition~\ref{def:niceDBTD}), and 
	let $v_i$ be the $i$th vertex in the order $Q=\langle v_1,\dots,v_n \rangle$. For every $v_j\in B(v_i)$, it holds that $N(v_j)\cap \set{v_{i+1},\dots,v_n}= N(v_j)\cap (\nodes_{B(v_i)}{\setminus}B(v_i))$.
}
\begin{lemma}
	\label{lem:useDBTD}
	\lemuseDBTD
\end{lemma}
\eat{
\begin{proof}	
	Let $v_k\in N(v_j)\cap \set{v_{i+1},\dots,v_n}$. By Lemma~\ref{lem:BisInjective}, we have that $v_j \prec_Q v_i$ for every $v_j \in B(v_i){\setminus}\set{v_i}$, and hence $B(v_i) \subseteq \set{v_1,\dots,v_i}$. Therefore, $v_k \notin B(v_i)$,
	and there exists a node $u\in \nodes(\T)$, distinct from $B(v_i)$, such that $v_j,v_k\in \jointreeMapFunction(u)$. 
	Since $(\T,\jointreeMapFunction)$ is a DBTD where $v_j\in \jointreeMapFunction(u)\cap \jointreeMapFunction(B(v_i))$, then by Proposition~\ref{prop:propertyDBDT} either $u\in \nodes(\T_{B(v_i)})$ or $B(v_i)\in \T_{u}$. Since $k\geq i+1$ and $v_k\in \jointreeMapFunction(u)$, the latter is impossible. Therefore, $u\in \nodes(\T_{B(v_i)})$, and hence $v_k \in \jointreeMapFunction(u)\subseteq \nodes_{B(v_i)}{\setminus}B(v_i)$.
	
	Let $v_k\in \nodes_{B(v_i)}{\setminus}B(v_i)$. We claim that $k \geq i+1$. If not, then since $v_k\notin B(v_i)$, then $k\neq i$, and hence $v_k \in V_{i-1}$. By the running intersection property, if $v_k \in V_{i-1}\cap \nodes_{B(v_i)}$, then $v_k\in B(v_i)$, which is a contradiction. Therefore, $v_k\notin V_i$, and $v_k\in \set{v_{i+1},\dots,v_n}$.
\end{proof}
}

\eat{
\begin{hproof}
	Suppose that $\theta^i:V_i \rightarrow F$ is extendable according to Definition~\ref{def:consistentExtendable}. In particular, this means that $\theta^i(v_i)\neq \bot$, and hence \algname{IsExtendable} will not return with \texttt{false} in line~\ref{line:returnFalse}.
	
	From Lemma~\ref{lem:BisInjective}, we have that $B(w)<B(v_i)$, and by definition, $w\prec_Q v_i$ for every $w\in B(v_i)$. Therefore, $B(v_i)\subseteq V_i$. In particular, this means that $\theta^i[B(v_i)]$ assigns a label in $F$ to every vertex in $B(v_i)$. Hence, $\theta^i[B(v_i)]: B(v_i) \rightarrow F$. 
	Since $\theta^i$ is extendable, there exists a minimal dominating set $D\in \sch(G)$ that is consistent with $\theta^i$, and induces the appropriate labels to vertices $V_i$ according to Definition~\ref{def:induceLabels}. We show that $\M_{B(v_i)}(\theta^i[B(v_i)])=1$, thereby proving the claim.	
	To that end, we show that $D_i\eqdef D\cap \nodes_{B(v_i)}$ is consistent with $\theta^i[B(v_i)]$ according to Definition~\ref{def:induceLabelssubtree}. By Definition~\ref{def:induceLabelssubtree}, $D_i$ is consistent with $\theta^i[B(v_i)]$ if, for every $w\in B(v_i)$, the label induced by $D_i$ on $w$ is $\theta^i[B(v_i)](w)=\theta^i(w)$. The label that $D_i$ induces on $w$ is determined by $N(w) \cap (D_i{\setminus}B(v_i))$. 
	From Lemma~\ref{lem:useDBTD}, we have that $N(w)\cap (\nodes_{B(v_i)}{\setminus}B(v_i))=N(w)\cap \set{v_{i+1},\dots,v_n}=N(w)\cap (\nodes(G){\setminus}V_i)$. In particular, $N(w) \cap (D_i{\setminus}B(v_i))=N(w)\cap (D_i{\setminus}V_i)$. Consequently, if $D$ is consistent with $\theta^i$, then  $D_i$ is consistent with $\theta^i[B(v_i)](w)$ for all $w\in B(v_i)$. Now, if $w\in \nodes_{B(v_i)}{\setminus}B(v_i)$, then by the running intersection property, $N(w)\subseteq \nodes_{B(v_i)}$. Therefore $N(w)\cap D=N(w)\cap (D\cap \nodes_{B(v_i)})=N(w)\cap D_i$. Overall, we get that if $D$ is consistent with $\theta^i$, then $D_i$ is consistent with $\theta^i[B(v_i)]$ according to Definition~\ref{def:induceLabelssubtree}.
	
	For the other direction, suppose that \algname{IsExtendable} returns \texttt{true}. We show that there is a minimal dominating set $D\in \sch(G)$ such that $D$ is consistent with $\theta^i$ according to Definition~\ref{def:induceLabels}. We first note that, by our assumption, $\theta^i$ is the result of calling \algname{IncrementLabeling} on the pair $(\theta^{i-1},c_i)$, where $\theta^{i-1}$ is extendable. Hence, there exists a set $S\in \sch(G)$ that is consistent with $\theta^{i-1}$. We let $S'\eqdef S\cap (\nodes(G){\setminus}\nodes_{B(v_i)})$. Since \algname{IsExtendable} returns \texttt{true}, then there is a subset $D_i\subseteq \nodes_{B(v_i)}$ such that $D_i$ is consistent with $\theta^i[B(v_i)]$. We claim that $D\eqdef S'\cup D_i$ is a minimal dominating set that is consistent with $\theta^i$.
	To prove the claim we show, in Section~\ref{sec:AppendixEnumeration} that $D$ meets the conditions of Proposition~\ref{prop:DSMin}, and is consistent with $\theta^i$.
\end{hproof}

}
\eat{
\begin{proof}
	Suppose that $\theta^i:V_i \rightarrow F$ is extendable according to Definition~\ref{def:consistentExtendable}. In particular, this means that $\theta^i(v_i)\neq \bot$, and hence \algname{IsExtendable} will not return with \texttt{false} in line~\ref{line:returnFalse}.
	
	From Lemma~\ref{lem:BisInjective}, we have that $B(w)<B(v_i)$, and by definition, $w\prec_Q v_i$ for every $w\in B(v_i)$. Therefore, $B(v_i)\subseteq V_i$. In particular, this means that $\theta^i[B(v_i)]$ assigns a label in $F$ to every vertex in $B(v_i)$. Hence, $\theta^i[B(v_i)]: B(v_i) \rightarrow F$. 
	Since $\theta^i$ is extendable, there exists a minimal dominating set $D\in \sch(G)$ that is consistent with $\theta^i$, and induces the appropriate labels to vertices $V_i$ according to Definition~\ref{def:induceLabels}. We show that $\M_{B(v_i)}(\theta^i[B(v_i)])=1$, thereby proving the claim.	
	To that end, we show that $D_i\eqdef D\cap \nodes_{B(v_i)}$ is consistent with $\theta^i[B(v_i)]$ according to Definition~\ref{def:induceLabelssubtree}.
	By Definition~\ref{def:induceLabelssubtree}, $D_i$ is consistent with $\theta^i[B(v_i)]$ if, for every $w\in B(v_i)$, the label induced by $D_i$ on $w$ is $\theta^i[B(v_i)](w)=\theta^i(w)$. The label that $D_i$ induces on $w$ is determined by $N(w) \cap (D_i{\setminus}B(v_i))$. 
	From Lemma~\ref{lem:useDBTD}, we have that $N(w)\cap (\nodes_{B(v_i)}{\setminus}B(v_i))=N(w)\cap \set{v_{i+1},\dots,v_n}=N(w)\cap (\nodes(G){\setminus}V_i)$. In particular, $N(w) \cap (D_i{\setminus}B(v_i))=N(w)\cap (D_i{\setminus}V_i)$. Consequently, if $D$ is consistent with $\theta^i$, then  $D_i$ is consistent with $\theta^i[B(v_i)](w)$ for all $w\in B(v_i)$. Now, if $w\in \nodes_{B(v_i)}{\setminus}B(v_i)$, then by the running intersection property, $N(w)\subseteq \nodes_{B(v_i)}$. Therefore $N(w)\cap D=N(w)\cap (D\cap \nodes_{B(v_i)})=N(w)\cap D_i$. Overall, we get that if $D$ is consistent with $\theta^i$, then $D_i$ is consistent with $\theta^i[B(v_i)]$ according to Definition~\ref{def:induceLabelssubtree}.

	For the other direction, suppose that \algname{IsExtendable} returns \texttt{true}. We show that there is a minimal dominating set $D\in \sch(G)$ such that $D$ is consistent with $\theta^i$ according to Definition~\ref{def:induceLabels}. We first note that, by our assumption, $\theta^i$ is the result of calling \algname{IncrementLabeling} on the pair $(\theta^{i-1},c_i)$, where $\theta^{i-1}$ is extendable and $c_i \in F$. Hence, there exists a set $S\in \sch(G)$, such that $S$ is consistent with $\theta^{i-1}$. We let $S'\eqdef S\cap (\nodes(G){\setminus}\nodes_{B(v_i)})$. Since \algname{IsExtendable} returns \texttt{true}, then there exists a subset $D_i\subseteq \nodes_{B(v_i)}$ such that $D_i$ is consistent with $\theta^i[B(v_i)]$. We claim that $D\eqdef S'\cup D_i$ is a minimal dominating set that is consistent with $\theta^i$.
	To prove the claim we show that $D$ meets the conditions of Proposition~\ref{prop:DSMin}, and is consistent with $\theta^i$.
	
	Let us start with a vertex $w\notin \nodes_{B(v_i)}$. Therefore, $w\in D$ if and only if $w\in S' \subseteq S$. By definition of $B(v_i)$, it holds that $N(v_i)\subseteq \nodes_{B(v_i)}$, and hence $v_i\notin N(w)$. Since only vertices in $N(v_i) \cap V_i$ are affected by the assignment of $v_i \gets c_i$ in \algname{IncrementLabeling}, then $\theta^{i-1}(w)=\theta^i(w)$. Since $S$ is consistent with $\theta^{i-1}(w)$, and since $D\cap N(w)=S\cap N(w)$, then $D$ is consistent with $\theta^i(w)$.
	
	Now, let $w\in \nodes_{B(v_i)}{\setminus}B(v_i)$. By the running intersection property, $N(w)\subseteq \nodes_{B(v_i)}$, and hence $N(w)\cap S'=\emptyset$.
	We claim that $w\notin V_i$. Otherwise, by our assumption that $w \notin B(v_i)$, then $w\neq v_i$, and hence $w\in V_{i-1}$. But then, since $w\in \nodes_{B(v_i)}$, then by the running intersection property it must hold that $w\in B(v_i)$, which is a contradiction. Therefore, $w\in \set{v_{i+1},\dots,v_n}$ and $N(w)\subseteq \nodes_{B(v_i)}$.
	Labels to the vertex-set $\set{v_{i+1},\dots,v_n}$ are not specified by $\theta^i$. Therefore, the set $D$ vacuously induces the appropriate label on $w$. Since $D_i$ is consistent with $\theta^i[B(v_i)]$, then $w$ is either dominated by $D_i\subseteq D$, has a private neighbor in $\nodes_{B(v_i)}{\setminus}D_i=\nodes_{B(v_i)}{\setminus}D$, or $N(w)\cap D_i=N(w)\cap D=\emptyset$. 
	
	Finally, assume that $w\in B(v_i)$. From Lemma~\ref{lem:useDBTD}, we have that $N(w)\cap \set{v_{i+1},\dots,v_n}= N(w)\cap (\nodes_{B(v_i)}{\setminus}B(v_i))$. By our assumption, $D_i$ is consistent with $\theta^i(w)$. Since $N(w)\cap \set{v_{i+1},\dots,v_n} \cap D = N(w)\cap D \cap (\nodes_{B(v_i)}{\setminus}B(v_i))$, then $N(w)\cap \set{v_{i+1},\dots,v_n} \cap D=N(w)\cap (D_i {\setminus}B(v_i))$.
	 Therefore, if $D_i$ is consistent with $\theta^i(w)$, then $D$ is consistent with $\theta^i(w)$. Overall, we showed that $D$ is consistent with $\theta^i$. 
\end{proof}
}
\eat{

\section{Enumeration}
\label{sec:nonRankedEnumeration}
Recall from Section~\ref{subsec:encoding}, that every minimal dominating set $D\subseteq \nodes(G)$ induces a labeling to the vertices $\nodes(G)$. The label assigned to a vertex depends on whether it is in $D$, in which case it is assigned a label in $F_\sigma=\set{\sigma_I,[0]_\sigma,[1]_\sigma}$, or a private neighbor, in which case it is assigned a label in $F_\omega=\set{[0]_\omega, [1]_\omega}$, or a dominated vertex, in which case it is assigned a label in $F_\rho=\set{[j]_\rho: j\in [0,2]}$. Accordingly, for a set $D\subseteq \nodes(G)$, we define the following notation.
\begin{align}
	\label{eq:nodesOfD}
	&\nodes_\omega(D)\eqdef\set{v \in \nodes(G){\setminus} D : |N(v)\cap D|=1} \\
	& \nodes_\rho(D)\eqdef\set{v \in \nodes(G){\setminus} D : |N(v)\cap D|\geq 2}\\
	& \nodes_{\sigma_I}(D)\eqdef\set{v \in D : |N(v)\cap D|=0, |N(v)\cap \nodes_\omega(D)|=0}
\end{align}
The idea behind the enumeration algorithm is that it enumerates assignments $F^{\nodes(G)}$ that represent minimal dominating sets. To make this formal, we assume a complete order $Q=(v_1,v_2,\dots,v_n)$ over the vertices $\nodes(G)$. The properties of this order are crucial for obtaining the FPL delay, and we show how to generate $Q$ and prove its properties in Section~\ref{sec:vertexOrdering}.
In what follows, the set $\set{v_1,\dots,v_i}$ represents the first $i$ vertices in the vertex-order $Q$.
\begin{definition}[Extendable Assignment]
	\label{def:extendable}
Let $\set{v_1,\dots,v_i}$ be the first $i$ vertices in the order $Q$. We say that the assignment:
$\theta^i: \set{v_1,\dots,v_i}\rightarrow F$ is \e{extendable} if there exists a minimal dominating set $D$ such that the following holds for every $v_j\in \set{v_1,\dots,v_i}$:
\begin{enumerate}
	\item $\theta^i(v_j)\in F_\sigma$ iff $v_j\in D$.
	\item $\theta^i(v_j)=\sigma_I$ iff $N(v_j)\cap D=\emptyset$ and $N(v_j)\cap \nodes_\omega(D)=\emptyset$.
	\item $\theta^i(v_j)=[p]_\sigma$ iff $|N(v_j)\cap \nodes_\omega(D)|\geq 1$ and $|N(v_j)\cap \nodes_\omega(D)\cap \set{v_{i+1},\dots,v_n}|\geq p$ where $p\in \set{0,1}$.
	\item $\theta^i(v_j)=[p]_\omega$ iff $|N(v_j)\cap D|=1$ and $|N(v_j)\cap D\cap \set{v_{i+1},\dots,v_n}|= p$ where $p\in \set{0,1}$.
	\item $\theta^i(v_j)=[p]_\rho$ iff $|N(v_j)\cap D|\geq 2$ and $|N(v_j)\cap D\cap \set{v_{i+1},\dots,v_n}|\geq p$ where $p\in \set{0,1,2}$.
\end{enumerate}
\end{definition}
We denote by $\theta^0$ the empty assignment, which is vacuously extendable.
For every $i\in [1,n]$, we denote by $\bTheta^i$ the set of extendable assignments $\theta^i:\set{v_1,\dots,v_i}\rightarrow F$. Formally:
\begin{align}
	\label{eq:extendable}
	\bTheta^i \eqdef \set{\theta^i: \set{v_1,\dots,v_i}\rightarrow F | \theta^i \text{ is extendable}} && \forall i\in [0,n]
\end{align}
\begin{example}
	Consider the simple path $v_1-v_2-v_3$. The assignment $\theta^1:\set{ v_1 \rightarrow \sigma_I}$ is extendable because the minimal dominating set $D_1=\set{v_1,v_3}$ meets the constraints of the assignment; $v_1 \in D_1$, $|N(v_1)\cap D_1|=0$, and $|N(v_1)\cap \nodes_\omega(D_1)|=|N(v_1)\cap \emptyset|=0$. On the other hand, the assignment $\theta^2: \set{ v_1\gets [0]_\sigma, v_2\gets [0]_\omega}$ is not extendable because the only minimal dominating set that contains $v_1$, and excludes $v_2$, must contain $v_3$, thus violating the constraint that $|N(v_2)\cap D \cap \set{v_3}|=0$ (see Definition~\ref{def:extendable} item 4). In fact, also $\theta^1: \set{v_1\gets [1]_\sigma}$ is not extendable for the same reason as above; $v_1$ must contain an external private neighbor, which can only be $v_2$. Since $v_2$ is not part of the solution, then $v_3$ must be included in it. But then, $v_2$ is not longer a private neighbor because both its neighbors are dominating.
\end{example}
For an assignment $\theta^i{:} \set{v_1,\dots,v_i}{\rightarrow} F$, and label $a\in F$, we denote:
\begin{align}
	\nodes_a(\theta^i)\eqdef\set{v\in \set{v_1,\dots,v_i}: \theta^i(v)=a}&& \forall a\in F
\end{align}
Also, we let $\nodes_\sigma(\theta^i)\eqdef \nodes_{[0]_\sigma}(\theta^i)\cup \nodes_{[1]_\sigma}(\theta^i)\cup \nodes_{\sigma_I}(\theta^i)$. 
\begin{lemma}
	\label{lem:assignmentsAreDS}
An assignment $\theta^n:\set{v_1,\dots,v_n}\rightarrow F$ is extendable if and only if $\nodes_\sigma(\theta^n)$ is a minimal dominating set.
\end{lemma}
\begin{proof}
If $\theta^n$ is extendable, then it expresses only constraints that can be met with respect to the empty set of vertices. Therefore, for every $v_i$, it holds that $\theta^n(v_i)\in \set{[0]_\rho,[0]_\omega,[0]_\sigma,\sigma_I}$. By Proposition~\ref{prop:DSMin}, it follows that $\nodes_\sigma(\theta^n)$ is a minimal dominating set.
If $\nodes_\sigma(\theta^n)$ is a minimal dominating set, then again, the claim follows from Proposition~\ref{prop:DSMin}.
\end{proof}

Hence, we enumerate the set $\set{\nodes_\sigma(\theta^n): \theta^n\in \bTheta^n}$ (see~\eqref{eq:extendable}). To avoid repetitions, we prove (in the Appendix) that there is a bijection between the set of minimal dominating sets of $G$, and the set of extendable assignments $\bTheta^n$.
\def\distinctnessLemma{
	Let $\theta_1^n,\theta_2^n\in \bTheta_n$. Then $\theta_1^n\neq \theta_2^n$ iff $\nodes_\sigma(\theta_1^n)\neq \nodes_\sigma(\theta_2^n)$.  
}
\begin{lemma}
	\label{lemma:distinctness}
	\distinctnessLemma
\end{lemma}

\subsection{Vertex Ordering for Enumeration}
\label{sec:vertexOrdering}
Let $\left(\T,\jointreeMapFunction\right)$ be 
a nice disjoint branch TD (Definition~\ref{def:niceDBTD}) that has undergone the preprocessing phase of Section~\ref{sec:PreprocessingForEnumeration}, and let $u_1\in \nodes(\T)$ be its root. Define $P=(u_1,\dots,u_{|\nodes(\T)|})$ to be a depth first order of $\nodes(\T)$. Thus, 
for every $i > 1$, $\parent(u_i)$ in $\T$ is a node $u_j\in\nodes(\T)$, with $j<i$ closer to the root $u_1$. Define $B:\nodes(G)\rightarrow \set{1,\dots,|\nodes(\T)|}$ to map every $v\in \nodes(G)$ to the earliest bag $u\in \nodes(\T)$ in the order $P$, such that $v \in \jointreeMapFunction(u)$. Let $Q=(v_1,\dots,v_n)$ be a complete ordering of $\nodes(G)$ consistent with $B$, such that whenever $B(v_i)<B(v_j)$, then $v_i\prec_Q v_j$; ties are broken arbitrarily.  With some abuse of notation, we denote by $B(v_i)$ both the identifier of the node in the TD, and the bag $\jointreeMapFunction(B(v_i))$.
\begin{example}
	Consider the TD in Figure~\ref{fig:TDG}. Then $B(a)=B(b)=B(c)=u$, and $v \prec_Q d$ for all $v\in \set{a,b,c}$.
\end{example}
The order $Q$ serves two purposes. First, to construct the tries  associated with the tree nodes (Section~\ref{sec:compactRepresentation}). Hence, querying a trie $A_u$ in an order compatible with $Q$ takes time $O(|\jointreeMapFunction(u)|)$. Second, as we explain, the enumeration algorithm will follow the order $Q=(v_1,\dots,v_n)$.

\def\orderNbrLem{
		Let $v_i$ be the $i$th vertex in the order $Q=(v_1,\dots,v_n)$.
	Then $N(v_i)\cap \set{v_1,\dots,v_{i-1}}\subseteq B(v_i)$.
}
\begin{lemma}
	\label{lem:orderNbr}
\orderNbrLem
\end{lemma}
\def\necessarythm{
		Let $\theta^i\in \bTheta^i$ (see~\eqref{eq:extendable}) be an extendable assignment, $c_{i+1}\in F$, and $K_{i+1}\eqdef N(v_{i+1})\cap \set{v_1,\dots,v_i}\cap B(v_{i+1})$. If the assignment $\theta^{i+1}=(\theta^i,c_{i+1})$ is extendable, then the following holds:
	\begin{align*}
		\text{if }  c_{i+1}&{=}[p]_\omega  \quad \text{ then} && |K_{i+1}{\cap} \nodes_\sigma(\theta^i)|=1{-}p \\
		\text{if }  c_{i+1}&{=}\sigma_I\quad\quad  \text{ then} &&|K_{i+1}{\cap} \nodes_\sigma(\theta^i)|=|K_{i+1}{\cap} \nodes_\omega(\theta^i)|=0 \\
		\text{if } c_{i+1}&{=}[q]_\rho \quad \text{ then} && |K_{i+1}{\cap} \nodes_{\sigma}(\theta^i)| \geq 2{-}q\\
		\text{if }  c_{i+1}&{=}[p]_\sigma \quad
		\text{ then} && |K_{i+1}{\cap} \nodes_{[0]_\omega}(\theta^i)|{=}0, |K_{i+1}{\cap} \nodes_{[1]_\omega}(\theta^i)| {\geq} 1{-}p
	\end{align*}
	in which case we say that $c_{i+1}$ is an \e{admissible} label to $v_{i+1}$.
}
The following theorem defines a necessary condition for an assignment to be extendable. Proof is deferred to the Appendix.
\begin{theorem}
	\label{thm:necessary}
\necessarythm
\end{theorem}

Theorem~\ref{thm:necessary} establishes a necessary but insufficient condition for the assignment $\theta^{i+1}=(\theta^i,c_{i+1})$ to be extendable (Definition~\ref{def:extendable}).

\noindent\textbf{Neighborhood Constraints with Respect to the Vertex Order.}
We recall that for every node $u\in \nodes(\T)$, the assignment $\FA_u:\jointreeMapFunction(u)\rightarrow F$ corresponds to a set of neighborhood constraints imposed on the vertices with respect to the solution-set (see Section~\ref{subsec:encoding}). For example, if $v_i\in \jointreeMapFunction(u)$ is the $i$th vertex in $Q$, then $\FA_u(v_i)=[1]_\rho$ implies that $|D\cap N(v_i)\cap \nodes(G_u)|\geq 1$ in every solution $D$.
For the purpose of enumeration, we wish to define these neighborhood constraints with respect to the vertices that follow $v_i$ in the order $Q$. 
For example,  the assignment $\FA_u(v_i)=[1]_\rho$ implies that $|D\cap \nodes(G_u)\cap  N(v_i)\cap \set{v_{i+1},\dots,v_n}|\geq 1$ (i.e., instead of $|D{\cap} \nodes(G_u){\cap}  N(v_i)|{\geq} 1$).
\begin{lemma}
\label{lemma:NBRConstraints}
Let $v_i$ be the $i$th vertex in the order $Q$, where $v_i{\in} \jointreeMapFunction(u)$, and $u{\in} \nodes(\T)$. Then $N(v_i){\cap} \nodes(G_u){\cap} \set{v_1,\dots,v_{i-1}}\subseteq \jointreeMapFunction(u)$.
\end{lemma}
\begin{proof}
Let $v_j \in N(v_i)\cap \nodes(G_u)\cap \set{v_1,\dots,v_{i-1}}$. By Lemma~\ref{lem:orderNbr}, it holds that $v_j\in B(v_i)$. Since $v_i \in \jointreeMapFunction(u)$, then by definition of the vertex-order $P$, $u \succeq_P B(v_i)$. Since $v_j \in \nodes(G_u)$, then $v_j$ is contained in some bag $t\in \T_u$, and by the running intersection property, $v_j$ is in every bag on the path from $B(v_i)$ to $t$. In particular, $v_j \in \jointreeMapFunction(u)$.
\end{proof}
We wish to transform the assignments $\FA_u(v_i)$ such that they represent neighborhood constraints with respect to $\nodes(G_u)\cap \set{v_{i+1},\dots,v_n}$. By Lemma~\ref{lemma:NBRConstraints}, we have that $$N(v_i)\cap \nodes(G_u) \cap \set{v_1,\dots,v_{i-1}}\subseteq N(v_i)\cap \jointreeMapFunction(u).$$
Therefore, we adjust $\FA_u(v_i)$ based on the states of $v_i$'s neighbors in $N(v_i)\cap \jointreeMapFunction(u)\cap \set{v_1,\dots,v_{i-1}}$ as follows. For $a\in F$, let  
\begin{equation*}
K_{i,a}\eqdef N(v_i)\cap \nodes_a(\FA_u) \cap \set{v_1,\dots,v_{i-1}}
\end{equation*}
(see definition of $V_a(\FA_u)$ in~\eqref{eq:AssignedaInFA_u}).
Therefore, based on $v_i$'s neighbors in $N(v_i)\cap \jointreeMapFunction(u)\cap \set{v_1,\dots,v_{i-1}}$, the label $\FA_u(v_i)$ can be one of:
	\begin{align}
	\label{eq:nbrConsts}
		&[0]_\omega &&\text{ if }&&|K_{i,\sigma}|=1, \text{ or } \nonumber\\
		&[1]_\omega &&\text{ if }&&|K_{i,\sigma}|=0, \text{ or } \nonumber \\
		&\sigma_I &&\text{ if }&&|K_{i,\sigma}|{=}|K_{i,\omega}|{=}0, \text{ or }\\
		&	[\max\set{0,1{-}|K_{i,[1]_\omega}|}]_\sigma &&\text{ if }&&|K_{i,[0]_\omega}|=0, \text{ or }\nonumber\\
		&[\max\set{0,2{-}|K_{i,\sigma}|}]_\rho&& \nonumber
\end{align}
in which case, we say that the assignment $\FA_u(v_i)$ is \e{admissible}.

\subsection{Algorithm Description}
The enumeration algorithm $\algname{EnumDS}$ is presented in Figure~\ref{fig:EnumDS}. The input at each stage $i$ of the algorithm is: (1) $v_i$: the $i$th vertex in the order $Q$, and (2) the assignment $\theta^{i-1}: \set{v_1,\dots,v_{i-1}}\rightarrow F$ representing the neighborhood constraints of vertices $\set{v_1,\dots,v_{i-1}}$ with respect to $\set{v_i,\dots,v_n}$ (see Definition~\ref{def:extendable}). 
During the $i$th iteration, the assignment $\theta^{i-1}$ is extended to $\theta^i:\set{v_1,\dots,v_i}\rightarrow F$ 
by assigning an admissible value $\theta^i(v_i):v_i \rightarrow F$ to $v_i$ (see Theorem~\ref{thm:necessary}). If $v_i$ is assigned a label in $F_\sigma$ or $F_\omega$ in the $i$th iteration, then the neighborhood constraints of $N(v_i)\cap \set{v_1,\dots,v_{i-1}}$ are updated to reflect the constraints with respect to $\set{v_{i+1},\dots,v_n}$. 
\begin{example}
Let $v_j,v_i\in \nodes(G)$ where $v_j\prec_Q v_i$. Consider the call to \algname{EnumDS} with $v_i$ and $\theta^{i-1}$ where $\theta^{i-1}(v_j)=[2]_\rho$.	If $v_i\in N(v_j)$ is added to the solution (i.e., assigned a label in $F_\sigma$), then the state of $v_j$ is updated in $\theta^i$ to $\theta^{i}(v_j)=[1]_\rho$, effectively relaxing the neighborhood constraint on $v_j$ from having at least two neighbors in the solution set among $\set{v_i,\dots,v_n}$, to having at least one neighbor in the solution set among $\set{v_{i+1},\dots,v_n}$.
On the other hand, if $\theta^{i-1}(v_j)=[1]_\omega$, then adding $v_i$ to the solution effectively forbids the addition of any vertex in $N(v_j)\cap \set{v_{i+1},\dots,v_n}$ to the solution. That is, the state of $v_j$ is updated to $\theta^i(v_j)=[0]_\omega$.
\end{example}
Let $\theta^{i-1}{:}\set{v_1,\dots,v_{i-1}}{\rightarrow} F$. We formalize the updates to the neighborhood constraints of $N(v_i){\cap}\set{v_1,\dots,v_{i-1}}$ following the assignment $v_i{\gets} c_i$ ($c_i\in F$) of an admissible label $c_i$ (see Theorem~\ref{thm:necessary}). We let $\theta^i\eqdef(\underline{\theta}^{i-1},c_i)$
where $\underline{\theta}^{i-1}:\set{v_1,\dots,v_{i-1}}\rightarrow F$ is derived from the pair $(\theta^{i-1},c_i)$ as follows. For every $j\in \set{1,\dots,i-1}$:
\begin{equation}
	\label{eq:gj}
	\underline{\theta}^{i-1}(v_j){\gets} 
	\begin{cases}
		\theta^{i-1}(v_j) & \substack{c_i{\in} F_\rho \text{, or } v_j {\notin} N(v_i), \\\text{ or }  c_i{\in} F_\sigma \text{ and } \theta^{i-1}(v_j)\in \set{[0]_\sigma,[1]_\sigma}, \\ \text{ or } c_i{\in} F_\omega \text{ and } \theta^{i-1}(v_j)\notin F_\sigma} \\	
		[0]_\sigma & c_i{\in} F_\omega\text{ and } \theta^{i-1}(v_j){\in} \set{[0]_\sigma,[1]_\sigma}\\	
		[0]_\omega & c_i{\in} F_\sigma, \theta^{i-1}(v_j)=[1]_\omega\\	
		[\max\set{0,\ell{-}1}]_\rho & c_i{\in} F_\sigma, \theta^{i-1}(v_j)=[\ell]_\rho\\
		\bot & \text{ otherwise}
		\end{cases}
\end{equation}
where $\bot\notin F$ means that the assignment $\theta^i(v)\gets c_i$ violates the neighborhood constraints of $v_j$. Specifically, this happens if $c_i\in F_\sigma$ and $\theta^{i-1}(v_j)=[0]_\omega$, or if $c_i\in F_\sigma\cup F_\omega$ and $\theta^{i-1}(v_j)=\sigma_I$.
We denote by $\theta_{|B(v_i)}$  the assignment $\theta$ restricted to the vertices $\jointreeMapFunction(B(v_i))$.
\def\lemcorrectness{
	Let $\theta^{i-1}\in \bTheta^{i-1}$ be an extendable assignment, $c_i \in F$ be an admissible value as in Theorem~\ref{thm:necessary}, and $K_i\eqdef B(v_i)\cap \set{v_1,\dots,v_i}$. The assignment $\theta^i=(\underline{\theta}^{i-1},c_i)$ (see~\eqref{eq:gj}) is extendable if and only if $A_{B(v_i)}(\theta^i_{|K_i})=\true$ (or the corresponding node appears in the trie $A_{B(v_i)}$).
}
\begin{lemma}
	\label{lem:correctness}
\lemcorrectness
\end{lemma}

Algorithm \algname{EnumDS} is initially called with $v_1$, and the empty assignment $\theta^0$, which is vacuously extendable. 

If $i=n+1$, then the algorithm prints the minimal dominating set $\nodes_\sigma(\theta^n)$ (line~\ref{line:printS}).
Since the algorithm processes the vertices by order $Q$, then $\nodes_\sigma(\theta^{i-1})\cup \nodes_\omega(\theta^{i-1}) \subseteq \set{v_1,\dots,v_{i-1}}$.
By lemma~\ref{lem:orderNbr}, $N(v_i)\cap \set{v_1,\dots,v_{i-1}}\subseteq B(v_i)$. Therefore,
$n_i\eqdef|L_i \cap \nodes_\sigma(\theta^{i-1})|=|N(v_i)\cap \nodes_\sigma(\theta^{i-1})|$,
where $L_i\eqdef B(v_i)\cap \set{v_1,\dots,v_i}\cap N(v_i)$ (see line~\ref{line:EnumKi}).
Similarly, $m_i=|L_i \cap \nodes_\omega(\theta^{i-1})|$.
Since $|B(v_i)| \leq w$ for all $i\in [1,n]$, the calculation in lines~\ref{line:EnumKi}-\ref{line:Enummi} takes time $O(w)$, where $w=\tw(G)$.
Then, the algorithm identifies the set of admissible labels for $v_i$ according to Theorem~\ref{thm:necessary}. This happens in lines~\ref{line:cSigmaI},\ref{line:cSigma},\ref{line:cOmega}, and~\ref{line:cRho}.
For example, by Theorem~\ref{thm:necessary}, if $v_i$ is assigned label $\sigma_I$, then it must hold that $n_i=m_i=0$ (line~\ref{line:cSigmaI}).
Following this step, the algorithm verifies that the 
resulting assignment $\theta^i=(\underline{\theta}^{i-1},a)$ is extendable (Definition~\ref{def:extendable}). 
By Lemma~\ref{lem:correctness}, this is if and only if $A_{B(v_i)}(\theta^i_{|K_i})=\true$ (line~\ref{line:testBeforeRecurse}). 
For that, the algorithm makes use of the memoisation trie $A_{B(v_i)}$ (Section~\ref{sec:compactRepresentation}), and recurses with the updated assignment $\theta^i$ in line~\ref{line:recurse}.

\subsection{Proof of Correctness and Runtime Analysis}
\label{sec:ProofCorrectness}
We show that the algorithm prints the minimal dominating sets of $G$, denoted $\sch(G)$, in succession, without repetition, with delay $O(nw)$.
\iffull
If the dynamic-programming algorithm of the preprocessing phase established that $\sch(G)=\emptyset$, then we simply do not activate the enumeration algorithm. So, assume that $\sch(G)\neq \emptyset$. 
\fi

\def\corretnessOnlyDS{
Let $G$ be a graph where $\sch(G)\neq \emptyset$. For every $i\in [1,n+1]$, $\algname{EnumDS}$ is called with the pair $(v_i,\theta^{i-1})$ if and only if $\theta^{i-1} \in \bTheta^{i-1}$ (i.e., $\theta^{i-1}$ is extendable).
}
\begin{lemma}
\label{lem:corretnessOnlyDS}
\corretnessOnlyDS
\end{lemma}
\begin{prf}
	By induction on $i$. 
	Since $\theta^{i-1}$ is extendable, then by Theorem~\ref{thm:necessary}, we have that each one of the states $c_i$ considered in line~\ref{line:forall} is admissible. Since $\theta^{i-1}$ is extendable, then the assignment $\theta^i\eqdef(\underline{\theta}^{i-1},c_i)$ is extendable for at least one such admissible label $c_i$.
	By Lemma~\ref{lem:correctness}, the assignment $\theta^i\eqdef(\underline{\theta}^{i-1},c_i)$ generated according to~\eqref{eq:gj} is extendable if and only if the condition in line~\ref{line:testBeforeRecurse} holds, in which case the algorithm indeed recurses in line~\ref{line:recurse}. Consequently, we have that $\theta^i\in \bTheta^i$. In other words, if the algorithm recurses with assignment $\theta^i$ in line~\ref{line:recurse}, then $\theta^i$ is extendable.
	
	By the induction hypothesis, the algorithm is called with the pair $(\theta^{i-1},v_i)$ for every $\theta^{i-1}\in \bTheta^{i-1}$.  
	By Theorem~\ref{thm:necessary}, if the assignment $\theta^i\eqdef (\underline{\theta}^{i-1},c_i)$ is extendable, then $c_i$ is admissible. The algorithm attempts to extend $\theta^{i-1}$ with all possible admissible labels (line~\ref{line:forall}). Hence, it recurses with all assignments $\theta^i \in \bTheta^i$.
\end{prf}

\begin{theorem}
	$\algname{EnumDS}$ prints the set $S \subseteq \nodes(G)$ iff $S\in \sch(G)$, and it prints the vertex-sets in $\sch(G)$ without repetition.
\end{theorem}
\begin{prf}
By Lemma~\ref{lem:corretnessOnlyDS}, $\algname{EnumDS}$ is called with $(v_{n+1},\theta^n)$ exactly once for every $\theta^n \in \bTheta_n$. The result follows from Lemmas~\ref{lem:assignmentsAreDS} and~\ref{lemma:distinctness}, which establish the bijection between $\sch(G)$ and $\set{\nodes_\sigma(\theta^n): \theta^n \in \bTheta^n}$.
\end{prf}
\renewcommand{\algorithmicrequire}{\textbf{Input:}}
\begin{algseries}{H}{Algorithm for enumerating minimal dominating sets in FPL-delay when parameterized by treewidth. \label{fig:EnumDS}}
	\begin{insidealg}{EnumDS}{$v_i$, $\theta^{i-1}$}		
		\REQUIRE{$v_i$: the $i$th vertex in the order $Q$, and $\theta^{i-1}:\set{v_1,\dots,v_{i-1}\rightarrow F}$}
		\IF{$i=(n+1)$}
			\STATE print $\nodes_\sigma(\theta^{i-1})$ \label{line:printS}	
			\RETURN
		\ENDIF
		\STATE $L_i\gets B(v_i)\cap \set{v_1,\dots,v_{i-1}}\cap N(v_i)$ \label{line:EnumKi}
		\STATE $n_i \gets |L_i\cap \nodes_\sigma(\theta^{i-1})|$ \label{line:Enumni}
		\STATE $m_i \gets |L_i \cap \nodes_\omega(\theta^{i-1})|$ \label{line:Enummi}
		\STATE $c_\omega \gets \bot$, $c_\sigma \gets \bot$, $c_{\sigma_I}\gets \bot$
		\IF{$n_i=0$ and $m_i=0$}
			\STATE $c_{\sigma_I}\gets \sigma_I$ \label{line:cSigmaI}
		\ENDIF
		\IF{$|L_i \cap \nodes_{\sigma_I}(\theta^{i-1})|{=}0$ and $|L_i\cap \nodes_{[0]_\omega}(\theta^{i-1})|{=}0$}
			\STATE $c_\sigma \gets [\max\set{0,1-m_i}]_\sigma$  \label{line:cSigma}
		\ENDIF
		\IF{$|L_i \cap \nodes_{\sigma_I}(\theta^{i-1})|=0$ and $n_i\leq 1$}
			\STATE $c_\omega \gets [1-n_i]_\omega$ \label{line:cOmega}
		\ENDIF
		\STATE $c_\rho \gets [\max\set{0,2-n_i}]_\rho$
		\label{line:cRho}
		\STATE $K_i \gets B(v_i)\cap \set{v_1,\dots,v_{i}}$
		\FORALL{$a \in \set{c_\sigma,c_\rho,c_\omega,c_{\sigma_I}}\setminus \set{\bot}$} \label{line:forall}
			\STATE update $\theta^{i-1}$ to $\underline{\theta}^{i-1}$ according to~\eqref{eq:gj}
			 \label{line:update}
			 \STATE $\theta^i\gets (\underline{\theta}^{i-1},v_i \gets a)$
			\IF{$A_{B(i)}(\theta^i_{|K_i})$} \label{line:testBeforeRecurse}
				\STATE $\algname{EnumDS}(v_{i+1},\theta^i)$ \label{line:recurse}
			\ENDIF
		\ENDFOR	 	
	\end{insidealg}
\end{algseries}

\begin{proposition}
\label{prop:runtime}
Given a preprocessed disjoint branch TD $(\T,\jointreeMapFunction)$, the delay until the output of the first solution, and between the output of consecutive minimal dominating sets, is $O(nw)$.
\end{proposition}
\begin{prf}
	All operations performed in lines~\ref{line:EnumKi}-\ref{line:Enummi} take time $O(w)$.
Updating $\theta^{i-1}$ to $\underline{\theta}^{i-1}$ according to~\eqref{eq:gj} in line~\ref{line:update}, takes time $O(w)$ because, by Lemma~\ref{lem:orderNbr}, $N(v_i)\cap \set{v_1,\dots,v_{i-1}}\subseteq B(v_i)$, and $|B(v_i)|\leq w$. For every $i\in [1,n]$, the trie $A_{B(v_i)}$ is queried at most four times: when assigning $v_i$ a label from the set $\set{c_\sigma,c_\rho,c_\omega,c_{\sigma_I}}$ (line~\ref{line:forall}). Each such query takes time $O(w)$ (see Section~\ref{sec:compactRepresentation}). 
\end{prf}

}

%% file: conclusion.tex
\section{Conclusion}
\label{sec:conclusion}
We present an efficient algorithm for enumerating minimal hitting sets parameterized by the treewidth of the hypergraph. We introduced the nice disjoint-branch TD suited for this task and potentially useful for other problems involving the enumeration of minimal or maximal solutions with respect to some property, common in data management.
As future work, we intend to publish an implementation of our algorithm and empirically evaluate it on datasets for data-profiling problems~\cite{DBLP:series/synthesis/2018Abedjan,DBLP:journals/vldb/KossmannPN22}.\eat{, including the discovery of Unique Column Combinations (UCCs) and functional dependencies.} We also plan to investigate ranked enumeration of minimal hitting sets parameterized by treewidth.

%% file: ProofsFromSection3.tex
\section{Proofs From Section~\ref{sec:Preliminaries}}
\begin{repproposition}{\ref{prop:transversalCover}}
	\transversalCover
\end{repproposition}
\begin{proof}
	Recall that $\nodes(\H^d)\eqdef\set{y_e: e\in \Hedges(\H)}$, and $\Hedges(\H^d)\eqdef \set{f_v: v\in \nodes(\H)}$, where $f_v \eqdef \set{y_e\in \nodes(\H^d): v\in e}$.
	Let $f_v\in \Hedges(\H^d)$. Since $\D$ is an edge cover of $\H$, then there exists a hyperedge $e\in \D$ such that $v\in e$. By definition, this means that $y_e\in f_v$, and hence $y_e\in f_v\cap \set{y_e: e\in \D}$. Therefore, the set $\set{y_e: e\in \D}$ is a transversal of $\H^d$. We now show that $\set{y_e: e\in \D}$ is a minimal transversal of $\H^d$. Suppose it is not, and that $\set{y_e: e\in \D'}$ is a transversal of $\H^d$, where $\D'\eqdef \D{\setminus}\set{e}$  for some $e\in \D$. That is, $f_v\cap \set{y_e: e\in \D'}\neq \emptyset$ for all $v\in \nodes(\H)$. But this means that $\bigcup_{e\in \D'}e=\nodes(\H)$, and hence $\D'\subset \D$ is an edge cover; a contradiction.
	
	For the other direction, we first show that $\D$ is an edge cover of $\H$. Let $v\in \nodes(\H)$. Since $\set{y_e: e\in \D}$ is a transversal of $\H^d$, then $f_v \cap \set{y_e: e\in \D} \neq \emptyset$. Therefore, there exists a vertex $y_e\in\set{y_e: e\in \D}$ such that $y_e\in f_v$. By construction, this means that $v\in e$. In particular, $v\in \bigcup_{e\in \D}e$, and hence $\D$ is an edge cover of $\H$. If $\D$ is not minimal, then there exists a $\D'\subset \D$ that is an edge cover of $\H$. By the previous, this means that $\set{y_e: e\in \D'}\subset \set{y_e: e\in \D}$ is a transversal of $\H^d$. But this is a contradiction to the minimality of the transversal $\set{y_e: e\in \D}$.
\end{proof}

\section{Proofs From Section~\ref{sec:TransEnumReduction}}
\begin{reptheorem}{\ref{thm:domenum}}
\thmdomenum
\end{reptheorem}
\begin{proof}	
	If $\M$ is a transversal, then 
	every vertex in $\set{y_e : e \in \Hedges(\H)}$ is adjacent to at least one vertex in $\M$. By the assumption of the Theorem, $\nodes(\H)\setminus \M\neq \emptyset$. Since the vertices in $\nodes(\H)\setminus \M$ are dominated by $v$, we conclude that $\M \cup \set{v}$ is a dominating set of $B(\H)$. Since $\M$ is a minimal transversal, then for any $\M' \subset \M$, it holds that $e \cap \M'=\emptyset$ for some hyperedge $e \in \Hedges(\H)$. But this means that $N(y_e) \cap \M'=\emptyset$ for $y_e \in \nodes(B(\H))$. Hence, $\M\cup \set{v}$ is a minimal dominating set of $\B(\H)$.
	
	If $\M\cup \set{v}$ is a minimal dominating set of $B(\H)$, then $\M$ must be a minimal transversal of $\H$. If not, then let $\M'\subset \M$ be a transversal of $\H$. But then, by the previous, $\M' \cup \set{v}\subset \M\cup\set{v}$ is a dominating set of $B(\H)$; a contradiction.
\end{proof}

\begin{replemma}{\ref{lemma:addvertexwidth}}
	\lemmaaddvertexwidth
\end{replemma}
\begin{proof}
	Let $(\T,\jointreeMapFunction')$ be the tree that results from $(\T,\jointreeMapFunction)$ by adding vertex $v$ to all bags of $(\T,\jointreeMapFunction)$. In other words, for every $u\in \nodes(\T)$, we let $\jointreeMapFunction'(u)\eqdef\jointreeMapFunction(u)\cup \set{v}$. Since $(\T,\jointreeMapFunction)$ is a tree-decomposition for $G$, then clearly $(\T,\jointreeMapFunction')$ meets the three conditions in Definition~\ref{def:TreeDecomposition} with respect to $G'$. Therefore, $(\T,\jointreeMapFunction')$ is a tree-decomposition for $G'$. Since the size of every bag increased by exactly one vertex, the width of $(\T,\jointreeMapFunction')$ is $k+1$.
\end{proof}

%% file: AppendixOverviewDetails.tex
\section{Missing Details from Section~\ref{sec:overview}}
\label{sec:AppendixOverview}

\begin{repproposition}{\ref{prop:DSMin}}
	\propDSMin
\end{repproposition}
\begin{proof}
	Let $D$ be a minimal dominating set, and let $u\in D$. Since $D'\eqdef D{\setminus}\set{u}$ is no longer dominating, there is a vertex $v\in \nodes(G){\setminus} D'$ such that $N(v)\cap D'=\emptyset$. If $v=u$, then $N(u)\cap D'= N(u)\cap D=\emptyset$. If $v\neq u$, then $v\in \nodes(G){\setminus} D$. Since $N(v)\cap D'=\emptyset$, and $N(v)\cap D\neq \emptyset$, we conclude that $N(v)\cap D=\set{u}$.
	
	Now, let $D\subseteq \nodes(G)$ be a dominating set for which the conditions of the proposition hold, and let $D'\eqdef D{\setminus} \set{u}$ for a vertex $u\in D$. If $D'$ is a dominating set, then $N(u)\cap D'=N(u)\cap D\neq \emptyset$. Also, $|N(v)\cap D'|\geq 1$ for every $v\in N(u){\setminus}D'$. In particular, $|N(v)\cap D|\geq 2$ for every $v\in N(u){\setminus}D$. But this brings us to a contradiction.  
\end{proof}

\begin{reptheorem}{\ref{thm:correctness}}
	\correctnessthm
\end{reptheorem}
\begin{proof}
	By induction on $i$. 
	The algorithm is initially called with the pair $(\theta^0,v_1)$, where $\theta^0$ is the empty labeling. Since $\bTheta^0\eqdef \set{\theta^0}$, the claim follows.
	
	Let $i\geq 2$. We now assume the claim holds for all $k\leq i-1$, and prove the claim for $i$. 
	We first show that if the algorithm is called with the pair $(\theta^i,v_{i+1})$ in line~\ref{line:recurseShort}, then $\theta^i\in \bTheta^i$. This immediately follows from the fact that the algorithm is called with the pair $(\theta^i,v_{i+1})$ in line~\ref{line:recurseShort} only if the procedure $\algname{IsExtendable}(\M,\theta^i)$ returned \texttt{true} in line~\ref{line:isExtendable}. Therefore, $\theta^i\in \bTheta^i$.
	
	We now show that \algname{EnumDS} is called with the pair $(\theta^i,v_{i+1})$, for every $\theta^i\in \bTheta^i$. So let $\theta^i\in \bTheta^i$. By definition of extendable assignment, there exists a minimal dominating set $D\in \sch(G)$ such that $D$ is consistent with $\theta^i$. That is, for every $v\in V_i$, the label that $D$ induces on $v$ is $\theta^i(v)$ (Definition~\ref{def:induceLabels}). Now, let $\theta^{i-1}:V_{i-1}\rightarrow F$ denote the assignment that $D$ induces on the vertex set $V_{i-1}=\set{v_1,\dots,v_{i-1}}$. By definition, $D\in \sch(G)$ is consistent with $\theta^{i-1}$, and hence $\theta^{i-1}\in \bTheta^{i-1}$. By the induction hypothesis, \algname{EnumDS} is called with the labeling $\theta^{i-1}\in \bTheta^{i-1}$. Let $c_i\eqdef \theta^i(v_i)$. Since $D$ is consistent with $\theta^i$, and with $\theta^{i-1}$, and since for every $j\leq n$, the assignment $\theta^j$ that $D$ induces on $V_j$ is unique, then it must hold that $\theta^i\in \algname{IncrementLabeling}(\theta^{i-1},c_i)$.
	\eat{
		
		***************************************
		Since $\theta^{i-1}$ is extendable, then by definition, there exists a label $c_i \in F$ such that $\theta^i\in \bTheta^i$ where $\theta^i\gets \algname{IncrementLabeling}(\theta^{i-1},c_i)$, and $\theta^u(v_i)=c_i$. 
		
		by Theorem~\ref{thm:necessary}, we have that each one of the states $c_i$ considered in line~\ref{line:forall} is admissible. Since $\theta^{i-1}$ is extendable, then the assignment $\theta^i\eqdef(\underline{\theta}^{i-1},c_i)$ is extendable for at least one such admissible label $c_i$.
		By Lemma~\ref{lem:correctness}, the assignment $\theta^i\eqdef(\underline{\theta}^{i-1},c_i)$ generated according to~\eqref{eq:gj} is extendable if and only if the condition in line~\ref{line:testBeforeRecurse} holds, in which case the algorithm indeed recurses in line~\ref{line:recurse}. Consequently, we have that $\theta^i\in \bTheta^i$. In other words, if the algorithm recurses with assignment $\theta^i$ in line~\ref{line:recurse}, then $\theta^i$ is extendable.
		
		By the induction hypothesis, the algorithm is called with the pair $(\theta^{i-1},v_i)$ for every $\theta^{i-1}\in \bTheta^{i-1}$.  
		By Theorem~\ref{thm:necessary}, if the assignment $\theta^i\eqdef (\underline{\theta}^{i-1},c_i)$ is extendable, then $c_i$ is admissible. The algorithm attempts to extend $\theta^{i-1}$ with all possible admissible labels (line~\ref{line:forall}). Hence, it recurses with all assignments $\theta^i \in \bTheta^i$.
	}
\end{proof}

\subsection{Proofs from Section~\ref{sec:ordering}}
\label{sec:orderingAppendix}
\begin{replemma}{\ref{lem:orderNbr}}
	\orderNbrLem
\end{replemma}
\begin{proof}
	Let $v_k\in N(v_i)\cap \set{v_1,\dots,v_{i-1}}$.
	Since $v_k \prec_Q v_i$, then $B(v_k)\leq B(v_i)$. If $B(v_k)= B(v_i)$, then $v_k\in B(v_i)$, and we are done. Otherwise,  $B(v_k)< B(v_i)$, and hence $B(v_k)\notin \nodes(\T_{B(v_i)})$.
	Since $v_i \in N(v_k)$, then by Definition~\ref{def:TreeDecomposition} there exists a node $l\in \nodes(\T)$ such that $v_i,v_k\in \jointreeMapFunction(l)$. Since $B(v_i)$ is the first bag (in the depth-first-order $P$) that contains $v_i$, then $l \geq B(v_i)$.
	By the running intersection property, $v_k$ appears in every bag on the path between nodes $B(v_k)$ and $l$ in $\T$, and all bags that contain $v_i$ are descendants of $B(v_i)$ in $\T$. Therefore, node $l$ belongs to $\T_{B(v_i)}$. Hence, $v_k\in \jointreeMapFunction(l)$ belongs to a bag $\jointreeMapFunction(l)$ in the subtree rooted at $B(v_i)$, and to a bag $B(v_k)\notin \nodes(\T_{B(v_i)})$. By the running intersection property, $v_k\in B(v_i)$.
\end{proof}

\subsection{Proofs from Section~\ref{sec:incrementLabeling}}
\begin{repproposition}{\ref{prop:increment}}
	\incrementProposition
\end{repproposition}
\begin{proof}
	Since $\theta^i\in \bTheta^i$ is extendable, then there is a minimal dominating set $D\in \sch(G)$ that is consistent with $\theta^i$. By Definition~\ref{def:induceLabels}, it holds that $\nodes_\sigma(\theta^i)\subseteq D$, for every $v_j\in \nodes_{\sigma_I}(\theta^i)$ it holds that $N(v_j)\cap D=\emptyset$, for every $v_j\in \nodes_\omega(\theta^i)$ it holds that $|N(v_j)\cap D|=1$, and for every $v_j\in \nodes_\rho(\theta^i)$ that $|N(v_j)\cap D|\geq 2$.
	
	Suppose, by contradiction, that item~$(1)$ does not hold. Then there is a vertex $w\in N(v_j)\cap V_i$ such that $\theta^i(w)\in F_\omega \cup F_\sigma$. If $\theta^i(w)\in F_\sigma$, then $N(v_j)\cap D\neq \emptyset$, which means that $D$ is not consistent with $\theta^i$, a contradiction.
	If $\theta^i(w)\in F_\omega$, then by Definition~\ref{def:induceLabels}, it holds that $N(w)\cap D=\set{v_j}$. But this means that $v_j$ has a private neighbor in $\nodes(G){\setminus}D$, which again brings us to a contradiction. 
	The rest of the items of the proposition are proved in a similar fashion, as a direct consequent of Definition~\ref{def:induceLabels}, and are omitted.
\end{proof}

We present the procedure \algname{IncrementLabeling} in Figure~\ref{fig:IncrementLabeling}.
We remind the reader that for an assignment $\theta^i:V_i\rightarrow F$, we denote by $\nodes_\sigma(\theta^i)\eqdef \set{v\in V_i:\theta^i(v) \in F_\sigma}$, and by $\nodes_\omega(\theta^i)\eqdef \set{v\in V_i:\theta^i(v) \in F_\omega}$.
In line~\ref{line:initThetai}, the algorithm initializes the assignment $\theta^i$ by combining the assignment $\theta^{i-1}:V_{i-1}\rightarrow F$ and $c_i\in F$, provided as input. That is, $\theta^i(v_i)=c_i$, and $\theta^i(v_j)=\theta^{i-1}(v_j)$ if $j<i$. The procedure then updates the labels of the vertices in $N(v_i)\cap V_i$ according to the label $c_i$. For example, if $c_i\in F_\sigma$, and $v_j\in N(v_i)\cap V_i$ where $\theta^{i-1}(v_j)=[1]_\rho$, then $v_j$'s label is updated to $[0]_\rho$ in line~\ref{line:updaterho}.
If $\theta^i(v_j)=[1]_\omega$, it means that $N(v_j)\cap \nodes_\sigma(V_{i-1})=\emptyset$. If, for example, $c_i=[0]_\sigma$, then $v_j\in N(v_i)$ is a private neighbor of $v_i$. Therefore, the label of $v_j$ is updated to $\theta^i(v_j)=[0]_\omega$ in line~\ref{line:update0Omega}, indicating that $v_j$ should not have any neighbors that belong to the solution set among $\nodes(G){\setminus}V_i$. This maintains the constraint that $v_j$ have exactly one neighbor in the solution set represented by the assignment $\theta^i$.
In lines~\ref{line:bot1},\ref{line:bot2},\ref{line:bot3}, and~\ref{line:bot4} the procedure returns $\emptyset$ indicating that one of the conditions in Proposition~\ref{prop:increment} has been violated. This means that assigning $\theta^i(v_i)\gets c_i$ will necessarily result in an assignment that is not extendable. 

\begin{replemma}{\ref{lem:incrementLabeling}}
	\incrementLabelingLemma
\end{replemma}
\begin{proof}
	Let $\theta^{i-1}: V_{i-1}\rightarrow F$ be an extendable assignment, and $c_i \in F$. If $c_i \in F_\rho$, then the only extendable assignment $\theta^{i}: V_{i}\rightarrow F$ is the one where $c_i=\max\set{0, 2-|N(v_i)\cap \nodes_\sigma(\theta^{i-1})|}$, and $\theta^i(v_j)=\theta^{i-1}(v_j)$ for all $j<i$. 
	
	If $c_i \in F_\sigma$, then the only extendable assignment $\theta^{i}: V_{i}\rightarrow F$ is the one where (1) $\theta^i(w)=\theta^{i-1}(w)$ for all $w\notin N(v_i)$,  (2) $\theta^i(w)=\theta^{i-1}(w)$ for all vertices $w\in N(v_i)$ where $\theta^{i-1}(w)\in \set{[0]_\sigma, [1]_\sigma}$, (3) $\theta^i(w)=[0]_\omega$ for all vertices $w\in N(v_i)$ where $\theta^{i-1}(w)= [1]_\omega$, and (4) $\theta^i(w)=[\max\set{0,x-1}]_\rho$ for all vertices $w\in N(v_i)$ where $\theta^{i-1}(w)= [x]_\rho$.
	If $c_i = [1]_\omega$, then the only extendable assignment $\theta^{i}: V_{i}\rightarrow F$ is the one where (1) $\theta^{i-1}(w)\notin F_\sigma$ for all $w\in N(v_i)$, and (2) $\theta^i(w)=\theta^{i-1}(w)$ for all $w\in 
	V_i$.
	If $c_i = [0]_\omega$, and $\theta^{i-1}$ is extendable with $c_i = [0]_\omega$, then there exists a single vertex $w\in N(v_i)\cap V_i$ such that $\theta^{i-1}(w)=[1]_\sigma$. In this case, there exist two extendable assignments $\theta^{i}, f^i: V_{i}\rightarrow F$. In the first, $\theta^i(w)=[0]_\sigma$ (line~\ref{line:update0Sigma}), and in the second $f^i(w)=[1]_\sigma$ (line~\ref{line:update1Sigma}). Also, for every $v\in V_i{\setminus}\set{w}$, it holds that $\theta^i(v)=f^i(v)=\theta^{i-1}(v)$.	

\end{proof}

\begin{replemma}{\ref{lem:incrementLabelingRuntime}}
	\incrementLabelingRuntime
\end{replemma}
\begin{proof}
	Observe that the runtime of \algname{IncrementLabeling} with input $(\theta^{i-1},c_i)$ is on $O(|N(v_i)\cap V_i|)$.
	If $V_i$ are the first $i$ vertices of $Q=\langle v_1,\dots,v_n\rangle$ defined in Section~\ref{sec:ordering}, then by Lemma~\ref{lem:orderNbr} it holds that $N(v_i) \cap V_i \subseteq B(v_i)$. Since $|B(v_i)|\leq \tw(G)+1$, then $|N(v_i)\cap B(v_i)|\leq \tw(G)$, and the claim follows.
\end{proof}

\renewcommand{\algorithmicrequire}{\textbf{Input:}}
\renewcommand\algorithmicensure{\textbf{Output:}}
\begin{algserieswide}{h}{Algorithm for incrementing a labeling $\theta^{i-1}\in \bTheta^{i-1}$, to a set of at most two labelings $\theta^i$, where $\theta^i(v_i)=c_i$, if possible. Otherwise, return an empty set.\label{fig:IncrementLabeling}}
	\begin{insidealgwide}{IncrementLabeling}{$\theta^{i-1}$,$c_i$}		
		\REQUIRE{A labeling $\theta^{i-1}:V_{i-1}\rightarrow F$, and $c_i \in F$}
		\ENSURE{A set of at most two labelings $\theta^i, f^i: V_i \rightarrow F$}
		\STATE $K_i\gets N(v_i)\cap V_i$
		\STATE $N_i\gets K_i\cap \nodes_\sigma(\theta^{i-1})$, $W_i\gets K_i \cap \nodes_\omega(\theta^{i-1})$
		\STATE $\theta^i \gets(\theta^{i-1},c_i)$ \label{line:initThetai}
		\IF{$c_i \in F_\sigma$}
			\FORALL{$v\in K_i$ s.t. $\theta^{i-1}(v)=[x]_\rho$}
				\STATE $\theta^{i}(v)\gets [\max\set{0,x-1}]_\rho$ \label{line:updaterho}
			\ENDFOR
		\ENDIF
		\IF{$c_i=\sigma_I$ AND $(N_i\neq \emptyset \text{ OR } W_i\neq \emptyset)$}
			 \RETURN $\emptyset$
			 \label{line:bot1}
		\ENDIF
		\IF{$c_i\in \set{[0]_\sigma,[1]_\sigma}$}
			\IF{$\left(\exists w\in K_i \text{ s.t. }\theta^{i-1}(w)\in \set{\sigma_I,[0]_\omega}\right)$ OR $\left(c_i=[0]_\sigma \text{,}W_i= \emptyset\right)$ }
				\RETURN $\emptyset$
				\label{line:bot2}
			\ELSE
				\FORALL{$w\in K_i$ s.t. $\theta^{i-1}(w)=[1]_\omega$}
					\STATE $\theta^i(w)\gets [0]_\omega$ \label{line:update0Omega}
				\ENDFOR
			\ENDIF		
		\ENDIF
		\IF{$c_i =[x]_\rho$ AND $x\neq \max\set{0,2-|N_i|}$}
			\RETURN $\emptyset$
			\label{line:bot3}
		\ENDIF
		\STATE $f^i \gets \theta^i$
		\IF{$c_i\in \set{[0]_\omega,[1]_\omega}$}
			\IF{$\left(\exists w\in K_i \text{ s.t. }\theta^{i-1}(w)=\sigma_I\right) \text{ OR } (|N_i|{\geq} 2)$ \\ $ \text{ OR }(N_i{=}\emptyset, c_i{=}[0]_\omega)  \text{ OR } (N_i {\neq} \emptyset, c_i{=}[1]_\omega)$} \label{line:omegaIfCondition}
				\RETURN $\emptyset$
				\label{line:bot4}
			\ELSIF{$c_i=[0]_\omega$}
				\STATE $\set{v}\gets N_i$ \COMMENT{In this case $|N_i|=1$}
				\IF{$\theta^i(v) = [0]_\sigma$}
					\RETURN $\emptyset$
					\label{line:bot5}
				\ELSE
					\STATE $\theta^i(v)\gets [0]_\sigma$ \label{line:update0Sigma}
					\STATE $f^i(v)\gets [1]_\sigma$ \label{line:update1Sigma}
				\ENDIF
			\ENDIF
		\ENDIF
		\RETURN $\set{\theta^i,f^i}$
	\end{insidealgwide}
\end{algserieswide}

%% file: preprocessing8States.tex
\section{Details and Proofs from Section~\ref{sec:PreprocessingForEnumeration}}
\label{sec:AppendixProofsPreproc}
\eat{
\subsection{Preprocessing Phase: Dynamic Programming over Nice-Disjoint-Branch TDs}
We show how to compute $\M_r(\FA_r)$ for every $\FA_r \in \FAs{r}$ and every $r\in \nodes(\T)$ by dynamic programming over the nice disjoint branch TD $(\T,\jointreeMapFunction)$. 
\blue{
	The preprocessing phase receives as input a nice disjoint-branch TD $(\T,\jointreeMapFunction)$, and a set of local constraints $\set{\kappa_u: u\in \nodes(\T)}$ defined in items~\eqref{item:localConstraints1}-\eqref{item:localConstraints5} in the proof of Lemma~\ref{lem:DBJTLemma} in Section~\ref{sec:buildDBJTAppendix}. Specifically, $\kappa_u: F^{\jointreeMapFunction(u)}\rightarrow \set{0,1}$ assigns zero to all assignments $\FA_u:\jointreeMapFunction(u) \rightarrow F$ that do not meet the constraints generated during the transformation to a nice disjoint-branch TD.
}

Let $A\subset \jointreeMapFunction(u)$, and $B=\jointreeMapFunction(u){\setminus}A$. Notation-wise, for an assignment $\FA_u\in \FAs{u}$, we sometimes write $\FA_u=(\FA_A,\FA_B)$ where $\FA_A$ and $\FA_B$ are the restriction of assignment $\FA_u$ to the disjoint vertex-sets $A$ and $B$ respectively, where $A{\cup}B=\jointreeMapFunction(u)$.

Let $u\in \nodes(\T)$, $\FA_u:\jointreeMapFunction(u)\rightarrow F$. We recall that
\begin{align}
	\label{eq:nodecategoriesAssignment}
	\nodes_a(\FA_u)\eqdef \set{w\in \jointreeMapFunction(u): \FA_u(w)=a} && \text{ for every }a\in F\\
	\nodes_\sigma(\FA_u)\eqdef \set{w\in \jointreeMapFunction(u): \FA_u(w)\in F_\sigma} \\
	\nodes_\omega(\FA_u)\eqdef \set{w\in \jointreeMapFunction(u): \FA_u(w)\in F_\omega}
\end{align}
Now, let $v\in \jointreeMapFunction(u)$. For every $a\in F$, we denote by $k_{v,a}(\FA_u)\eqdef |N(v)\cap \nodes_a(\FA_u)|$. Also, $k_{v,\sigma}(\FA_u)\eqdef |N(v)\cap \nodes_\sigma(\FA_u)|$, and $k_{v,\omega}(\FA_u)\eqdef |N(v)\cap \nodes_\omega(\FA_u)|$.

\noindent \textbf{Leaf Node}: $r$ is a leaf node. Since $\jointreeMapFunction(r)=\emptyset$, the only possible solution for this subtree is $\emptyset$, and hence $\M_r\equiv 1$.

\noindent \textbf{Introduce node}: Let $r \in \nodes(\T)$ be an introduce node with child node $u$, and let $\jointreeMapFunction(r)=\jointreeMapFunction(u) \cup \set{v}$, where $v\notin \jointreeMapFunction(u)$. Therefore, for $\FA_r\in \FAs{r}$, we write $\FA_r=(\FA_u,c_v)$, where $\FA_u\in \FAs{u}$ and $c_v \in F$. The label $\FA_r(w)$ of every $w\in \jointreeMapFunction(r)$ puts a constraint on the labels of $w$'s neighbors in $N(w)\cap (\nodes_r{\setminus}\jointreeMapFunction(r))$. However, since $\nodes_r=\nodes_u\cup \set{v}$, then for every $w\in \jointreeMapFunction(r)$, it holds:
\begin{align*}
N(w)\cap (\nodes_r{\setminus}\jointreeMapFunction(r))&=N(w)\cap ((\nodes_u\cup \set{v}){\setminus} (\jointreeMapFunction(u)\cup\set{v}))\\
&=N(w)\cap (\nodes_u{\setminus}\jointreeMapFunction(u)).
\end{align*}
Therefore, $\FA_r(w)=\FA_u(w)$ for every $w\in \jointreeMapFunction(u)$.  Since $v\notin \nodes_u$, then $N(v)\cap (\nodes_r{\setminus}\jointreeMapFunction(r))=N(v)\cap (\nodes_u{\setminus}\jointreeMapFunction(u))=\emptyset$. Therefore, if any of the following assertions fail for $\FA_r(v)$, then $\M_r(\FA_r)\gets 0$.
\begin{enumerate}
	\item \label{enum:assertionsIntroduce1} $\FA_r(v) \in \set{\sigma_I, [0]_\rho, [0]_\omega, [0]_\sigma}$.
	\item If $\FA_r(v)=\sigma_I$ then $k_{v,\sigma}(\FA_r)=k_{v,\omega}(\FA_r)=0$.
	\item  \label{enum:assertionsIntroduce3} If $\FA_r(v)\in \set{[0]_\sigma,[0]_\omega}$, then $k_{v,\sigma_I}=0$.
\end{enumerate} 
Then, we compute $\M_r(\FA_r)$ as follows:
\begin{align*}
	\M_r(\FA_r)=\M_r(\FA_u,c_v)=\begin{cases}
		\M_u(\FA_u)\textcolor{blue}{\wedge \kappa_r(\FA_r)} & \substack{\FA_r(v)  \text{ meets conditions} \\ \text{(1)-(3) above}}\\
		0 & \text{otherwise}
		\end{cases}
\end{align*}
\blue{The conjunction with $\kappa_r(\FA_r)$ accounts for the local constraints derived from the transformation to a disjoint-branch TD.}
\noindent \textbf{Disjoint Join Node}:
Let $r\in \nodes(\T)$ be a disjoint join node with child nodes $u_0$ and $u_1$, where $\jointreeMapFunction(r)=\jointreeMapFunction(u_0)\cup \jointreeMapFunction(u_1)$ and $\jointreeMapFunction(u_0)\cap \jointreeMapFunction(u_1)=\emptyset$. This means that  $\nodes_{u_i}\cap \jointreeMapFunction(u_{1-i})=\emptyset$ for $i\in \set{0,1}$.
Let $\FA_r=(\FA_{u_0},\FA_{u_1})$. Let $w\in \jointreeMapFunction(r)$, and suppose wlog that $w\in \jointreeMapFunction(u_0)$. The label $\FA_r(w)$ puts a constraint on the labels of $w$'s neighbors in $N(w)\cap (\nodes_r{\setminus}\jointreeMapFunction(r))$. Now, if wlog $w\in \jointreeMapFunction(u_0)$, then by definition of disjoint join node, $w\notin \jointreeMapFunction(u_1)$, and hence $N(w)\cap (\nodes_r{\setminus}\jointreeMapFunction(r))=N(w)\cap (\nodes_{u_0}{\setminus}\jointreeMapFunction(u_0))$.
Therefore, similar to the case of the introduce node, we have that $\FA_r[\jointreeMapFunction(u_i)]=\FA_{u_i}$ for $i\in \set{0,1}$. Therefore, to compute $\M_r(\FA_r)$, we simply take the conjunction of $\M_{u_0}(\FA_{u_0})\wedge \M_{u_1}(\FA_{u_1})$. \blue{To account for the local constraints $\kappa_u$ derived from the transformation to a disjoint-branch TD, we get that:}
$$\M_r(\FA_r)=\M_r(\FA_{u_0},\FA_{u_1})=\M_{u_0}(\FA_{u_0})\wedge \M_{u_1}(\FA_{u_1}) \textcolor{blue}{\wedge \kappa_r(\FA_r)}.$$ 
\noindent \textbf{Forget Node}:
Let $r \in \nodes(\T)$ be a forget node with child node $u$ where $\jointreeMapFunction(u)=\jointreeMapFunction(r)\cup \set{v}$. 
We are going to compute $\M_r(\FA_r)$ as follows. 
\begin{align}
	\label{eq:newForgetNode}
	\M_r(\FA_r) =\textcolor{blue}{\kappa_r(\FA_r)\wedge}\left[\bigvee_{c_v\in F}\M_u(\FA_{r,c_v},c_v)\right]
\end{align}
\blue{The conjunction with $\kappa_r(\FA_r)$ accounts for the local constraints derived from the transformation to a disjoint-branch TD.}

We define $\FA_{r,c_v}:\jointreeMapFunction(r)\rightarrow F$ for every $c_v\in F$ as a function of $\FA_r\in F^{\jointreeMapFunction(r)}$.
Before defining the mappings $\FA_{r,c_v}:\jointreeMapFunction(r)\rightarrow F$ for every $c_v\in F$, we observe that $u$ is the first node (in depth first order) that contains the vertex $v$. Therefore, the possible labels for $v$ in every category (i.e., $F_\sigma$, $F_\omega$, and $F_\rho$) are determined by the assignment to its neighbors in $N(v)\cap \jointreeMapFunction(u)$. Therefore, before computing $\M_r(\FA_r)$ according to~\eqref{eq:newForgetNode}, we make the following assertions to $\FA_u$. In case any one of the assertions does not hold, we set $\M_u(\FA_u)\gets 0$.
\begin{enumerate}
	\item If $\FA_u(v)=\sigma_I$ then $k_{v,\omega}(\FA_u)=k_{v,\sigma}(\FA_u)=0$.
	\item If $\FA_u(v)=[j]_\rho$ then $j= \max\set{0,2-k_{v,\sigma}(\FA_u)}$.
	\item If $\FA_u(v)=[1]_\omega$ then $k_{v,\sigma}(\FA_u)=0$.
	\item If $\FA_u(v)=[0]_\omega$ then $k_{v,\sigma_I}(\FA_u)=k_{v,[1]_\sigma}(\FA_u)=0$ and $k_{v,[0]_\sigma}(\FA_u)=1$.
	\item If $\FA_u(v)=[1]_\sigma$ then $k_{v,\sigma_I}(\FA_u)=k_{v,\omega}(\FA_u)=0$.
	\item If $\FA_u(v)=[0]_\sigma$ then $k_{v,\sigma_I}(\FA_u)=k_{v,[1]_\omega}(\FA_u)=0$, and $k_{v,[0]_\omega}(\FA_u)\geq 1$.
\end{enumerate}
We now define the mappings $\FA_{r,c_v}:\jointreeMapFunction(r)\rightarrow F$ for every $c_v\in F$, as a function of $\FA_r\in F^{\jointreeMapFunction(r)}$.
\begin{align}
	\FA_{r,\sigma_I}(w)\eqdef \begin{cases}
			\FA_r(w) & w\notin N(v) \\
			[\max\set{0,(j-1)}]_\rho & \FA_r(w)=[j]_\rho \\
			\bot & \text{ otherwise }
		\end{cases}
\end{align}

\begin{align}
	\FA_{r,[0]_\sigma}(w)\eqdef \begin{cases}
		\FA_r(w) & \substack{w\notin N(v)\text{ or }\\\FA_r(w)\in \set{[0]_\sigma,[1]_\sigma}} \\
		[\max\set{0,(j-1)}]_\rho & \FA_r(w)=[j]_\rho \\
		[0]_\omega &  \FA_r(w)=[1]_\omega \\
		\bot & \text{ otherwise }
	\end{cases}
\end{align}

\begin{align}
	\FA_{r,[1]_\sigma}(w)\eqdef \begin{cases}
		\FA_r(w) & \substack{w\notin N(v)\text{ or }\\\FA_r(w)\in \set{[0]_\sigma,[1]_\sigma}} \\
		[\max\set{0,(j-1)}]_\rho & \FA_r(w)=[j]_\rho \\
		\bot & \text{ otherwise }
	\end{cases}
\end{align}

\begin{align}
	\FA_{r,[1]_\omega}(w)\eqdef \begin{cases}
		\FA_r(w) & w\notin N(v)\text{ or }\FA_r(w)\notin F_\sigma \\
		\bot & \text{ otherwise }
	\end{cases}
\end{align}

\begin{align}
	\FA_{r,[0]_\omega}(w)\eqdef \begin{cases}
		\FA_r(w) & w\notin N(v)\text{ or }\FA_r(w)\notin F_\sigma \\
		[\max\set{0,j-1}]_\sigma & \FA_r(w)=[j]_\sigma \\
		\bot & \text{ otherwise }
	\end{cases}
\end{align}
\eat{
We express every $\FA_u: \jointreeMapFunction(u) \rightarrow F$ as a pair $(\FA_r, c_v)$ where $\FA_r:\jointreeMapFunction(r)\rightarrow F$ and $c_v\in  F$. Therefore, we compute $\mu_r(\FA_r)$:
\begin{align}
	\label{eq:forgetNodeExistence}
	\mu_r(\FA_r)=\bigvee_{\substack{c_v \in
			\set{[1]_\sigma, \sigma_I,[1]_\omega,[2]_\rho}		}}\mu_u(\FA_r,c_v) 
\end{align}
Observe that since $v\notin \jointreeMapFunction(r)$, then by the running intersection property, vertex $v$ can only belong to bags in $\T_u$. Since $N(v)\subseteq \nodes(G_u)$, then when node $u$ is processed, the only possible labels for $v$ in $\FA_u$ are $\FA_u(v) \in \set{[1]_\sigma, \sigma_I,[1]_\omega,[2]_\rho}$.
}
\noindent \textbf{Root Node}: The root node $r$ is a forget node where $\jointreeMapFunction(r)=\emptyset$, its only child $u$ contains a single vertex $\jointreeMapFunction(u)=\set{v}$, and $G_r=G_u=G$. Hence, $\M_r$ is computed as in~\eqref{eq:newForgetNode}. \eat{
	Note that the same restrictions to the assignment of a label to $v$ apply.
	we can compute $\mu_r$: does there exist a $[\sigma,\rho]$ dominating set in $G$ ? For example, if $\rho$ is co-finite and $\min\set{\rho}=r$, and $\sigma$ is finite with range $[0,\ell_1]$ then:
	\begin{align}	
		\mu_r&=\mu_u([r]_\omega)\vee \mu_u([{\geq}(r+1)]_\rho)\vee \mu_u([{\leq}\ell_1,{\geq}1]_\sigma) \label{eq:orRoot}\\ \mu_r&=\arg\min\set{\mu_u([r]_\omega).\second,\mu_u([{\geq}(r+1)]_\rho).\second,\mu_u([{\leq}\ell_1,{\geq}1]_\sigma.\second} \label{eq:minRoot}
	\end{align}
	Since all the neighbors of $u$ belong to $\nodes(G_u)$, then in~\eqref{eq:orRoot} we take the disjunction of the following predicates: (1) There exists a $[\sigma,\rho]$-DS $D$ where $u\notin D$, but is a private neighbor to exactly $r$ members of $D$, (2) There exists a $[\sigma,\rho]$-DS $D$ where $u\notin D$ and $|N(u)\cap D|\geq r+1$, and (3) There exists a $[\sigma,\rho]$-DS $D$ where $u\in D$, $|N(u)\cap D|\leq \max\set{\sigma}=\ell_1$. Since we want $D$ to be minimal then it must have at least one private neighbor. The case for the optimization variant in~\eqref{eq:minRoot} is analogous. 
}\eat{
	we sum all configurations $c_v\in C_{\sigma_1}$ to assert the fact that $u$ is not redundant and has at least one private neighbor. Finally, in~\eqref{eq:sumRho} we sum all valid configurations $c_v\in C_\rho$.
}

\noindent \textbf{Runtime.} The dynamic programming algorithm over a disjoint branch TD takes time $O(nws^w)$ where $s=|F|$\eat{is the number of states}, $w$ is the width of the (disjoint branch) TD, and $n=|\nodes(G)|$.
}

\subsection{From Nice to Disjoint-Branch-Nice TD}
\label{sec:buildDBJTAppendix}

In what follows, we denote by ${\circ}$ the concatenation operator on strings. For example, letting $a$ and $b$ refer to characters in some alphabet, by $a{\circ}b$, we refer to the string $ab$. For a rooted TD $(\T,\jointreeMapFunction)$ with $m\eqdef |\nodes(\T)|$ nodes, and node $u\in \nodes(\T)$, we remind the reader that $\parent(u)$ is the unique parent of $u$ in $\T$; if $u$ is the root node of $\T$, then $\parent(u)\eqdef \texttt{nil}$. Let $P\eqdef \langle u_1,\dots,u_m \rangle$ be a depth-first order of $\nodes(\T)$. If $u$ is a join node of $(\T,\jointreeMapFunction)$, then it has exactly two child nodes $u_0$ and $u_1$, where wlog we assume that $u_0 \prec_P u_1$. We refer to $u_0$ and $u_1$ as the left and right children of $u$, respectively; in notation, $\leftchild(u)\eqdef u_0$ and $\rightchild(u)\eqdef u_1$.

Let $(\T,\jointreeMapFunction)$ be a nice TD for the graph $G$. We associate with every node $u\in \nodes(\T)$ a string, denoted $\branch(u)$ over the alphabet $\set{0,1}$ as follows. If $u$ is the root node, then $\branch(u)\eqdef \varepsilon$, where $\varepsilon$ represents the empty string. If $u$ is the unique child of its parent $\parent(u)$ in the TD, then $\branch(u)\eqdef \branch(\parent(u))$. Otherwise, since $(\T,\jointreeMapFunction)$ is a nice TD, it holds that $\parent(u)$ is a join node in $(\T,\jointreeMapFunction)$. We then set:
\begin{align*}
	\branch(u)\eqdef \begin{cases}
			\branch(\parent(u)) {\circ}0 & \text{if } u=\leftchild(\parent(u)) \\
			\branch(\parent(u)) {\circ}1 & \text{if } u=\rightchild(\parent(u)) \\
		\end{cases}
\end{align*}

Algorithm \algname{TransformToDBJT} in Figure~\ref{alg:TDToDBJT} receives as input a nice TD $(\T,\jointreeMapFunction)$ for the graph $G$, and transforms it to a nice DBJT (Definition~\ref{def:niceDBTD}) $(\T',\jointreeMapFunction')$ with the following property. 
Every node $u\in \nodes(\T)$ is also part of the new TD $(\T',\jointreeMapFunction')$, but with the following modification:
$v\in \jointreeMapFunction(u)$ if and only if $v{\circ} \branch(u) \in \jointreeMapFunction'(u)$; that is we update the name of the vertex in $\jointreeMapFunction(u)$ according to $\branch(u)$. 
If $u$ is a join node in $(\T,\jointreeMapFunction)$, then $\T'$ contains two additional nodes $u'$ and $u''$, that are both descendants of $u$ in $\T'$, where $\jointreeMapFunction'(u'')\eqdef \bigcup_{v\in \jointreeMapFunction'(u)}\set{v{\circ}0,v{\circ}1}$, and $\jointreeMapFunction'(u')\eqdef \jointreeMapFunction'(u)\cup \jointreeMapFunction'(u'')$. The path between $u$ and $u'$ in $\T'$ contains $2|\jointreeMapFunction(u)|-1$ new forget nodes, and the path between $u'$ and $u''$ in $\T'$ contains $|\jointreeMapFunction(u)|-1$ new introduce nodes. See Figure~\ref{fig:JTToDBJTIllustrationNew} for an illustration. For brevity, we denote $v{\circ}0$ ($v{\circ}1$) by $v_0$ ($v_1$).

Let $u\in \nodes(\T)$ be a join node of the nice TD $(\T,\jointreeMapFunction)$. For every vertex $v\in \jointreeMapFunction(u)$, algorithm \algname{TransformToDBJT} generates two new vertices $v{\circ}0$ and $v{\circ}1$ that are used to represent $v$'s neighbors in $\nodes(G_{u_0}){\setminus}\jointreeMapFunction(u)$ and $\nodes(G_{u_1}){\setminus}\jointreeMapFunction(u)$, respectively.  To maintain consistency with the original graph (and the original TD $(\T,\jointreeMapFunction)$), the three vertices $v,v{\circ}0$, and $v{\circ}1$, are placed in the same bag; namely $v,v{\circ}0,v{\circ}1 \in \jointreeMapFunction'(u')$, and the factor for node $u'$ is adjusted accordingly. For the precise details, see the pseudocode of algorithm \algname{TransformToDBJT} in Figure~\ref{alg:TDToDBJT}, and the proof of Lemma~\ref{lem:DBJTLemma}.
\eat{
\begin{figure}[t]
		\centering
		\begin{subfigure}[t]{0.4\linewidth} 
			\begin{tikzpicture}[node distance=0.2cm,every node/.style={scale=0.475},x=1.28cm, y=0.7cm,font=\footnotesize]
				\tikzset{vertex/.style = {draw, ellipse, font=\fontsize{11}{12}}}
				\tikzset{doublevertex/.style = {draw, double, ellipse, font=\fontsize{11}{12}}}
				
				\node[vertex, label=left:$u_0$] (a0b0) at (-0.5,6) {$ab$};
				\node[vertex, label=right:$u_1$] (a1b1) at (0.5,6) {$ab$};
				\node[vertex, label=left:$u$] (u) at (0,7) {$ab$};
				\node[vertex] (abc) at (0,8) {$abc$};
				
				\draw[-{Stealth[length=1mm]}] (abc) -- (u);
				\draw[-{Stealth[length=1mm]}] (u) -- (a0b0);
				\draw[-{Stealth[length=1mm]}] (u) -- (a1b1);
				
				\draw[thick, dotted] (0,9)--(abc);
				
				\draw   (-0.5,5.5) coordinate(a0) --
				(-0.2,4.5) coordinate (c0) --
				(-0.8,4.5) coordinate (b0) -- cycle;
				\draw[thick, dotted] (-0.5,5.5) -- (-0.5,4.5);
				\draw[-{Stealth[length=1mm]}] (a0b0) -- (a0);
				
				\draw   (0.5,5.5) coordinate(a1) --
				(0.8,4.5) coordinate (c1) --
				(0.2,4.5) coordinate (b1) -- cycle;
				\draw[thick, dotted] (0.5,5.5) -- (0.5,4.5);
				\draw[-{Stealth[length=1mm]}] (a1b1) -- (a1);	
			\end{tikzpicture}
			\caption{Part of the nice TD $(\T,\jointreeMapFunction)$, with join node $u$.}
		\label{fig:origNiceTD}
	\end{subfigure} \hspace{0.3cm}
    \begin{subfigure}[t]{0.4\linewidth} 
    	\centering
		\begin{tikzpicture}[node distance=0.2cm,every node/.style={scale=0.475},x=1.28cm, y=0.7cm,font=\footnotesize]
			\tikzset{vertex/.style = {draw, ellipse, font=\fontsize{11}{12}}}
			\tikzset{doublevertex/.style = {draw, double, ellipse, font=\fontsize{11}{12}}}			
			\node[vertex, label=left:$u_0$] (a0b0) at (-0.5,0) {$a_0b_0$};
			\node[vertex, label=right:$u_1$] (a1b1) at (0.5,0) {$a_1b_1$};
			\node[vertex, label=left:$u''$] (utt) at (0,1) {$a_0b_0a_1b_1$};
			\node[vertex, fill=yellow] (i1) at (0,2) {$ba_0b_0a_1b_1$};
			\node[vertex, label=left:$u'$] (ut) at (0,3) {$aba_0b_0a_1b_1$};
			\node[vertex,fill=pink] (f1) at (0,4) {$aba_0b_0a_1$};
			\node[vertex, fill=pink] (f2) at (0,5) {$aba_0b_0$};
			\node[vertex, fill=pink] (f3) at (0,6) {$aba_0$};
			\node[vertex, label=left:$u$] (u) at (0,7) {$ab$};
			\node[vertex] (abc) at (0,8) {$abc$};
			
			\draw[-{Stealth[length=1mm]}] (abc) -- (u);
			\draw[-{Stealth[length=1mm]}] (u) -- (f3);
			\draw[-{Stealth[length=1mm]}] (f3) -- (f2);
			\draw[-{Stealth[length=1mm]}] (f2) -- (f1);
			\draw[-{Stealth[length=1mm]}] (f1) -- (ut);
			\draw[-{Stealth[length=1mm]}] (ut) -- (i1);
			\draw[-{Stealth[length=1mm]}] (i1) -- (utt);
			\draw[-{Stealth[length=1mm]}] (utt) -- (a0b0);
			\draw[-{Stealth[length=1mm]}] (utt) -- (a1b1);
			
			\draw[thick, dotted] (0,9)--(abc);
			
			\draw   (-0.5,-0.5) coordinate(a0) --
			(-0.2,-1.5) coordinate (c0) --
			(-0.8,-1.5) coordinate (b0) -- cycle;
			\draw[thick, dotted] (-0.5,-0.55) -- (-0.5,-1.5);
			\draw[-{Stealth[length=1mm]}] (a0b0) -- (a0);
			
			\draw   (0.5,-0.5) coordinate(a1) --
			(0.8,-1.5) coordinate (c1) --
			(0.2,-1.5) coordinate (b1) -- cycle;
			\draw[thick, dotted] (0.5,-0.55) -- (0.5,-1.5);
			\draw[-{Stealth[length=1mm]}] (a1b1) -- (a1);
		\end{tikzpicture}
		\caption{Part of the disjoint branch TD $(\T',\jointreeMapFunction')$ that results from processing the join node $u$.}
	\label{fig:DBJTAfterTransformation}
	\end{subfigure}
	\caption{The result of processing join node $u\in \nodes(\T)$ in the nice TD $(\T,\jointreeMapFunction)$ in Figure~\ref{fig:origNiceTD}. Note nodes $u'$ and $u''$ in $(\T',\jointreeMapFunction')$ in Figure~\ref{fig:DBJTAfterTransformation}, the introduction of the three forget nodes (in pink), and the single introduce node (in yellow). Observe that $2\cdot |\jointreeMapFunction(u)|-1=3$, and $|\jointreeMapFunction(u)|-1=1$.}
		\label{fig:JTToDBJTIllustrationNew}
\end{figure}
}
\begin{replemma}{\ref{lem:DBJTLemma}}
	\DBJTBuildLemma_2
\end{replemma}
\begin{prf}
	Let $(\T,\jointreeMapFunction)$ be a nice TD with root node $r$.  We first set $(\T',\jointreeMapFunction')\gets (\T,\jointreeMapFunction)$, and root $\T'$ at node $r$. Let $u\in \nodes(\T)$. From lines \eqref{line:item1Bijection1}-\eqref{line:item1Bijection2}, it is apparent that $u\in \nodes(\T')$, and that $\jointreeMapFunction'(u)=\set{v{\circ}\branch(u):v\in \jointreeMapFunction(u)}$. This proves that $(\T',\jointreeMapFunction')$ possesses the first property of the Lemma; that $\nodes(\T)\subseteq \nodes(\T')$, and that there is a bijection between $\jointreeMapFunction(u)$ and $\jointreeMapFunction'(u)$, for every $u\in \nodes(\T)$.
	
	Let $u \in \nodes(\T')$  be the join node closest to the root (e.g., that violates the disjoint property, breaking ties arbitrarily) with child nodes $u_0$ and $u_1$. 
	By Definition~\ref{def:induceLabelssubtree}, the label assigned to a vertex $v\in \jointreeMapFunction'(u)$ represents constraints to the labels of $v$'s neighbors in $\nodes(G_u){\setminus}\jointreeMapFunction'(u)$. Since $\nodes(G_u){\setminus}\jointreeMapFunction'(u)=(\nodes(G_{u_0}){\setminus}\jointreeMapFunction'(u))\cup (\nodes(G_{u_1}){\setminus}\jointreeMapFunction'(u))$, then for every 
	$v\in \jointreeMapFunction'(u)$, we can partition $N_{G_u}(v){\setminus}\jointreeMapFunction'(u)$ as follows:
	\begin{align}
	\label{eq:distributeNbrsBody}
		N_{G_u}(v){\setminus}\jointreeMapFunction'(u) {=} \underbrace{(N_{G_{u_0}}(v){\setminus} \jointreeMapFunction'(u))}_{\eqdef N_{G_u}^0(v)} {\cup}\underbrace{(N_{G_{u_1}}(v){\setminus} \jointreeMapFunction'(u))}_{\eqdef N_{G_u}^1(v)}
	\end{align}
	where $N_{G_u}^0(v)\eqdef N_{G_{u_0}}(v){\setminus} \jointreeMapFunction'(u)$, and $N_{G_u}^1(v)\eqdef N_{G_{u_1}}(v){\setminus} \jointreeMapFunction'(u)$. \eat{Clearly, $(N_{G_u}(v){\cap} \jointreeMapFunction(u))\cap N_{G_u}^i(v)=\emptyset$, for $i\in \set{0,1}$.} Observe that $N_{G_u}^0(v)\cap N_{G_u}^1(v)=\emptyset$. If not, then there is a vertex $w\in N_{G_u}^0(v)\cap N_{G_u}^1(v)$. By the running intersection property of TDs, it must hold that $w\in \jointreeMapFunction'(u)$, which brings us to a contradiction.

	We update the tree $(\T',\jointreeMapFunction')$ as follows. For every $v\in \jointreeMapFunction'(u)$ we create $v_0,v_1$ two fresh copies of $v$ (these are referred to as $v{\circ}0$ and $v{\circ}1$ respectively in the pseudocode in Figure~\ref{alg:TDToDBJT}). 
	Every occurrence of $v$ in $\jointreeMapFunction'(u_0)$ ($\jointreeMapFunction'(u_1)$) is replaced with $v_0$ ($v_1$); this happens in the recursive calls in lines~\ref{line:recCall1} and~\ref{line:recCall2}.
	Following this transition, the bags $\jointreeMapFunction'(u_0)$ and $\jointreeMapFunction'(u_1)$ are disjoint. \eat{The cardinality of the bags in subtrees $\T'_{u_0}$ and $\T'_{u_1}$ remain unchanged.}
	We create a new node $u'$, and define its bag to be $\jointreeMapFunction'(u')\eqdef \jointreeMapFunction'(u)\cup \bigcup_{v\in \jointreeMapFunction'(u)}\set{v_0,v_1}$. 
	Now, we distribute the neighbors $N_{G_u}(v){\setminus}\jointreeMapFunction'(u)$ among vertices $v_0,v_1\in \jointreeMapFunction'(u')$ according to~\eqref{eq:distributeNbrsBody}, as follows. 
	For the new $v_i$ ($i{\in} \set{0,1}$), we define:
	\begin{align}
		N_{G_{u'}}(v_i)\eqdef N_{G_{u_i}}(v_i){\setminus}\jointreeMapFunction'(u')=N_{G_{u_i}}(v_i){\setminus}\jointreeMapFunction'(u)=N^i_{G_u}(v) \label{eq:DBJTProofBody1}
	\end{align} 
   and hence:
	\begin{align}
		\label{eq:DBJTProofBody2}
		N_{G_{u}}(v){\setminus}\jointreeMapFunction'(u) {=} N_{G_{u'}}(v_0) {\cup} N_{G_{u'}}(v_1)
	\end{align}
	and, by definition, the two sets of vertices $N_{G_{u'}}(v_0)$, and $N_{G_{u'}}(v_1)$ are disjoint, as in~\eqref{eq:distributeNbrsBody}.

	Consider the assignment $\FA_{u'}: \jointreeMapFunction'(u') \rightarrow F$. We can express $\FA_{u'}=(\FA_u,\FA_{u_0},\FA_{u_1})$ where $\FA_u: \jointreeMapFunction'(u)\rightarrow F$, and $\FA_{u_i}: \jointreeMapFunction'(u_i)\rightarrow F$ where $i\in \set{0,1}$. 
	We define the set of local constraints $\kappa_{u'}$ by restricting the set of assignments
	$\FA_{u'}: \jointreeMapFunction'(u') \rightarrow F$, such that the following holds:
	\begin{enumerate}
		\eat{
		\item \textcolor{red}{\sout{For every triple $(v_0,v_1,v)$ where $v\in \jointreeMapFunction'(u)$, all three vertices are assigned the same category: $\FA_{u'}(v_0),\FA_{u'}(v_1),\FA_{u'}(v)\in F_a$ where $a\in \set{\sigma, \rho, \omega}$. \label{item:localConstraints1}}}
		\item \textcolor{red}{\sout{$\FA_{u'}(v)=[k]_\sigma$, where $k\in \set{0,1}$, if and only if $\FA_{u'}(v_0)=[k_0]_\sigma$ and  $\FA_{u'}(v_1)=[k_1]_\sigma$ where $k_0 + k_1=k$.}}
			\item  \textcolor{red}{\sout{$\FA_{u'}(v)=[k]_\rho$, where $k\in \set{0,1,2}$ if and only if $\FA_{u'}(v_0)=[k_0]_\rho$ and  $\FA_{u'}(v_1)=[k_1]_\rho$, where $\min \{ 2, k_0 + k_1\} =k$.  \label{item:localConstraints5}}}
	}
		\item $\FA_{u'}(v)=\sigma_I$ if and only if $\FA_{u'}(v_0){=}\FA_{u'}(v_1){=}\sigma_I$. \label{item:localConstraints1}
		\item  $\FA_{u'}(v)=[0]_\sigma$ if and only if $\FA_{u'}(v_0)= \FA_{u'}(v_1)=[0]_\sigma$.
		\item  $\FA_{u'}(v)=[1]_\sigma$ if and only if ($\FA_{u'}(v_0)=[1]_\sigma$ and  $\FA_{u'}(v_1)=[1]_\sigma$) or ($\FA_{u'}(v_0)=[1]_\sigma$ and  $\FA_{u'}(v_1)=[0]_\sigma$) or ( $\FA_{u'}(v_0)=[0]_\sigma$ and  $\FA_{u'}(v_1)=[1]_\sigma$).
		\item $\FA_{u'}(v)=[k]_\omega$, where $k\in \set{0,1}$, if and only if $\FA_{u'}(v_0)=[k_0]_\omega$ and  $\FA_{u'}(v_1)=[k_1]_\omega$ where $k_0 + k_1=k$.
		\item $\FA_{u'}(v)=[0]_\rho$ if and only if $\FA_{u'}(v_0)= \FA_{u'}(v_1)=[0]_\rho$.
		\item $\FA_{u'}(v)=[1]_\rho$ if and only if ($\FA_{u'}(v_0)=[0]_\rho$ and  $\FA_{u'}(v_1)=[1]_\rho$) or ($\FA_{u'}(v_0)=[1]_\rho$ and  $\FA_{u'}(v_1)=[0]_\omega$).
		\item $\FA_{u'}(v)=[2]_\rho$ if and only if ($\FA_{u'}(v_0)= \FA_{u'}(v_1)=[1]_\rho$) or ($\FA_{u'}(v_0)=[2]_\rho$ and  $\FA_{u'}(v_1)=[0]_\omega$) or ($\FA_{u'}(v_0)=[0]_\omega$ and  $\FA_{u'}(v_1)=[2]_\rho$). \label{item:localConstraintsLast}
	\end{enumerate}
	In other words, $\kappa_{u'}: F^{\jointreeMapFunction'(u')}\rightarrow \set{0,1}$, where $\kappa(\FA_{u'})=1$ if and only if $\FA_{u'}$ meets the constraints~\eqref{item:localConstraints1}-\eqref{item:localConstraintsLast}.	Hence, for all $v\in \jointreeMapFunction'(u)$, it holds that $\FA_{u'}(v)$ is determined by $\FA_{u'}(v_0)$ and $\FA_{u'}(v_1)$.\eat{,and $\nodes_{\sigma,v}(\FA_u)$.} So, while $|\jointreeMapFunction'(u')|=|\jointreeMapFunction'(u_0)|+|\jointreeMapFunction'(u_1)|+|\jointreeMapFunction'(u)|=3|\jointreeMapFunction'(u)|=3|\jointreeMapFunction(u)|$, the label of every vertex $v\in \jointreeMapFunction(u)$ is uniquely determined by vertices $v_0\in \jointreeMapFunction'(u_0)$ and $v_1\in \jointreeMapFunction'(u_1)$. Hence, we have that $\efftw(\T',\jointreeMapFunction')\leq 2w$. 
	
	Setting the neighbors of $v_0$ and $v_1$ as in~\eqref{eq:DBJTProofBody1}, for all vertices $v\in \nodes(G)$, can be done in one bottom-up traversal of the TD, and takes time $O(nw)$. To maintain the ``niceness'' of the resulting disjoint-branch TD, we need to add $|\jointreeMapFunction(u)|-1$ introduce nodes, and $2|\jointreeMapFunction(u)|-1$ forget nodes (see Algorithm \algname{TransformToDBJT} in Figure~\ref{alg:TDToDBJT}, and Figure~\ref{fig:JTToDBJTIllustrationNew}). The set of vertices in these new nodes are, by construction, strict subsets of the vertices $\jointreeMapFunction'(u')=\jointreeMapFunction'(u)\cup \bigcup_{v\in \jointreeMapFunction'(u)}\set{v_0,v_1}$ described above (see Figure~\ref{fig:JTToDBJTIllustrationNew}). Therefore, the local constraints associated with these new nodes are simply a projection of the local constraints $\kappa_{u'}: F^{\jointreeMapFunction'(u')}\rightarrow \set{0,1}$.

	Since $|\jointreeMapFunction(u)|\leq w$, then processing a single node takes time $O(w^2)$. 
	Since the TD contains $O(n)$ nodes, then the overall runtime of the algorithm is in $O(nw)+O(nw^2)=O(nw^2)$.
	From the above construction, and due to~\eqref{eq:DBJTProofBody2}, we get that the resulting structure possesses property (2) of the lemma (see also illustration in Figure~\ref{fig:DBJTAfterTransformation}). 
\end{prf}

\renewcommand{\algorithmicrequire}{\textbf{Input:}}
\renewcommand{\algorithmicensure}{\textbf{Output:}}
\begin{algserieswide}{t}{Transforms a nice TD $(\T,\jointreeMapFunction)$  to a DBJT $(\T',\jointreeMapFunction')$ representing the same neighborhoods for all $v\in \nodes(G)$.  \label{alg:TDToDBJT}}
	\begin{insidealgwide}{TransformToDBJT}{$(\T,\jointreeMapFunction)$,$u\in \nodes(\T)$}
		\REQUIRE{$(\T,\jointreeMapFunction)$: a nice TD \\$\qquad\quad \, \, \, \, u$: the current node being processed}
		\ENSURE{A Disjoint-Branch Data Structure $(\T',\jointreeMapFunction')$ corresponding to $(\T,\jointreeMapFunction)$.}
		\IF{$u$ is a leaf node}
			\RETURN
		\ENDIF
		\STATE $p \gets \parent(u)$
		\IF{$p= \texttt{nil}$} 
			\STATE $\branch(u)\gets \varepsilon$  \COMMENT{$u$ is the root node}
			\eat{
			\IF{$u$ is a join node}
				\STATE \algname{TransformToDBJT}$((\T,\jointreeMapFunction), \leftchild(u))$ 
				\STATE \algname{TransformToDBJT}$((\T,\jointreeMapFunction), \rightchild(u))$ 
			\ELSE
				\STATE \algname{TransformToDBJT}$((\T,\jointreeMapFunction), \child(u))$ 
			\ENDIF
			\RETURN}
		\ELSIF{$p$ is a join node}
			\IF{$u=\leftchild(p)$}
				\STATE $\branch(u)\gets \branch(p){\circ}0$
			\ELSE
				\STATE $\branch(u)\gets \branch(p){\circ}1$
			\ENDIF
		\ELSE
			\STATE $\branch(u)\gets \branch(p)$
		\ENDIF
		\STATE $\jointreeMapFunction'(u)\gets \emptyset$ \COMMENT{Init $\jointreeMapFunction'(u)$}
		\FORALL{$v\in \jointreeMapFunction(u)$} \label{line:item1Bijection1}
				\STATE $\jointreeMapFunction'(u)\gets \jointreeMapFunction'(u)\cup \set{v{\circ} \branch(u)}$ \label{line:item1Bijection2}
		\ENDFOR
		\IF{$u$ is a join node}
			\STATE $u_0{\gets} \leftchild(u)$, $u_1{\gets} \rightchild(u)$
			\STATE \algname{TransformToDBJT}($(\T,\jointreeMapFunction')$,$u_0$) \label{line:recCall1}
			\STATE \algname{TransformToDBJT}($(\T,\jointreeMapFunction')$,$u_1$)  \label{line:recCall2}
			\STATE Define $u''$ to be a node where $\jointreeMapFunction'(u'')\eqdef \jointreeMapFunction'(u_0) \cup\jointreeMapFunction'(u_1)$
			\STATE Define $u'$ to be a node where $\jointreeMapFunction'(u')\eqdef \jointreeMapFunction'(u) \cup \jointreeMapFunction'(u'')$
			\STATE set $\parent(u_0)\gets u''$, $\parent(u_1)\gets u''$
			\STATE Add $|\jointreeMapFunction(u)|-1$ introduce nodes between $u'$ and $u''$ \COMMENT{e.g., see Figure~\ref{fig:DBJTAfterTransformation}}
			\STATE Add $2|\jointreeMapFunction(u)|-1$ forget nodes between $u'$ and $u$.
			\STATE For all nodes $t$ on the path between $u$ and $u''$, set $\branch(t)\gets \branch(u)$
			\FORALL{$v\in  \jointreeMapFunction'(u)$}
				\STATE $N_{G_{u'}}(v_0) \eqdef N^0_{G_u}(v)$ \COMMENT{see~\eqref{eq:distributeNbrsBody} and~\eqref{eq:DBJTProofBody1}}
				\STATE $N_{G_{u'}}(v_1) \eqdef N^1_{G_u}(v)$
			\ENDFOR
		\ELSE
			\STATE \algname{TransformToDBJT}($(\T,\jointreeMapFunction')$,$\child(u)$)
		\ENDIF
	\end{insidealgwide}
\end{algserieswide}

\def\buildNiceDBTD{
	If graph $G$ has a DBTD of width $k$, then it has a nice DBTD of the same width which has at most $3kn$ nodes and can be constructed in time $O(k^2n)$.
}

\begin{lemma}
	\label{lem:buildDBTD}
	\buildNiceDBTD
\end{lemma}
\begin{proof}
	Let $(\T,\jointreeMapFunction)$ be a DBTD of $G$ rooted at node $r$. We modify $(\T,\jointreeMapFunction)$ in order for it to have the desired properties of Definition~\ref{def:niceDBTD}. Throughout, we maintain the invariant that $(\T,\jointreeMapFunction)$ is a DBTD of width $k$.
	First, to every leaf $\ell$ of $\T$ we add, as a single child, another leaf node $\ell'$ and set $\jointreeMapFunction(\ell')=\emptyset$. At this point $(\T,\jointreeMapFunction)$ is a DBTD of width $k$ that satisfies the first property.
	
	As long as $(\T,\jointreeMapFunction)$ has a node $u\in \nodes(\T)$ with more than two children, we do the following. Let $\set{u_1,\dots,u_p}$ be the children of $u$. We make a new node $u'$ and configure it as follows. We set $\jointreeMapFunction(u')\eqdef \jointreeMapFunction(u)\cap \left(\bigcup_{i=2}^p \jointreeMapFunction(u_i)\right)$, and we make $u_2,\dots,u_p$ children of $u'$ (i.e., instead of $u$). Finally, we set $u'$ as a child of $u$. Since $(\T,\jointreeMapFunction)$ is a DBTD, then $\jointreeMapFunction(u_i)\cap \jointreeMapFunction(u_j)=\emptyset$ for every pair of nodes $u_i,u_j \in \set{u_1,\dots,u_p}$. In particular, $\jointreeMapFunction(u_1)\cap \jointreeMapFunction(u')=\emptyset$. Therefore, this operation maintains the property that $(\T,\jointreeMapFunction)$ is a DBTD of width $k$. Furthermore, now $u$ has two children namely $u_1$ and $u'$, and the problem has shifted from node $u$ to node $u'$ whose degree is one less that the degree of $u$ before this step. Repeating this procedure for every node in $\nodes(\T)$ will result in the addition of at most $O(n)$ nodes, and will take time $O(nk)$.
	
	We now transform the DBTD so that for every node $u\in \nodes(\T)$ with two children $u_0$ and $u_1$, it holds that $\jointreeMapFunction(u)=\jointreeMapFunction(u_0)\cup \jointreeMapFunction(u_1)$ and $\jointreeMapFunction(u_0)\cap \jointreeMapFunction(u_1)=\emptyset$. Let $u\in \nodes(\T)$ where $u$ has two children $u_0$ and $u_1$. Since  $(\T,\jointreeMapFunction)$ is a DBTD, then $\jointreeMapFunction(u_0)\cap \jointreeMapFunction(u_1)=\emptyset$. If $\jointreeMapFunction(u) \neq \jointreeMapFunction(u_0)\cup \jointreeMapFunction(u_1)$, we do the following. 
	We define two new nodes $y_0$ and $y_1$ and configure them as follows: $\jointreeMapFunction(y_0)\eqdef \jointreeMapFunction(u) \cap \jointreeMapFunction(u_0)$ and $\jointreeMapFunction(y_1)\eqdef \jointreeMapFunction(u) \cap \jointreeMapFunction(u_1)$. We define a new node $y$ and set $\jointreeMapFunction(y)\eqdef \jointreeMapFunction(y_0)\cup \jointreeMapFunction(y_1)$. We set $y_0$ and $y_1$ as children of $y$, $u_0$ as a child of $y_0$ (i.e., instead of $u$), and $u_1$ as a child of $y_1$ (i.e., instead of $u$). Finally, we set $y$ as the unique child of $u$. If $\jointreeMapFunction(u)=\jointreeMapFunction(y)$ or  $\jointreeMapFunction(y_0)=\jointreeMapFunction(u_0)$ or $\jointreeMapFunction(y_1)=\jointreeMapFunction(u_1)$, we remove one of the redundant nodes. See Figure~\ref{fig:illustrationToNiceDBTD} for an illustration of this step.
	At the completion of this step the number of nodes that violate the fourth property of Definition~\ref{def:niceDBTD} is reduced by one. Hence, after at most $2n$ iterations of this step, each taking time $O(k)$, we get a DBTD that satisfies the first and fourth properties of Definition~\ref{def:niceDBTD}. This step adds at most $3n$ nodes to the DBTD.
	
	Now, we transform $(\T,\jointreeMapFunction)$ so that it satisfies the second and third properties of Definition~\ref{def:niceDBTD}. Let $u\in \nodes(\T)$ with a single child node $u_1$, such that $\jointreeMapFunction(u)\subset \jointreeMapFunction(u_1)$, and $|\jointreeMapFunction(u_1)\setminus \jointreeMapFunction(u)| > 1$. We introduce a new vertex $u'$ as a child of $u$ and make $u'$ the parent of $u_1$. Let $v\in \jointreeMapFunction(u_1)\setminus \jointreeMapFunction(u)$. We set $\jointreeMapFunction(u')\eqdef \jointreeMapFunction(u)\cup \set{v}$. Following this step, $u$ is a forget node that satisfies the third property of Definition~\ref{def:niceDBTD}, but $u'$ may violate it. However, the difference $|\jointreeMapFunction(u_1)\setminus \jointreeMapFunction(u')|=|\jointreeMapFunction(u_1)\setminus \jointreeMapFunction(u)|-1$. Hence, after adding at most $k$ nodes, all the nodes on the path between $u$ and $u_1$ are forget nodes that satisfy the third property of Definition~\ref{def:niceDBTD}.
	
	Finally, let $u\in \nodes(\T)$ with child node $u_1$ where $\jointreeMapFunction(u_1) \subset \jointreeMapFunction(u)$ and $|\jointreeMapFunction(u)\setminus \jointreeMapFunction(u_1)|>1$. We introduce a new node $u'$ as the single child of $u$ and make $u'$ the parent of $u_1$. Let $v\in \jointreeMapFunction(u)\setminus \jointreeMapFunction(u_1)$. We set $\jointreeMapFunction(u')\eqdef \jointreeMapFunction(u)\setminus \set{v}$. Following this step, $u$ is an introduce node that satisfies the second property of Definition~\ref{def:niceDBTD}, but $u'$ may violate it. By adding at most $k$ nodes, all the nodes on the path between $u$ and $u_1$ will be introduce nodes that satisfy the second property of Definition~\ref{def:niceDBTD}.
	
	The number of nodes in the new DBTD is at most $3kn$, which is also the total number of steps required to transform $(\T, \jointreeMapFunction)$ to a nice DBTD. Each step takes time $O(k)$ and hence, we can transform a DBTD to a nice DBTD in time $O(k^2n)$.
	
\end{proof}

\begin{figure}
	\centering
	\begin{subfigure}[t]{0.48\textwidth} 
		\begin{tikzpicture}[node distance=0.4cm,every node/.style={scale=0.475},x=1.28cm, y=1cm,font=\footnotesize]
			\tikzset{vertex/.style = {draw, circle, font=\fontsize{11}{12}}}
			\tikzset{doublevertex/.style = {draw, double, circle, font=\fontsize{11}{12}}}
			
			\node[vertex, label=left:$u$] (u) at (0,1) {$abcdg$};
			\node[vertex, label=left:$u_0$] (u0) at (-0.5,0) {$abe$};
			\node[vertex, label=left:$u_1$] (u1) at (0.5,0) {$cdf$};
			
			\draw[-{Stealth[length=1mm]}] (u) -- (u0);
			\draw[-{Stealth[length=1mm]}] (u) -- (u1);
		\end{tikzpicture}
		\caption{Node $u$ in DBTD.}
		\label{fig:illustrationToNiceDBTDBefore}
	\end{subfigure}
	\hfill
	\begin{subfigure}[t]{0.48\textwidth} 
		\begin{tikzpicture}[node distance=0.4cm,every node/.style={scale=0.475},x=1.28cm, y=1cm,font=\footnotesize]
			\tikzset{vertex/.style = {draw, circle, font=\fontsize{11}{12}}}
			\tikzset{doublevertex/.style = {draw, double, circle, font=\fontsize{11}{12}}}
			
			\node[vertex, label=left:$u$] (u) at (0,3) {$abcdg$};
			\node[vertex, label=left:$y$] (y) at (0,2) {$abcd$};
			
			\node[vertex, label=left:$y_0$] (y0) at (-0.5,1) {$ab$};
			\node[vertex, label=left:$y_1$] (y1) at (0.5,1) {$cd$};
			\node[vertex, label=left:$u_0$] (u0) at (-0.5,0) {$abe$};
			\node[vertex, label=left:$u_1$] (u1) at (0.5,0) {$cdf$};
			
			\draw[-{Stealth[length=1mm]}] (u) -- (y);
			\draw[-{Stealth[length=1mm]}] (y) -- (y0);
			\draw[-{Stealth[length=1mm]}] (y) -- (y1);
			\draw[-{Stealth[length=1mm]}] (y0) -- (u0);
			\draw[-{Stealth[length=1mm]}] (y1) -- (u1);
		\end{tikzpicture}
		\caption{After adding the disjoint join node.}
		\label{fig:illustrationToNiceDBTDAfter}
	\end{subfigure}
	\caption{Illustration for the process of transforming a DBTD to a nice DBTD.}
	\label{fig:illustrationToNiceDBTD}
\end{figure}

\subsection{Preprocessing Phase: Dynamic Programming over Nice-Disjoint-Branch TDs}
\label{sec:preprocessingDP}
We show how to compute $\M_r(\FA_r)$ for every $\FA_r \in \FAs{r}$ and every $r\in \nodes(\T)$ by dynamic programming over the nice disjoint branch TD $(\T,\jointreeMapFunction)$. 
	The preprocessing phase receives as input a nice disjoint-branch TD $(\T,\jointreeMapFunction)$, and a set of local constraints $\set{\kappa_u: u\in \nodes(\T)}$ defined in items~\eqref{item:localConstraints1}-\eqref{item:localConstraintsLast} in the proof of Lemma~\ref{lem:DBJTLemma} in Section~\ref{sec:buildDBJTAppendix}. Specifically, $\kappa_u: F^{\jointreeMapFunction(u)}\rightarrow \set{0,1}$ assigns zero to all assignments $\FA_u:\jointreeMapFunction(u) \rightarrow F$ that do not meet the constraints generated during the transformation to a nice disjoint-branch TD.

Let $A\subset \jointreeMapFunction(u)$, and $B=\jointreeMapFunction(u){\setminus}A$. Notation-wise, for an assignment $\FA_u\in \FAs{u}$, we sometimes write $\FA_u=(\FA_A,\FA_B)$ where $\FA_A$ and $\FA_B$ are the restriction of assignment $\FA_u$ to the disjoint vertex-sets $A$ and $B$ respectively, where $A{\cup}B=\jointreeMapFunction(u)$.

Let $u\in \nodes(\T)$, $\FA_u:\jointreeMapFunction(u)\rightarrow F$. We recall that
\begin{align}
	\label{eq:nodecategoriesAssignment}
	\nodes_a(\FA_u)\eqdef \set{w\in \jointreeMapFunction(u): \FA_u(w)=a} && \text{ for every }a\in F\\
	\nodes_\sigma(\FA_u)\eqdef \set{w\in \jointreeMapFunction(u): \FA_u(w)\in F_\sigma} \\
	\nodes_\omega(\FA_u)\eqdef \set{w\in \jointreeMapFunction(u): \FA_u(w)\in F_\omega}
\end{align}
Now, let $v\in \jointreeMapFunction(u)$. For every $a\in F$, we denote by $k_{v,a}(\FA_u)\eqdef |N(v)\cap \nodes_a(\FA_u)|$. Also, $k_{v,\sigma}(\FA_u)\eqdef |N(v)\cap \nodes_\sigma(\FA_u)|$, and $k_{v,\omega}(\FA_u)\eqdef |N(v)\cap \nodes_\omega(\FA_u)|$.

\noindent \textbf{Leaf Node}: $r$ is a leaf node. Since $\jointreeMapFunction(r)=\emptyset$, the only possible solution for this subtree is $\emptyset$, and hence $\M_r\equiv 1$.

\noindent \textbf{Introduce node}: Let $r \in \nodes(\T)$ be an introduce node with child node $u$, and let $\jointreeMapFunction(r)=\jointreeMapFunction(u) \cup \set{v}$, where $v\notin \jointreeMapFunction(u)$. Therefore, for $\FA_r\in \FAs{r}$, we write $\FA_r=(\FA_u,c_v)$, where $\FA_u\in \FAs{u}$ and $c_v \in F$. The label $\FA_r(w)$ of every $w\in \jointreeMapFunction(r)$ puts a constraint on the labels of $w$'s neighbors in $N(w)\cap (\nodes_r{\setminus}\jointreeMapFunction(r))$. However, since $\nodes_r=\nodes_u\cup \set{v}$, then for every $w\in \jointreeMapFunction(r)$, it holds:
\begin{align*}
	N(w)\cap (\nodes_r{\setminus}\jointreeMapFunction(r))&=N(w)\cap ((\nodes_u\cup \set{v}){\setminus} (\jointreeMapFunction(u)\cup\set{v}))\\
	&=N(w)\cap (\nodes_u{\setminus}\jointreeMapFunction(u)).
\end{align*}
Therefore, $\FA_r(w)=\FA_u(w)$ for every $w\in \jointreeMapFunction(u)$.  Since $v\notin \nodes_u$, then $N(v)\cap (\nodes_r{\setminus}\jointreeMapFunction(r))=N(v)\cap (\nodes_u{\setminus}\jointreeMapFunction(u))=\emptyset$. Therefore, if any of the following assertions fail for $\FA_r(v)$, then $\M_r(\FA_r)\gets 0$.
\begin{enumerate}
	\item \label{enum:assertionsIntroduce1} $\FA_r(v) \in \set{\sigma_I, [0]_\rho, [0]_\omega, [0]_\sigma}$.
	\item If $\FA_r(v)=\sigma_I$ then $k_{v,\sigma}(\FA_r)= k_{v,[1]_\omega}(\FA_r)=0$.
	\item  \label{enum:assertionsIntroduce3} If $\FA_r(v) = [0]_\sigma$, then $k_{v,\sigma_I}=k_{v,[1]_\omega}=0$.
\end{enumerate} 
\eat{
If, in addition, $r$ is the root node of $(\T,\jointreeMapFunction)$, then also the following must hold: 
\begin{enumerate}
	\item If $\FA_r(v)=[0]_\rho$, then $k_{v,\sigma}(\FA_r)\geq 2$.
	\item If $\FA_r(v)=[0]_\omega$, then $k_{v,[0]_\sigma}(\FA_r)=1$, and $k_{v,\sigma}(\FA_r)=1$.
	\item If $\FA_r(v)=[0]_\sigma$, then $k_{v,[0]_\omega}(\FA_r)\geq 1$.
\end{enumerate}
}
Then, we compute $\M_r(\FA_r)$ as follows:
\begin{align*}
	\M_r(\FA_r)=\M_r(\FA_u,c_v)=\begin{cases}
		\M_u(\FA_u)\wedge \kappa_r(\FA_r) & \substack{\FA_r(v)  \text{ meets conditions} \\ \text{(1)-(3) above}}\\
		0 & \text{otherwise}
	\end{cases}
\end{align*}

The conjunction with $\kappa_r(\FA_r)$ accounts for the local constraints derived from the transformation to a disjoint-branch TD.

\noindent \textbf{Disjoint Join Node}:
Let $r\in \nodes(\T)$ be a disjoint join node with child nodes $u_0$ and $u_1$, where $\jointreeMapFunction(r)=\jointreeMapFunction(u_0)\cup \jointreeMapFunction(u_1)$ and $\jointreeMapFunction(u_0)\cap \jointreeMapFunction(u_1)=\emptyset$. This means that  $\nodes_{u_i}\cap \jointreeMapFunction(u_{1-i})=\emptyset$ for $i\in \set{0,1}$.
Let $\FA_r=(\FA_{u_0},\FA_{u_1})$. Let $w\in \jointreeMapFunction(r)$, and suppose wlog that $w\in \jointreeMapFunction(u_0)$. The label $\FA_r(w)$ puts a constraint on the labels of $w$'s neighbors in $N(w)\cap (\nodes_r{\setminus}\jointreeMapFunction(r))$. Now, if wlog $w\in \jointreeMapFunction(u_0)$, then by definition of disjoint join node, $w\notin \jointreeMapFunction(u_1)$, and hence $N(w)\cap (\nodes_r{\setminus}\jointreeMapFunction(r))=N(w)\cap (\nodes_{u_0}{\setminus}\jointreeMapFunction(u_0))$.
Therefore, similar to the case of the introduce node, we have that $\FA_r[\jointreeMapFunction(u_i)]=\FA_{u_i}$ for $i\in \set{0,1}$. Therefore, to compute $\M_r(\FA_r)$, we simply take the conjunction of $\M_{u_0}(\FA_{u_0})\wedge \M_{u_1}(\FA_{u_1})$. To account for the local constraints $\kappa_u$ derived from the transformation to a disjoint-branch TD, we get that:
$$\M_r(\FA_r)=\M_r(\FA_{u_0},\FA_{u_1})=\M_{u_0}(\FA_{u_0})\wedge \M_{u_1}(\FA_{u_1}) \wedge \kappa_r(\FA_r).$$
\noindent \textbf{Forget Node}:
Let $r \in \nodes(\T)$ be a forget node with child node $u$ where $\jointreeMapFunction(u)=\jointreeMapFunction(r)\cup \set{v}$. \eat{\blue{Define $F^* = F_\sigma \cup F_\rho \cup \{ [1]_\omega, [0]_{\omega}^{0},  [0]_{\omega}^{1}\}$}}
We are going to compute $\M_r(\FA_r)$ as follows.
\begin{align}
	\label{eq:newForgetNode}
	\M_r(\FA_r) =\kappa_r(\FA_r)\wedge\left[\bigvee_{c_v\in F{\setminus}\set{[0]_\omega}}\M_u(\FA_{r,c_v},c_v) \bigvee \left(\M_u(\FA^0_{r,[0]_\omega},[0]_\omega)\vee\M_u(\FA^1_{r,[0]_\omega},[0]_\omega)\right)\right]
\end{align}
The conjunction with $\kappa_r(\FA_r)$ accounts for the local constraints derived from the transformation to a disjoint-branch TD.
\eat{ \blue{Note that for $i=0,1$ we consider $\M_u(\FA_{r,[0]_{\omega}^{i}},[0]_{\omega}^{i}) = \M_u(\FA_{r,[0]_{\omega}^{i}},[0]_{\omega})$. Furthermore, when we forget a vertex in $[0]_\omega$ the bag, $\jointreeMapFunction(u)$ ,must contain exactly one neighbor in $[1]_\sigma$. In this case, we check exactly two extendable assignments for $r$ that differ only in the label of that one vertex.}}

We define $\FA_{r,c_v}:\jointreeMapFunction(r)\rightarrow F$ for every $c_v\in F{\setminus}\set{[0]_\omega}$ as a function of $\FA_r\in F^{\jointreeMapFunction(r)}$. When $c_v=[0]_\omega$, we define $\FA^0_{r,[0]_\omega}:\jointreeMapFunction(r)\rightarrow F$, and $\FA^1_{r,[0]_\omega}:\jointreeMapFunction(r)\rightarrow F$ as a function of $\FA_r\in F^{\jointreeMapFunction(r)}$ as well.
Before defining the mappings $\FA_{r,c_v}:\jointreeMapFunction(r)\rightarrow F$ for every $c_v\in F$, we observe that $u$ is the first node (in depth first order) that contains the vertex $v$. Therefore, the possible labels for $v$ in every category (i.e., $F_\sigma$, $F_\omega$, and $F_\rho$) are determined by the assignment to its neighbors in $N(v)\cap \jointreeMapFunction(u)$. Therefore, before computing $\M_r(\FA_r)$ according to~\eqref{eq:newForgetNode}, we make the following assertions to $\FA_u$. In case any one of the assertions do not hold, we set $\M_u(\FA_u)\gets 0$.
\begin{enumerate}
	\item If $\FA_u(v)=\sigma_I$ then $k_{v,\omega}(\FA_u)=k_{v,\sigma}(\FA_u)=0$.
	\item If $\FA_u(v)=[j]_\rho$ then $j + k_{v,\sigma}(\FA_u) \geq 2$.
	\item If $\FA_u(v)=[1]_\omega$ then $k_{v,\sigma}(\FA_u)=0$.
	\item If $\FA_u(v)=[0]_\omega$ then $k_{v,\sigma_I}(\FA_u)=0$ and $k_{v,\sigma}(\FA_u)=1$.
	\item If $\FA_u(v)=[1]_\sigma$ then $k_{v,\sigma_I}(\FA_u)=k_{v,[1]_{\omega}}(\FA_u)=0$.
	\item If $\FA_u(v)=[0]_\sigma$ then $k_{v,\sigma_I}(\FA_u)=k_{v,[1]_\omega}(\FA_u)=0$, and $k_{v,[0]_\omega}(\FA_u)\geq 1$.
\end{enumerate}
We now define the mappings $\FA_{r,c_v}:\jointreeMapFunction(r)\rightarrow F$ for every $c_v\in F$, as a function of $\FA_r\in F^{\jointreeMapFunction(r)}$.
\begin{align}
	\FA_{r,\sigma_I}(w)\eqdef \begin{cases}
		\FA_r(w) & w\notin N(v) \\
		[\max\set{0,(j-1)}]_\rho & \FA_r(w)=[j]_\rho \\
		\bot & \text{ otherwise }
	\end{cases}
\end{align}

\begin{align}
	\FA_{r,[0]_\sigma}(w)\eqdef \begin{cases}
		\FA_r(w) & \substack{w\notin N(v)\text{ or }\\\FA_r(w)\in \set{[0]_\sigma,[1]_\sigma}} \\
		[\max\set{0,(j-1)}]_\rho & \FA_r(w)=[j]_\rho \\
	[0]_{\omega}&  \FA_r(w)=[1]_\omega\\
		\bot & \text{ otherwise }
	\end{cases}
\end{align}

\begin{align}
	\FA_{r,[1]_\sigma}(w)\eqdef \begin{cases}
		\FA_r(w) & \substack{w\notin N(v)\text{ or }\\\FA_r(w)\in \set{[0]_\sigma,[1]_\sigma}} \\
		[\max\set{0,(j-1)}]_\rho & \FA_r(w)=[j]_\rho \\
		[0]_\omega & \FA_r(w)=[1]_\omega \\
	\eat{	\textcolor{blue}{[1]_\omega} & \textcolor{blue}{\FA_r(w)=[0]_\omega} \\}
		\bot & \text{ otherwise }
	\end{cases}
\end{align}

\begin{align}
	\FA_{r,[1]_\omega}(w)\eqdef \begin{cases}
		\FA_r(w) & w\notin N(v)\text{ or }\FA_r(w)\notin F_\sigma \\
		\bot & \text{ otherwise }
	\end{cases}
\end{align}

\begin{align}
	\FA_{r,[0]_{\omega}}^0(w)\eqdef \begin{cases}
		\FA_r(w) & w\notin N(v)\text{ or }\FA_r(w)\notin F_\sigma \\
		\bot & \FA_r(w)=[0]_\sigma \\
		 [0]_\sigma		& \FA_r(w)=[1]_\sigma \\
		\bot & \text{ otherwise }
	\end{cases}
\end{align}
\begin{align}
	\FA_{r,[0]_{\omega}}^1(w)\eqdef \begin{cases}
		\FA_r(w) & w\notin N(v)\text{ or }\FA_r(w)\notin F_\sigma \\
		\bot & \FA_r(w)=[0]_\sigma \\
		[1]_\sigma	& \FA_r(w)=[1]_\sigma \\
		\bot & \text{ otherwise }
	\end{cases}
\end{align}

\eat{
\noindent \textbf{Root Node}: The root node $r$ is a forget node where $\jointreeMapFunction(r)=\emptyset$, its only child $u$ contains a single vertex $\jointreeMapFunction(u)=\set{v}$, and $G_r=G_u=G$. Hence, $\M_r$ is computed as in~\eqref{eq:newForgetNode}.
}
\begin{replemma}{\ref{lem:DPCorrectnessLemma}}
	\DPCorrectnessLemma
\end{replemma}
\begin{proof}
	We prove by induction on the height of the nice TD. If the height of the tree is one, then the root node $r$ has a single child that is a leaf node. By Definition of nice DBTD, this means that $r$ is an introduce node, and hence $\jointreeMapFunction(r)=\set{v}$ where $v\in \nodes(G)$. By Definition~\ref{def:induceLabelssubtree}, it holds that $\M_r(\FA_r)=\M_r(c_v)=\kappa_r(c_v)$ if $c_v \in \set{\sigma_I,[0]_\rho,[0]_\omega,[0]_\sigma}$, and $0$ otherwise. 
	
	We now assume that the claim holds for all nice DBTDs whose height is at most $h \geq 1$, and prove for the case where the height of the nice DBTD is $h+1$. So, let $r$ be the root of the nice DBTD whose height is $h+1$. Since the height of each one of $r$'s children is at most $h$, then by the induction hypothesis, for every child node $u$ of $r$, and every labeling $\FA_u:\jointreeMapFunction(u) \rightarrow F$, it holds that  $\M_u(\FA_u)=1$ if and only if $\kappa_u(\FA_u)=1$ and there exists a subset $D_u \subseteq V_u$ that is consistent with $\FA_u$ according to Definition~\ref{def:induceLabelssubtree}.

	If $r$ is an introduce node then $\jointreeMapFunction(r)=\jointreeMapFunction(u)\cup \set{v}$ where $u$ is the single child of $r$. 
	Therefore, we express $\FA_r=(\FA_u,c_v)$ where $c_v\in F$ is the assignment to vertex $v$ in $\FA_r$. By Definition~\ref{def:induceLabelssubtree}, vertex $v$ can only receive a label in $\set{\sigma_I,[0]_\sigma,[0]_\omega,[0]_\rho}$. If $\FA_r(v)=\sigma_I$, then by Definition~\ref{def:induceLabelssubtree} (item 5), it must hold that $k_{v,\sigma}=k_{v,[1]_\omega}=0$. If $\FA_r(v)\in \set{[0]_\sigma, [0]_\omega}$, then by Definition~\ref{def:induceLabelssubtree} (item 5) $v$ cannot have any neighbors assigned the label $\sigma_I$, and hence $k_{v,\sigma_I}=0$. If both of these conditions hold, then by the induction hypothesis, we have that $\M_r(\FA_r)=\M_u(\FA_u)\wedge \kappa_r(\FA_r)$.
	
	If $r$ is a disjoint join node, then by the induction hypothesis, and the fact that the label $\FA_r(w)$ places a constraint on the labels of $w$'s neighbors in $N(w)\cap (\nodes_r{\setminus}\jointreeMapFunction(r))$, then we have that $\M_r(\FA_r)=\M_r(\FA_{u_0},\FA_{u_1})=\M_{u_0}(\FA_{u_0})\wedge \M_{u_1}(\FA_{u_1})\wedge \kappa_r(\FA_r)$ where $u_0$ and $u_1$ are the children of $r$ in $\T$.
	
	The case when $r$ is a forget node follows the same reasoning, except we need to account for all possible labels that can be assigned to vertex $v\in \jointreeMapFunction(u){\setminus}\jointreeMapFunction(r)$. 
	
	For every $v\in \nodes(G)$, let $F_v\subseteq F$ denote the domain of $v$ (i.e., the set of labels it can receive). The processing of a single node takes time $O(ws^w)$, where $s=\max_{v\in \nodes(G)}|F_v|$, and $w$ is the width of the nice DBTD. Therefore, the total runtime is $O(nws^w)$
\end{proof}

%% file: AppendixNewEnumerationProofs.tex
\section{Missing Details from Section~\ref{sec:enuemration}}
\label{sec:AppendixEnumeration}
\begin{replemma}{\ref{lem:BisInjective}}
	\lemBisInjective
\end{replemma}
\begin{proof}
	Since $B(v)$ is the node, closest to the root, that contains $v$, then by definition $v\notin \parent(B(v))$. Since $(\T,\jointreeMapFunction)$ is a nice TD, the node $\parent(B(v))$ can be neither a join node, nor an introduce node because in both of these cases $\jointreeMapFunction(\parent(B(v)))\supseteq \jointreeMapFunction(B(v))$. Consequently, $\parent(B(v))$ must be  a forget node, and hence $\jointreeMapFunction(B(v)){\setminus}\set{v}= \jointreeMapFunction(\parent(B(v)))$. This means that for every $w\in B(v){\setminus}\set{v}$, it holds that $B(w)\leq \parent(B(v))<B(v)$, and thus the only vertex that is mapped to $B(v)$ is $v$.
\end{proof}
\eat{
\begin{repproposition}{\ref{prop:propertyDBDT}}
	\propertyDBDT
\end{repproposition}
}
\begin{proposition}
	\label{prop:propertyDBDT}
	\propertyDBDT
\end{proposition}
\begin{proof}
	Suppose that neither of the conditions hold. Let $u$ be the least common ancestor of $u_1$ and $u_2$ in $\T_r$. By the assumption, $u\notin\set{u_1,u_2}$, and contains two distinct child nodes $u_1'$ and $u_2'$ such that $u_1\in \nodes(\T_{u_1'})$ and $u_2\in \nodes(\T_{u_2'})$. By the running intersection property, $v\in \jointreeMapFunction(u)\cap \jointreeMapFunction(u_1')\cap \jointreeMapFunction(u_2')$, contradicting the fact that $(\T_r,\jointreeMapFunction)$ is disjoint branch.
\end{proof}

\begin{replemma}{\ref{lem:useDBTD}}
	\lemuseDBTD
\end{replemma}
\begin{proof}	
	Let $v_k\in N(v_j)\cap \set{v_{i+1},\dots,v_n}$. By Lemma~\ref{lem:BisInjective}, we have that $v_j \prec_Q v_i$ for every $v_j \in B(v_i){\setminus}\set{v_i}$, and hence $B(v_i) \subseteq \set{v_1,\dots,v_i}$. Therefore, $v_k \notin B(v_i)$,
	and there exists a node $u\in \nodes(\T)$, distinct from $B(v_i)$, such that $v_j,v_k\in \jointreeMapFunction(u)$. 
	Since $(\T,\jointreeMapFunction)$ is a DBTD where $v_j\in \jointreeMapFunction(u)\cap \jointreeMapFunction(B(v_i))$, then by Proposition~\ref{prop:propertyDBDT} either $u\in \nodes(\T_{B(v_i)})$ or $B(v_i)\in \T_{u}$. Since $k\geq i+1$ and $v_k\in \jointreeMapFunction(u)$, the latter is impossible. Therefore, $u\in \nodes(\T_{B(v_i)})$, and hence $v_k \in \jointreeMapFunction(u)\subseteq \nodes_{B(v_i)}{\setminus}B(v_i)$.
	
	Let $v_k\in \nodes_{B(v_i)}{\setminus}B(v_i)$. We claim that $k \geq i+1$. If not, then since $v_k\notin B(v_i)$, then $k\neq i$, and hence $v_k \in V_{i-1}$. By the running intersection property, if $v_k \in V_{i-1}\cap \nodes_{B(v_i)}$, then $v_k\in B(v_i)$, which is a contradiction. Therefore, $v_k\notin V_i$, and $v_k\in \set{v_{i+1},\dots,v_n}$.
\end{proof}

\begin{reptheorem}{\ref{thm:mainCorrectness}}
	\thmCorrectness
\end{reptheorem}
\begin{proof}
	We first show that \algname{IsExtendable} runs in time $O(w)$.
	The construction of $\FA_i$ in line~\ref{line:buildProjections} takes time $O(|B(v_i)|)=O(w)$ because, by definition, $|B(v_i)|\leq w+1$. Since $\M_{B(v_i)}$ is represented as a trie, then by the properties of this data structure described in Section~\ref{sec:compactRepresentation}, querying for $\M_{B(v_i)}(\FA_i)$ in line~\ref{line:query} takes time $O(|B(v_i)|)=O(w)$.
	
	Suppose that $\theta^i:V_i \rightarrow F$ is extendable according to Definition~\ref{def:consistentExtendable}.\eat{ In particular, this means that $\theta^i(v_i)\neq \bot$, and hence \algname{IsExtendable} will not return with \texttt{false} in line~\ref{line:returnFalse}.
} From Lemma~\ref{lem:BisInjective}, we have that $B(w)<B(v_i)$, and by definition, $w\prec_Q v_i$ for every $w\in B(v_i)$. Therefore, $B(v_i)\subseteq V_i$. In particular, this means that $\theta^i[B(v_i)]$ assigns a label in $F$ to every vertex in $B(v_i)$. Hence, $\theta^i[B(v_i)]: B(v_i) \rightarrow F$. 
	Since $\theta^i$ is extendable, there exists a minimal dominating set $D\in \sch(G)$ that is consistent with $\theta^i$, and induces the appropriate labels to vertices $V_i$ according to Definition~\ref{def:induceLabels}. We show that $\M_{B(v_i)}(\theta^i[B(v_i)])=1$, thereby proving the claim.	
	To that end, we show that $D_i\eqdef D\cap \nodes_{B(v_i)}$ is consistent with $\theta^i[B(v_i)]$ according to Definition~\ref{def:induceLabelssubtree}.
	By Definition~\ref{def:induceLabelssubtree}, $D_i$ is consistent with $\theta^i[B(v_i)]$ if, for every $w\in B(v_i)$, the label induced by $D_i$ on $w$ is $\theta^i[B(v_i)](w)=\theta^i(w)$. The label that $D_i$ induces on $w$ is determined by $N(w) \cap (D_i{\setminus}B(v_i))$. 
	From Lemma~\ref{lem:useDBTD}, we have that $N(w)\cap (\nodes_{B(v_i)}{\setminus}B(v_i))=N(w)\cap \set{v_{i+1},\dots,v_n}=N(w)\cap (\nodes(G){\setminus}V_i)$. In particular, $N(w) \cap (D_i{\setminus}B(v_i))=N(w)\cap (D_i{\setminus}V_i)$. Consequently, if $D$ is consistent with $\theta^i$, then  $D_i$ is consistent with $\theta^i[B(v_i)](w)$ for all $w\in B(v_i)$. Now, if $w\in \nodes_{B(v_i)}{\setminus}B(v_i)$, then by the running intersection property, $N(w)\subseteq \nodes_{B(v_i)}$. Therefore $N(w)\cap D=N(w)\cap (D\cap \nodes_{B(v_i)})=N(w)\cap D_i$. Overall, we get that if $D$ is consistent with $\theta^i$, then $D_i$ is consistent with $\theta^i[B(v_i)]$ according to Definition~\ref{def:induceLabelssubtree}.
	
	For the other direction, suppose that \algname{IsExtendable} returns \texttt{true}. We show that there is a minimal dominating set $D\in \sch(G)$ such that $D$ is consistent with $\theta^i$ according to Definition~\ref{def:induceLabels}. We first note that, by our assumption, $\theta^i$ is the result of calling \algname{IncrementLabeling} on the pair $(\theta^{i-1},c_i)$, where $\theta^{i-1}$ is extendable and $c_i \in F$. Hence, there exists a set $S\in \sch(G)$, such that $S$ is consistent with $\theta^{i-1}$. We let $S'\eqdef S\cap (\nodes(G){\setminus}\nodes_{B(v_i)})$. Since \algname{IsExtendable} returns \texttt{true}, then there exists a subset $D_i\subseteq \nodes_{B(v_i)}$ such that $D_i$ is consistent with $\theta^i[B(v_i)]$. We claim that $D\eqdef S'\cup D_i$ is a minimal dominating set that is consistent with $\theta^i$.
	To prove the claim we show that $D$ meets the conditions of Proposition~\ref{prop:DSMin}, and is consistent with $\theta^i$. To that end, we show that $D$ is consistent with $\theta^i(w)$, for every $w\in \nodes(G)$. 
	
	Let us start with a vertex $w\notin \nodes_{B(v_i)}$. Therefore, $w\in D$ if and only if $w\in S' \subseteq S$. By definition of $B(v_i)$, it holds that $N(v_i)\subseteq \nodes_{B(v_i)}$, and hence $v_i\notin N(w)$. Since only vertices in $N(v_i) \cap V_i$ are affected by the assignment of $v_i \gets c_i$ in \algname{IncrementLabeling}, then $\theta^{i-1}(w)=\theta^i(w)$. Since $S$ is consistent with $\theta^{i-1}(w)$, and since $D\cap N(w)=S\cap N(w)$, then $D$ is consistent with $\theta^i(w)$.
	
	Now, let $w\in \nodes_{B(v_i)}{\setminus}B(v_i)$. By the running intersection property, $N(w)\subseteq \nodes_{B(v_i)}$, and hence $N(w)\cap S'=\emptyset$.
	We claim that $w\notin V_i$. Otherwise, by our assumption that $w \notin B(v_i)$, then $w\neq v_i$, and hence $w\in V_{i-1}$. But then, since $w\in \nodes_{B(v_i)}$, then by the running intersection property it must hold that $w\in B(v_i)$, which is a contradiction. Therefore, $w\in \set{v_{i+1},\dots,v_n}$ and $N(w)\subseteq \nodes_{B(v_i)}$.
	Labels to the vertex-set $\set{v_{i+1},\dots,v_n}$ are not specified by $\theta^i$. Therefore, the set $D$ vacuously induces the appropriate label on $w$. Since $D_i$ is consistent with $\theta^i[B(v_i)]$, then $w$ is either dominated by $D_i\subseteq D$, has a private neighbor in $\nodes_{B(v_i)}{\setminus}D_i=\nodes_{B(v_i)}{\setminus}D$, or $N(w)\cap D_i=N(w)\cap D=\emptyset$. 
	
	Finally, assume that $w\in B(v_i)$. From Lemma~\ref{lem:useDBTD}, we have that $N(w)\cap \set{v_{i+1},\dots,v_n}= N(w)\cap (\nodes_{B(v_i)}{\setminus}B(v_i))$. By our assumption, $D_i$ is consistent with $\theta^i(w)$. Since $N(w)\cap \set{v_{i+1},\dots,v_n} \cap D = N(w)\cap D \cap (\nodes_{B(v_i)}{\setminus}B(v_i))$, then $N(w)\cap \set{v_{i+1},\dots,v_n} \cap D=N(w)\cap (D_i {\setminus}B(v_i))$.
	Therefore, if $D_i$ is consistent with $\theta^i(w)$, then $D$ is consistent with $\theta^i(w)$. Overall, we showed that $D$ is consistent with $\theta^i$. 
\end{proof}

%% file: JustificationForNotUsingTreeAutomaton.tex
\section{Comparison With \e{Courcelle}-based Approaches}
\label{sec:courcelle}
The seminal meta-theorem by Courcelle~\cite{DBLP:books/el/leeuwen90/Courcelle90,DBLP:journals/jal/ArnborgLS91} states that every decision and optimization graph problem definable in Monadic Second Order Logic (MSO) is FPT
(and even linear in the input size) when parameterized by the treewidth of the input
structure. MSO is the extension of First Order logic that allows quantification (i.e., $\forall,\exists$) over sets. In our setting, we require only existential quantification over vertex-sets.
 Bagan~\cite{DBLP:conf/csl/Bagan06}, Courcelle~\cite{DBLP:journals/dam/Courcelle09}, and Amarilli et al.~\cite{DBLP:conf/icalp/AmarilliBJM17} extended this result to enumeration problems, and showed that the delay between successive solutions to problems defined in MSO logic is fixed-parameter linear, when parameterized by the treewidth of the input structure.

In MSO logic, we can express the statement that $X \subseteq \nodes(G)$ is a minimal dominating set in an undirected graph $G$ with edge relation $E \subseteq V\times V$ as follows:
\begin{align}
	D(X)\eqdef  ~& \forall v. \left(v\in X \vee \exists y. \left(y\in X \wedge E(v,y)\right)\right) \bigwedge \label{eq:MSODS} \\
	& \forall x.\left( \forall y. \left( \left. x\in X \wedge y\in X \implies \neg E(x,y)\right.  \right) \vee  \right. \label{eq:MSOSigmaI} \\
	&~~~~~~~~~\left. \exists y. \left(\neg (y\in X) \wedge E(x,y) \wedge \forall z. \left(z\in X \wedge E(z,y) \implies z=x\right)\right)\right) \label{eq:MSOSigmaOne}
\end{align}
Eq.~\eqref{eq:MSODS} represents the constraint that $X$ is a dominating set of $G$. That is, every vertex $v\in \nodes(G)$ either belongs to $X$ or has a neighbor in $X$. 
By Proposition~\ref{prop:DSMin}, if $X$ is a minimal dominating set, then for every $x\in X$ at least one of the following holds: (1) $x$ has no neighbors in $X$, expressed in~\eqref{eq:MSOSigmaI} (note that we assume that $G$ does not have self-loops), or (2) there is a vertex $y\notin X$ such that $N_G(y)\cap X=\set{x}$; this constraint is represented in~\eqref{eq:MSOSigmaOne}. 

Courcelle-based algorithms for model-checking, counting, and enumerating the assignments satisfying an MSO formula $\varphi$ have the following form. Given a TD $(\T,\jointreeMapFunction)$ of a graph $G$, it is first converted to a labeled binary tree $T'$ over a fixed alphabet that depends on the width of the TD. 
The binary tree $T'$ has the property that $G\models \varphi$ if and only if $T' \models \varphi'$, where $\varphi'$ is derived from $\varphi$ (e.g., by the algorithm in~\cite{DBLP:journals/jacm/FlumFG02}). Then, the standard MSO-To-FTA conversion~\cite{DBLP:journals/mst/ThatcherW68,DBLP:journals/jcss/Doner70} is applied to $\varphi'$ to construct a deterministic FTA that accepts the binary tree $T'$ if and only if $T'\models \varphi'$. Since simulating the deterministic FTA on $T'$ takes linear time, this proves that Courcelle-based algorithms run in time $O(n\cdot f(||\varphi||, w))$, where $n=|\nodes(G)|$, $w$ is the width of the TD $(\T,\jointreeMapFunction)$, and $||\varphi||$ is the length of $\varphi$.

The main problem with Courcelle-based algorithms is that the function $f(||\varphi||, w)$ is an iterated exponential of height $O(||\varphi||)$. Specifically, the problem stems from the fact that every quantifier alternation in the MSO formula (i.e., $\exists \forall$ and $\forall \exists$) induces a power-set construction during the generation of the FTA. Therefore, an MSO formula with $k$ quantifier alternations, will result in a deterministic FTA whose size is a $k$-times iterated exponential. Previous work has shown that this blowup is unavoidable~\cite{DBLP:conf/dagstuhl/Reinhardt01} even for model-checking over trees. According to Gottlob, Pichler and Wei~\cite{DBLP:journals/ai/GottlobPW10}, the problem is even more severe on bounded treewidth structures because the original (possibly simple) MSO-formula $\varphi$ is first transformed
into an equivalent MSO-formula $\varphi'$ over trees, which is 
more complex than the original formula $\varphi$, and in general will have additional quantifier alternations.

The ``private neighbor'' condition (Definition~\ref{def:pn}) establishes that if $X\subseteq \nodes(G)$ is a minimal dominating set, then for every $x\in X$ where $x$ has a neighbor in $X$, it holds that there exists a $y\in \nodes(G){\setminus}X$, such that $x$ is $y$'s unique neighbor in $X$. In the language of MSO (see~\eqref{eq:MSOSigmaOne}), this is represented as:
\begin{equation}
	\label{eq:MSO_PN}
	\forall x. \left( \exists y. \left( \forall z. \left(z\in X \wedge E(z,y) \implies z=x \right) \right) \right)
\end{equation}
Observe that the MSO representation of the private neighbor condition in~\eqref{eq:MSO_PN} contains 2 quantifier alternations. In the automata theoretic framework, every quantifier alternation translates to an application of the powerset construction
for the tree automata. 
Consequently, the size of a tree automata for~\eqref{eq:MSO_PN} will be at least doubly exponential (for model-checking over trees, and probably worse for bounded-treewidth graphs~\cite{DBLP:journals/ai/GottlobPW10}). For illustration, even for a modest treewidth of $w=5$, the difference between $2^{2^5}$ and $5^{2w}$ is more than three orders of magnitude, and for $w=6$ this jumps to $11$ orders of magnitude\footnote{$2^{2^5}>4\cdot 10^9$, while $5^{10}< 10^7$; $2^{2^6}>4\cdot 10^{19}$, while $5^{12}< 3\cdot 10^8$}. 
One may argue that these blowups do not occur in practice, and that simple MSO formulas encountered in practice do lend themselves to efficient algorithms. However, Langer et al.~\cite{DBLP:journals/csr/LangerRRS14} and Langer and Kneis~\cite{DBLP:journals/entcs/KneisL09} report, based on the results of Weyer~\cite{WeyerThesis} who studied the growth function of the FTA as a function of the number of nested quantifier alternations of the input MSO-formula, that the exponential blowup is materialized even for simple MSO formula such as~\eqref{eq:MSOSigmaOne}. In particular, the algorithm of Amarilli et al.~\cite{DBLP:conf/icalp/AmarilliBJM17} for enumerating the satisfying assignments to an MSO formula by constructing a d-DNNF circuit, also proceeds by translating the MSO formula $\phi$ to a deterministic tree automaton $A$ using the result of Thatcher and Wright~\cite{DBLP:journals/mst/ThatcherW68}. The size of the resulting d-DNNF circuit is $O^*(|A|)$ (cf. Theorem 7.3 in~\cite{DBLP:conf/icalp/AmarilliBJM17}). Due to the lower bounds previously discussed, the size of the FTA $A$ is an iterated exponential. Thus, the $d$-DNNF circuit constructed will have prohibitive size. It is not hard to see that using the techniques of~\cite{DBLP:conf/ecai/PipatsrisawatD10}, the disjoint-branch structure generated in the prepossessing phase of our algorithm can be converted to a $d$-DNNF whose size has a single-exponential dependence on the treewidth.

One may hope to circumvent the price incurred by Courcelle-based algorithms by avoiding the MSO-to-FTA conversion. However, Frick and Grohe~\cite{FRICK20043} prove that any general model checking algorithm for MSO has a non-elementary lower bound regardless of whether the MSO-to-FTA transition is employed.
\begin{citedtheorem}{\cite{FRICK20043}}
	\label{thm:FrickAndGrohe}
	Assume that $P\neq NP$. There is no algorithm that, given an MSO sentence $\varphi$ and a tree $T$ decides whether $T\models \varphi$ in time $f(||\varphi||)n^{O(1)}$, where $f$ is an elementary function and $n=|\nodes(T)|$.
\end{citedtheorem}
The result by Frick and Grohe~\cite{FRICK20043} proves that no efficient runtime bounds can be proved for Courcelle's Theorem, and no general algorithm can perform better than the FTA based approach~\cite{DBLP:journals/csr/LangerRRS14}.
Consequently, most known algorithms with practical runtimes directly encode the dynamic programming strategy to solve the problem on the tree decomposition. This includes, for example, the manual creation of \e{monadic Datalog Programs} derived from an MSO-formula~\cite{DBLP:journals/ai/GottlobPW10}, and other algorithms for various NP-Hard problems~\cite{DBLP:journals/siamcomp/LokshtanovMS18}.

The expressive power of MSO logic combined with the versatility of Courcelle’s Theorem has motivated the development of general algorithms for optimization and enumeration~\cite{DBLP:conf/csl/Bagan06,DBLP:conf/icalp/AmarilliBJM17}, as well as implementations of model-checking algorithms based on Courcelle's Theorem~\cite{DURAND200529,DBLP:conf/wia/KlarlundMS00,DBLP:journals/ai/GottlobPW10}. However, the non-elementary dependence on treewidth, and the complexities involved in automata constructions and translations contribute to the difficulty in analyzing and implementing efficient algorithms based on Courcelle's theorem. In contrast, specialized algorithms that operate directly on the tree decomposition have not only lead to breakthroughs for important problems~\cite{DBLP:journals/siamcomp/LokshtanovMS18,DBLP:conf/focs/CyganNPPRW11}, but are also transparent in the sense that they are amenable to exact analysis, and the establishment of optimality guarantees for important problems~\cite{DBLP:journals/siamcomp/LokshtanovMS18,DBLP:conf/soda/BasteST20,DBLP:journals/tcs/DubloisLP22}. For this reason, several textbooks on parameterized algorithms have clearly stated that approaches based on Courcelle’s Theorem are restricted to classifying \e{which} problems are FPT, and cannot be used for designing efficient algorithms~\cite{ParameterizedAlgs,Grohe99,DBLP:books/ox/Niedermeier06}. Importantly, optimizing the dependence on treewidth (i.e., obtaining a small-growing function $f(w)$) is an active area of research in parameterized complexity~\cite{DBLP:conf/icalp/BodlaenderCKN13,DBLP:conf/soda/BasteST20,DBLP:journals/tcs/DubloisLP22,DBLP:conf/focs/CyganNPPRW11} (see also Section~\ref{sec:lowerBounds} of the Appendix). 
\eat{

To summarize, the impracticality of algorithms based on Courcelle's meta-theorem have been discussed in textbooks in parameterized algorithms~\cite{ParameterizedAlgs,Grohe99}, stating that approaches based on Courcelle’s Theorem are restricted to classifying \e{which} problems are FPT, and cannot be used for designing efficient algorithms. Similar statements
can be found in the text-book by Niedermeier~\cite{DBLP:books/ox/Niedermeier06}, and in other papers~\cite{DBLP:journals/tocl/GottlobPW10,DBLP:journals/csr/LangerRRS14,Pichler2010}. 
Consequently, general algorithms that rely on the rewriting of the MSO formula $\varphi$ to a deterministic tree automaton~\cite{DBLP:conf/csl/Bagan06,DBLP:conf/icalp/AmarilliBJM17,DBLP:journals/dam/Courcelle09} suffer from the exponential blowup previously described, and cannot be expected to handle even simple MSO formulae~\cite{DBLP:journals/csr/LangerRRS14}.
Importantly, optimizing the dependence on treewidth (i.e., obtaining a small-growing function $f(w)$) is an active area of research in parameterized complexity~\cite{DBLP:conf/icalp/BodlaenderCKN13,DBLP:conf/soda/BasteST20,DBLP:journals/tcs/DubloisLP22,DBLP:conf/focs/CyganNPPRW11} (see also Section~\ref{sec:lowerBounds} of the Appendix). 
}

%% file: lowerBounds.tex
\section{Lower Bounds and More Related Work}
\label{sec:lowerBounds}
In previous work, Dublois et al.~\cite{DBLP:journals/tcs/DubloisLP22} presented a dynamic-programming algorithm that finds the minimal dominating set of maximum size that runs in time $O^*(6^{pw})$ ($O^*$ hides poly-logarithmic factors), when parameterized by the \e{pathwidth} of the input graph; they call this problem \textsc{upper dominating set}. They also showed that under the Strong Exponential Time Hypothesis (SETH), the base of the exponent (i.e., $6$) cannot be improved. 

Every path-decomposition is, by definition, also a Disjoint-Branch TD. It is well-known that if the treewidth of a graph $G$ is $w$, then the pathwidth of $G$ is in $O(w \log |\nodes(G)|)$, and that this bound is tight~\cite{BODLAENDER19981}. We show in Section~\ref{sec:convertToDBJT} and Section~\ref{sec:AppendixProofsPreproc}, that with the introduction of the appropriate constraints, the width of the disjoint branch tree decomposition needs to grow only by a factor of $2$.

Our interest in the current paper is that of enumeration of minimal dominating sets. The delay of our algorithm is fixed-parameter-linear $O(nw)$ where $w$ is the treewidth of the input graph, following a preprocessing phase that takes time $O^*(8^w)$. Central to our approach, is the idea that the algorithm enumerates label-assignments to the vertices. To ensure that every solution is printed exactly once by the algorithm without repetitions, we need to ensure that there is a one-to-one bijection between the set of assignments produced by the enumeration algorithm, which is derived from the encoding presented in Section~\ref{sec:overview}, and the set of minimal dominating sets $\mathcal{S}(G)$; the bijection was established in Lemma~\ref{lem:bijection}. 

The encoding presented in the algorithm of Dublois et al.~\cite{DBLP:journals/tcs/DubloisLP22} does not meet this criteria. While this does not pose a problem for returning the (single) largest minimal dominating set, it does pose a problem for enumeration.
In particular, consider the graph $b-a-c$. In the encoding of Dublois et al.~\cite{DBLP:journals/tcs/DubloisLP22}, the minimal dominating set $\set{a}$ can be represented by two distinct mappings: in one, $b$ is the ``private neighbor'' of $a$ (i.e., assigned label ``P'') and $c$ is not (i.e., assigned label ``O''), and in the other assignment the roles of vertices $b$ and $c$ are reversed. In our setting, the minimal dominating set $\set{a}$ induces the unique assignment where both $b$ and $c$ are assigned the label $[0]_\omega$. The question of whether the base of the exponent of the preprocessing algorithm of Section~\ref{sec:PreprocessingForEnumeration}, and Section~\ref{sec:AppendixProofsPreproc} in the Appendix, can be made lower (i.e., than $8$) remains open at this point.